\documentclass[twocolumn,amssymb,nobibnotes, prx, superscriptaddress]{revtex4-2}

\usepackage[margin=2cm]{geometry}
\usepackage{amsmath,amssymb} 
\usepackage{graphicx}
\usepackage{dcolumn}
\usepackage{bm}
\usepackage[hidelinks]{hyperref}
\usepackage{physics} 
\usepackage[table,xcdraw,dvipsnames]{xcolor}
\usepackage{float}
\usepackage[customcolors]{hf-tikz} 
\usepackage[utf8]{inputenc}
\usepackage{enumitem}
\usepackage{booktabs}

\newcommand\nd{
    {\mbox{\tiny nd}}
}
\newcommand\st{
    {\mbox{\tiny st}}
}

\usepackage{pdfrender}
\newcommand*{\boldify}[1]{%
  \textpdfrender{%
    TextRenderingMode=FillStroke,%
    LineWidth=.60pt,%
  }{#1}%
}

\newcommand{\vbra}[1]{\bra{\boldify{#1}}}
\newcommand{\vket}[1]{\ket{\boldify{#1}}}
\newcommand{\vbraket}[2]{\braket{\boldify{#1}}{\boldify{#2}}}
\newcommand{\vketbra}[2]{\ketbra{\boldify{#1}}{\boldify{#2}}}

\usepackage{dsfont}
\newcommand{\ident}{
  \mathds{1}
}

\renewcommand{\cal}[1]{\mathcal{#1}}
\DeclareMathAlphabet\mathbfcal{OMS}{cmsy}{b}{n}

\usepackage{amsthm}
\newtheorem{theorem}{Theorem}[section]
\newtheorem{corollary}{Corollary}[section]
\newtheorem{lemma}{Lemma}[section]
\newtheorem{definition}{Definition}[section]

\newcounter{protocol}[section]
\newenvironment{protocol}[1][]{
  \refstepcounter{protocol}
  \par\noindent\rule{\columnwidth}{1.4pt}  \vspace{-0.3cm}
  \par\noindent Protocol~\theprotocol: \textbf{#1} \rmfamily 
  \par\noindent\rule{\columnwidth}{1pt} \vspace{-0.5cm} }{\par\noindent\rule{\columnwidth}{1.4pt}\medskip}

\usepackage[normalem]{ulem}
\def\>{\rangle}
\def\<{\langle}
\def\Tr{\mbox{tr}}
\def\I{ \ident }
\def\hc{^{\dagger}}
\def\H {\mathcal{H}}
\def\B {\mathcal{B}}
\def\U {\mathcal{U}}
\def\V {\mathcal{V}}
\def\F {\mathcal{F}}
\def\D {\mathcal{D}}
\def\R {\mathcal{R}}
\def\S {\mathcal{S}}
\newcommand{\E}{\mathcal{E}}

\begin{document}
\title{Estimation of correlations and non-separability\\ in quantum channels via unitarity benchmarking}
\date{15 April, 2022}
\author{Matthew Girling}
\email{m.j.girling@leeds.ac.uk}
\affiliation{School of Physics and Astronomy, University of Leeds, Leeds LS2 9JT, United Kingdom}
\author{Cristina C\^{i}rstoiu}
\affiliation{Cambridge Quantum Ltd, 13-15 Hills Road, Cambridge CB2 1NL, United Kingdom} 
\author{David Jennings}
\affiliation{School of Physics and Astronomy, University of Leeds, Leeds LS2 9JT, United Kingdom}
\affiliation{Department of Physics, Imperial College London, London SW7 2AZ, United Kingdom} 

\begin{abstract}
The ability to transfer quantum information between systems is a fundamental component of quantum technologies and leads to correlations within the global quantum process.  However correlation structures in quantum channels are less studied than those in quantum states.
Motivated by recent techniques in randomized benchmarking, we develop a range of results for efficient estimation of correlations within a bipartite quantum channel. 
We introduce sub-unitarity measures that are invariant under local changes of basis, generalize the unitarity of a channel, and allow for the analysis of quantum information exchange within channels. Using these, we show that unitarity is monogamous, and provide an information-disturbance relation. We then define a notion of correlated unitarity that quantifies the correlations within a given channel. Crucially, we show that this measure is strictly bounded on the set of separable channels and therefore provides a witness of non-separability. Finally, we describe how such measures for effective noise channels can be efficiently estimated within different randomized benchmarking protocols. We find that the correlated unitarity can be estimated in a SPAM-robust manner for any separable quantum channel, and we show that a benchmarking/tomography protocol with mid-circuit resets can reliably witness non-separability for sufficiently small reset errors. The tools we develop provide information beyond that obtained via simultaneous randomized benchmarking and so could find application in the analysis of cross-talk errors in quantum devices.
\end{abstract}
\maketitle 

\section{Introduction}
Efficiently certifying and benchmarking non-classical features in quantum theory is central to the development of quantum technologies \cite{kliesch2020theory,helsen2020general,proctor2019direct,gaebler2012randomized,erhard2019characterizing,eisert2020quantum,derbyshire2021randomized}, which requires precise control and manipulation of quantum systems. High-fidelity quantum gates and circuits are essential for scalable quantum computing so it is important to benchmark the effects of physical noise on how accurately a target unitary is realized on the quantum device. For example, noise due to unwanted correlations or leakage can detrimentally affect error rate thresholds required for fault tolerant quantum computing \cite{preskill2012sufficient, nickerson2019analysing, iverson2020coherence}. Therefore, detection and quantification
of noise correlations such as cross-talk in quantum devices not only
impacts NISQ era devices \cite{preskill2018quantum} by improving circuit fidelities
and error mitigation methods, but it goes beyond
it in providing necessary tools to test physical assumptions
of quantum error correction.

Direct process tomography~\cite{d2001quantum, greenbaum2015introduction} of noisy gates and circuits faces two non-trivial obstacles: first, the complexity of full tomography is known to scale exponentially, and second, there is the problem of characterizing errors in the presence of other types of errors such as those arising from state-preparation and measurement (SPAM). To circumvent these obstacles techniques have been developed such as gate-set tomography and randomized benchmarking (RB), which allow for efficient estimation of measures that are robust against SPAM errors.

The simplest instance of an RB protocol returns an estimate of the average gate infidelity $r(\E)$ of the noisy computational gate-set (e.g.\ Clifford gates), for the effective noise channel $\E$. The average gate infidelity of this quantum channel~\cite{watrous2018theoryvec,wilde2013quantum,nielsen2002quantum,bengtsson2017geometry} can be used to bound the worst-case error rate, defined in terms of the diamond norm \cite{wallman2014randomized}, which is the relevant quantity in the context of fault-tolerant computation~\cite{aharonov2008fault, harper2019fault}:
	\begin{equation}\label{eqn:r-diamond-bound}
		\frac{d}{d+1}r(\E) \le \frac{1}{2} || id - \E ||_\diamond \le \sqrt{d(d+1)r(\E)},
	\end{equation}
	where $||id -\E||_\diamond$ is the diamond norm distance of the channel $\E$ to the identity channel~\cite{watrous2018theoryvec}.

Recent work has extended the core benchmarking toolkit, for example through higher-order moment analysis \cite{nakata2021quantum}, character benchmarking techniques \cite{helsen2019new}, the extension to benchmarking of logical qubits \cite{combes2017logical} and analogue regimes \cite{derbyshire2020randomized}.  Simultaneous randomized benchmarking \cite{gambetta2012characterization} has also been developed as a means to quantify the \emph{addressability} of a subsystem in a device and thus provide a basic assessment of the presence of cross-talk and correlation errors. 

Beyond noise analysis in quantum technologies, there are other motivations why one would like to be able to efficiently assess correlative structures within quantum channels -- for example consider a bipartite quantum channel  $\E_{AB}: \B(\H_A \otimes \H_B) \to \B(\H_{A} \otimes \H_{B})$ from a bipartite quantum system $AB$ to itself with $A$ and $B$ having equal dimension.  Correlations within the channel are required for the transfer of a quantum state prepared on the first subsystem $A$ to the second subsystem $B$: for example to transform the input pure states $|\psi\>_A \otimes |\phi\>_B$ to $|\phi\>_A \otimes |\psi\>_B$ via the SWAP unitary. This transformation is impossible under product channels of the form $\E_{AB} = \E_A \otimes \E_B$, with $\E_A$ a channel from $A$ to $A$ and $\E_B$ a channel from $B$ to $B$, and so non-product channels are clearly required. However, quantifying these channel correlations is a distinct problem from measuring the correlation--\emph{generating} abilities of a quantum channel. The SWAP unitary perfectly transfers a quantum state on $A$ to $B$, however it has zero correlation generating abilities as it sends the set of product states $\rho_A \otimes \sigma_B$ to itself. In contrast the channel that sends all quantum states on $AB$ to a Bell state is maximal in generating correlations, however it clearly transmits zero information from $A$ to $B$. Intermediate between these two extremal channels are \emph{separable} channels that are defined as a convex mixture of product channels, $\E_{AB} = \sum_i p_i \E_A^i \otimes \E_B^i$. These channels can only create classical correlations between $A$ and $B$, but it is clear they do not transfer any quantum information from $A$ to $B$.

However, connections between non-classical channel correlations and correlations within quantum states do exist. Specifically, the set of separable channels play a central role in the resource theory of  Local Operations and Shared Randomness (LOSR)~\cite{gutoski2008properties, de2014nonlocality, geller2014quantifying, gallego2017nonlocality, rosset2019characterizing, schmid2020type, schmid2020standard, wolfe2020quantifying,hsieh2020entanglement,gutoski2009properties}, for the study of non-classicality in quantum theory. It has recently been argued that this framework is the appropriate setting in which to properly analyse Bell non-locality and the self-testing of quantum states~\cite{schmid2020standard, wolfe2020quantifying}.  Therefore a non-separable quantum channel requires the consumption of state correlations, and an ability to efficiently and robustly certify non-separability in a general quantum channel $\E_{AB}$ implies the use of non-local quantum resources.

 More broadly, since process tomography is exponentially hard, one can ask what non-classical features of quantum channels \cite{gour2020dynamical, bauml2019resource} can be accessed in practice. We know that actual physical systems only probe a very small region of the set of all possible quantum states, dubbed the ``physical corner of Hilbert space'' \cite{poulin2011quantum, eisert2013entanglement}, and so a similar question for quantum channels can be addressed by drawing on recent developments in randomized benchmarking theory.
 
\subsection{Aims and outline of the paper}
Motivated by (a) benchmarking the performance of quantum computers at the level of sub-systems, and (b) certifying non-classicality in quantum physics, we have the following two aims in this work:
\begin{enumerate}
\item To quantify the degree to which a quantum channel deviates from being separable in a form that can be estimated efficiently and robustly.
\item To demonstrate an application of this approach by deriving an information-disturbance relation that can be efficiently and robustly verified.
\end{enumerate}

 Our work exploits recent techniques from randomized benchmarking theory \cite{wallman2015estimating,dirkse2019efficient, sundaresan2020reducing, kukulski2020generating, sundaresan2020reducing} that were originally introduced to provide additional information on the average gate infidelity $r(\E)$ for noise channels. The central quantity of interest is the \emph{unitarity} $u(\E)$ of a quantum channel $\E$. This is defined as
\begin{equation}\label{unitarity:haar-defn}
    u(\mathcal{E}) := \frac{d}{d-1} \int \dd \psi  \tr[\mathcal{E}(\psi - \frac{\ident}{d})^2],
\end{equation}
where the integration is with respect to the Haar measure over pure states of the $d$-dimensional input system \cite{wallman2015estimating}. The unitarity provides a measure how far a quantum channel is from being a unitary channel and crucially can be estimated in an efficient and SPAM-robust protocol. It attains its extremal values of $u(\E) = 0$ if and only if $\E$ is a completely depolarizing channel and $u(\E) = 1$ if and only if $\E$ is an isometry channel,  and can also be shown to provide a tighter bound on diamond norm measures for quantum channels. While measures like the diamond norm have clear operational significance, such as for single-shot channel discrimination, they are in general neither efficiently estimatable nor robust to SPAM-errors, in contrast to the unitarity. 

We shall show that the unitarity of a quantum channel is well-suited to aims (1) and (2) above, and suggests a route to analysing similar structural questions about bipartite quantum channels in a form that is amenable to efficient and SPAM-robust experiments. 

We will show that for a bipartite quantum system $AB$ the concept of unitarity naturally extends to a collection of 9 \emph{sub-unitarities} $u_{X\rightarrow Y} (\E)$ of a quantum channel $\E$ on $AB$. Each of these sub-unitarities gives finer information about how the channel acts on the quantum systems $A$ and $B$, and allow us to address both (1) and (2) above. However we find that only non-trivial combinations of sub-unitarities are estimatable in a SPAM-robust protocol, and so this forces us to develop methods to estimate channel correlations for aim (1).

Objective (1) turns out to be substantially more challenging than (2), and we begin in Section II with the problem of quantifying channel correlations. We first note that the unitarity of a channel can be reformulated as a variance estimate, which then motivates a correlation measure $u_c(\E_{AB})$ that parallels the covariance between two classical random variables. The construction of this correlation measure leads to the definition of the 9 sub-unitarities in Section II.A and II.D.

Then in Section II.B we show that the simplest sub-unitarities lead to a novel form of the information-disturbance relation given by,
\begin{equation}
u(\E_{X\rightarrow A}) + u(\E_{X\rightarrow B}) \le 1,
\end{equation}
for any quantum channel $\E_{X\rightarrow A}$ from an input system $X$ to an output system $A$ and an associated complementary channel $\E_{X\rightarrow B}$. In contrast to prior formulations the unitarity-based relation provides the ability to efficiently and robustly verify this fundamental relation.

In Section II.E we prove that the measure $u_c(\E_{AB})$ certifies non-classical features of a channel. More precisely, we prove that over the set of separable quantum channels (i.e.~convex mixtures of product channels) it is strictly bounded away from the global maximum, and thus provides a witness of non-separability for quantum channels.

Finally, in Section III we address the problem of efficiently estimating the correlated unitarity of effective noise channels in a benchmarking scenario. For this we follow a similar approach to simultaneous randomized benchmarking in which one employs local $2$-designs on each subsystem. This is of relevance for quantifying cross-talk errors in quantum devices. We show that for bipartite separable channels the correlated unitarity can be obtained efficiently in a SPAM-robust protocol. For more general non-separable channels, we show that for weak reset errors that this can still be estimated and within a natural model demonstrate explicitly that the protocols can witness non-separability over a substantial range of reset errors. We end by discussing the relation between our work and simultaneous randomized benchmarking and show that our protocols provide additional, independent information on cross-talk and correlative errors.

\section{Sub-unitarities for bipartite Quantum Channels}

We wish to formulate an experimentally accessible measure of correlations in a general bipartite quantum channel. Paralleling the situation with quantum states, we say that a quantum channel $\E_{AB}:\B(\H_A \otimes \H_B) \rightarrow \B(\H_A \otimes \H_B)$ from a bipartite system $AB$ to itself is \emph{uncorrelated} or alternatively a \emph{product} channel  if $\E_{AB} = \E_A \otimes \E_B$ for a channel $\E_A$ from $A$ to itself and $\E_B$ from $B$ to itself. Otherwise it is said to be a \emph{correlated} channel. We shall also consider the set of \emph{separable} channels, which take the form of a convex mixture of product channels $\E_{AB} = \sum_i p_i \E_A^i \otimes \E^i_B$. A quantum channel is said to be \emph{non-separable} if it lies outside the convex set of separable channels. The extension to channels from input systems $AB$ to potentially different output systems $A'B'$ is obvious, but to avoid over-complicating notation we primarily focus on identical input and output systems and only discuss the more general case in Section~\ref{info-dist}, where it is required. The general definition is provided in Appendix~\ref{notation}.

\subsection{Elementary sub-unitarities of a channel}

Given two classical random variables $X$ and $Y$ a simple and direct method of measuring correlations is to compute the covariance of $X$ and $Y$. This is given as $\mbox{cov}(X,Y) := \<XY\> - \<X\>\<Y\>$, where the angle brackets denote taking the expectation value of the random variable. Moreover, we have that $\mbox{cov}(X,X) = \mbox{var}(X)$, the variance of the random variable $X$, which in turn quantifies the noisiness of $X$. The relevance here is that in \cite{korzekwa2018coherifying} it was noted that the unitarity of a channel can be expressed as
\begin{equation}
u(\cal{E}) = \tr[\mathrm{var}(\cal{E})], 
\end{equation}
where $\mathrm{var}(\cal{E}) := \ev{\cal{E}(\psi)^2} - \ev{\cal{E}(\psi)}^2$ and the angle brackets denote taking the expectation of an operator-valued random variable with respect to the Haar measure. 

As the unitarity can be viewed as the ``variance'' of a quantum channel, we can ask if a form of \emph{covariance} for a quantum channel exists similar to the covariance of two random variables in classical statistics. However, while there is a clear notion of a marginal distribution for a joint probability distribution the situation is more complex for a bipartite quantum channel where the reduction to `marginal channels' depends on the structure of the initial state considered \cite{hsieh2021quantum}. Instead, here we take the basic form of covariance of two random variables as a guide and construct a unitarity-based correlation measure $u_c(\E_{AB})$ for a bipartite quantum channel with certain desirable features. 

 As we want the measure $u_c(\E_{AB})$ for quantum channels will function like $cov(X,Y)$ for classical random variables, we must define sensible channel equivalents to $\<X\>, \<Y\>$ and $\<XY\>$. In the context of RB protocols on bipartite quantum channels we shall show in Section \ref{spam-robust-protocol-section} that a natural marginal channel measure $u_{A\rightarrow A}$ emerges that parallels the classical marginal expectation $\<X\>$. This is given by the following \emph{sub-unitarity} $u_{A\rightarrow A}$ of a bipartite quantum channel.

\begin{definition}\label{defn:sub-unitarity1}
    The sub-unitarity $u_{A \rightarrow A}$ of a bipartite channel $\cal{E}_{AB}$ is defined as 
    \begin{equation}\label{eqn:sub-A}
        \begin{split}
            u_{A \rightarrow A}(\cal{E}_{AB}) &:= u(\cal{E}_A),
        \end{split}
    \end{equation}
    where $ \mathcal{E}_A (\rho) := \tr_B[\mathcal{E}_{AB}(\rho \otimes \frac{\ident_B}{d_B})]$ for any state $\rho$ of $A$.
\end{definition}
The same construction applies for the $B$ subsystem with the associated channel $\mathcal{E}_B(\rho) := \tr_A[\mathcal{E}(\frac{\ident_A}{d_A} \otimes \rho)]$ giving $u_{B \rightarrow B}(\cal{E}_{AB}) := u(\cal{E}_B)$. It is also clear that we can define two further sub-unitarities $u_{A\rightarrow B}$ and $u_{B \rightarrow A}$ that are obtained simply as
\begin{equation}
u_{A\rightarrow B} (\E_{AB}) = u_{A\rightarrow A} ( SW\!\!AP \circ \E_{AB}),
\end{equation}
and similarly for $u_{B\rightarrow A}$, where $SW\!\!AP$ is the unitary that swaps the two subsystems $A$ and $B$.

From these definitions it is clear that the sub-unitarity $u_{X \rightarrow Y} (\E_{AB})$, with $X,Y$ being subsystems, is based on the situation in which a quantum state $\rho$ is prepared on $X$ with the the maximally mixed state on the other subsystem and then evolved under the channel $\E_{AB}$. The quantity $u_{X \rightarrow Y} (\E_{AB})$ inherits the properties of unitarity and therefore measures how close this global evolution is to being an isometric mapping of the state $\rho$ on $X$ into the output system $Y$.

Moreover, the sub-unitarities $u_{A \rightarrow A}$ and $u_{B \rightarrow B}$ for the bipartite quantum channel have the property that when applied to product channels give
\begin{align}\label{sub-prod}
u_{A\rightarrow A} (\E_A \otimes \E_B) &= u(\E_A) \nonumber \\
u_{B\rightarrow B} (\E_A \otimes \E_B) &= u(\E_B).
\end{align}
These relations can therefore imply that we can define a \emph{correlated unitarity} $u_c(\E_{AB})$ measure as
\begin{equation}
u_c(\E_{AB}) := u_{AB \rightarrow AB} (\E_{AB}) - u_{A\rightarrow A}(\E_{AB})u_{B\rightarrow B}(\E_{AB})
\end{equation}
provided we can also construct a sub-unitarity $u_{AB \rightarrow AB}$ such that
\begin{equation}\label{AB-prod}
u_{AB\rightarrow AB} (\E_A \otimes \E_B) = u(\E_A) u(\E_B).
\end{equation}
The definition of $u_{AB\rightarrow AB}$ is most easily expressed in the Liouville representation, and is provided in Section~\ref{sec:sub-unitarities}, and the justification for the naturalness of these terms is provided in Section  \ref{spam-robust-protocol-section} where we will show that these arise naturally from randomized benchmarking theory. The technical reason for this is that they are the quantities that arise if one considers quadratic order expectations over Haar random states where one includes the bipartite structure explicitly. 

However, before addressing benchmarking theory, in the next sub-section we show how the above sub-unitarities lead to a statement of the information-disturbance relation that is amenable to experimental verification. 

\subsection{Unitarity formulation of information-disturbance}\label{info-dist}
The information-disturbance relation~\cite{kretschmann2008information} is a fundamental result in quantum theory and can be summarized as saying that if a quantum channel is close to being a unitary, or more generally an isometry, then the leakage of quantum information into the environment must be ``small''. This trade-off can be expressed in terms of the diamond norm distance of the channel from a unitary channel for the output system, and the diamond norm distance of the complementary channel from a completely depolarizing channel for the environment. However, such quantities can neither be estimated efficiently nor in a SPAM-robust form. In this section we provide an alternative formulation of information-disturbance that does not suffer from these weaknesses.

In the definition of sub-unitarities, we assumed that the input and output systems are identical, but the above definitions can be extended to a channel from arbitrary input and output systems. Of particular interest is when one has a channel from a single input system $X$ into a bipartite system $AB$. In this setting the sub-unitarities of the channel coincide with the unitarities of the marginal channels into $A$ and $B$ separately. For this setting we now show the following result on sub-unitarities that provides a statement of quantum incompatibility \cite{heinosaari2016invitation, kretschmann2008information}. To our knowledge the question of efficiently and SPAM-robustly testing such foundational results has not been previously considered, and so such a result opens up this possibility by formulating in terms of quantities native to randomized benchmarking protocols.

For clarity in this section we shall put subscripts on the channels to denote their input and output systems explicitly, and write $\E_{X\rightarrow Y}$ to denote a channel from $X$ into $Y$. In the context of a single input system, we have that
\begin{equation}
u_{X\rightarrow A} (\E_{X\rightarrow AB}) = u (\E_{X\rightarrow A}),
\end{equation}
with a similar expression for $u_{X\rightarrow B}$. Given the ability to estimate unitarity in randomized benchmarking protocols we therefore expect that our relation could also be verified efficiently and robustly using existing hardware. We now state and prove the unitarity-based information-disturbance relation.

\begin{theorem}[Information-Disturbance Relation]
   Let $\E_{X\rightarrow A}$ be a quantum channel from an input system $X$ to an output system $A$, and let $\E_{X\rightarrow AB}(\rho) = V\rho V^\dagger$ be any isometry, with $V^\dagger V = \I$, that provides a Stinespring dilation of $\E_{X\rightarrow A}$ via $ \E_{X\rightarrow A} = \tr_B \circ \, \E_{X\rightarrow AB}$. Then
   \begin{equation}
u(\E_{X\rightarrow A}) +u(\E_{X\rightarrow B}) \le 1,
\end{equation}
where $\E_{X\rightarrow B} = \tr_A \circ \, \E_{X\rightarrow AB}$ is the associated complementary channel to $\E_{X\rightarrow A}$ in the dilation.
\end{theorem} 
\begin{proof} Let $d$ be the dimension of the system $X$. It can be shown~\cite{cirstoiu2020robustness} that the unitarity of a channel can be expressed as
\begin{equation}
u(\E_{X\rightarrow A}) = \frac{d}{d^2-1} (d\tr [ \tilde{\E}_{X\rightarrow A} (\I/d)^2]  - \tr [ \E_{X\rightarrow A} (\I/d)^2])
\end{equation}
where $ \tilde{\E}_{X\rightarrow A}$ is any complementary channel to $\E_{X\rightarrow A}$, which we can choose to be $\E_{X\rightarrow B}$. Applying the above expression to the complementary pair $(\E_{X\rightarrow A}, \E_{X\rightarrow B})$ we then have that
\begin{equation}
u(\E_{X\rightarrow A}) +u(\E_{X\rightarrow B}) = \frac{d}{d+1} (\gamma(\rho_A) +\gamma(\rho_B)),
\end{equation}
where $\gamma(\rho) := \tr[\rho^2]$ is the purity of a quantum state, $\rho_A = \E_{X\rightarrow A}(\I/d)$, and $\rho_B = \E_{X\rightarrow B}(\I/d)$.

We can also consider $\rho_{AB} = \E_{X\rightarrow AB}(\I/d)$, for which $\rho_A$ and $\rho_B$ are the marginals. For a general bipartite quantum state $\rho_{AB}$ it can be shown \cite{man2014deformed} that
\begin{equation}
\gamma(\rho_A) + \gamma(\rho_B) \le 1 + \gamma(\rho_{AB}),
\end{equation}
and so we have that
\begin{equation}
u(\E_{X\rightarrow A}) +u(\E_{X\rightarrow B}) \le \frac{d}{d+1} (1 +\gamma(\rho_{AB})).
\end{equation}
However we can now use that the channel $\E_{X\rightarrow AB}$ is an isometry and so
\begin{align}
\gamma(\rho_{AB}) = \tr [ (V(\I/d) V^\dagger)^2] = \frac{1}{d}.
\end{align}
Substituting this into the previous inequality we obtain,
\begin{equation}
u(\E_{X\rightarrow A}) +u(\E_{X\rightarrow B}) \le 1,
\end{equation}
which completes the proof.
\end{proof}
The result provides a compact form of information-disturbance \cite{kretschmann2008information}, which in turn implies no-cloning and no-broadcasting~\cite{wootters1982single,nocloning2,nobroadcasting1,nobroadcasting2}. More precisely, we can consider leakage of quantum information from a system into its environment, which is of relevance to, for example, quantum computing in a noisy environment when one wishes to approximate a unitary channel as accurately as possible. We can consider a quantum channel $\E_{X\rightarrow AB}$ from a system $X$ to a composite system $AB$ such that $\E_{X\rightarrow A} \approx  \mathcal{U}_{X\rightarrow A}$, for some target isometry $\mathcal{U}_{X\rightarrow A}$. As the unitarity is a continuous function of the channel, we can quantify this as $u(\E_{X \rightarrow A}) = u(\mathcal{U}_{X \to A}) - \epsilon$ for some $\epsilon \ge 0$ quantifying the approximation. However the unitarity of a channel equals $1$ if and only if it is an isometry \cite{cirstoiu2020robustness,wallman2015estimating} and so the above monogamy relation implies that $u(\E_{X \rightarrow B}) \le \epsilon$. However it is easily shown (see Lemma \ref{lemma:unitarity-zero-iff-depolar} in the Appendices) that the unitarity vanishes if and only if the channel is a completely depolarizing channel. This in turn implies that the channel $\E_{X \rightarrow B}$ must be $\epsilon$-close in terms of unitarity to a completely depolarizing channel. In other words, the relation implies that the information leaking into the environment necessarily decreases to zero as the channel $\E_{X\rightarrow A}$ approaches an isometry channel. 

\subsection{The Liouville representation of quantum channels}
Consider quantum channels $\E$: $\cal{B}(\mathcal{H}_A) \to \cal{B}(\mathcal{H}_A)$, where $\mathcal{B}(\mathcal{H}_A)$ denotes the space of linear operators on the Hilbert space $\H_A$ for a $d$-dimensional quantum system $A$. We choose an orthonormal basis of operators $X_0, X_1 ,\dots ,X_{d^2-1}$ for $\B(\H_A)$ with $X_0 = \I/\sqrt{d}$ and with respect to the Hilbert Schmidt inner product  $\ev{X_\mu,X_\nu} := \tr [ X^\dagger_\mu X_\nu] = \delta_{\mu,\nu}$. In particular, this means that $X_1, \dots, X_{d^2-1}$ are all traceless operators.

We define vectorization of operators via  $|vec( |a\>\<b|) \> := |a\> \otimes |b\>$ for any computational basis states \cite{watrous2018theoryvec}. This definition can be extended by linearity to get the mapping $M \rightarrow |vec(M)\>$ for any operator $M \in \B(\H_A)$. Then for any quantum channel $\E: \B(\H_A) \rightarrow \B(\H_A)$ we define its Liouville representation $\mathcal{L}(\E)$ through the relation 
\begin{equation}\label{eqn:vec}
  \mathcal{L}(\E) | vec(M)\> = |vec(\E(M))\>, 
\end{equation}
for all $M$. To simplify things going forward, we shall adopt the notation that we denote all vectorized quantities in boldface (this is similar to how a vector is sometimes represented in boldface as $\mathbf{v}=(v_1,v_2, \dots , v_n)$), and so write $|\boldify{M}\> := |vec(M)\>$ and $\boldsymbol{\E}: = \mathcal{L}(\E)$. Using this boldface notation we can re-express equation (\ref{eqn:vec}) in the more compact form
\begin{equation}\label{bold+simple}
    \mathbfcal{E}\vket{\rho} = \vket{\cal{E}(\rho)},
\end{equation}
for any state $\rho$, and any channel $\E$. Using equation (\ref{bold+simple}) we can therefore decompose any channel in the orthonormal basis $\{X_\mu\}$ as
\begin{equation}
    \mathbfcal{E}=\sum_{\mu=0}^{d^2-1} \vketbra{\mathcal{E}(X_\mu)}{X_\mu},
\end{equation}

More explicitly, in terms of matrix components we have that
\begin{equation}
    \mathbfcal{E} =
        \bordermatrix{~ 
        & \vket{X_{0}} & \vket{X_{j}}   \cr
        \vbra{X_{0}} &  1 & \bm{0}  \cr
        \vbra{X_{i}}  & \mathbf{x} & T  \cr
        },
\end{equation}
where $\mathbfcal{E}_{00}=1$ and $\mathbfcal{E}_{0j}=\bm{0}$ follow from the fact that the channel is a completely positive trace-preserving operation. The $d^2-1$ component vector $\mathbf{x}$ corresponds to the generalized Bloch vector of $\cal{E}(\ident/d)$, which characterizes the degree to which the channel breaks unitality. The matrix block $T$ encodes the remaining features of the channel. In this notation, the unitarity of a channel is then given by the simple relation \cite{wallman2015estimating}
\begin{equation}\label{eqn:unitarity-T}
    u(\mathcal{E})=\frac{1}{d^2-1}\tr[T^\dagger T].
\end{equation}

This core form is the one we use to define sub-unitarities in the next subsection.

\subsection{Liouville decomposition of bipartite quantum channels and general sub-unitarities}\label{sec:sub-unitarities}
We can also compute Liouville representations of bipartite channels, $\E_{AB}$: $\cal{B}(\mathcal{H}_A \otimes \mathcal{H}_B) \to \cal{B}(\mathcal{H}_A \otimes \mathcal{H}_B)$, where we assume for simplicity that the input and output systems are identical. 

For subsystem $A$, we choose an orthonormal basis of operators $X_\mu = (X_0 =  \frac{1}{\sqrt{d_A}}\I_A, X_i)$, where $d_A$ is dimension of the subsystem $A$, and similarly for $B$  a basis $Y_\mu = (Y_0 = \frac{1}{\sqrt{d_B}}\I_B, Y_i)$. Together these provide a basis for the full system which is given in the Liouville representation as \footnote{Note that the basis $\vket{X\otimes Y}$ is  a tensor product basis for $(\mathcal{H}_A\otimes \mathcal{H}_A)\otimes (\mathcal{H}_B\otimes \mathcal{H}_B)$ and up to re-ordering of (second and third) Hilbert spaces the same as vectorization of the matrix $X\otimes Y$. As these basis are isomorphic, the Liouville representation will be invariant under such permutations. } 
\begin{equation}
\vket{X_{a}\otimes Y_b} := \vket{X_a}\otimes \vket{Y_b}.
\end{equation}
This in turn provides the following matrix decomposition of $\mathcal{E}_{AB}$,
\begin{equation}\label{eqn:liouville-bipartitechannel}
    \mathbfcal{E}_{AB} =
    \resizebox{.9\hsize}{!}{
        \bordermatrix{~ 
        & \vket{X_{0} \otimes Y_{0}} & \vket{X_{j_1} \otimes Y_{0}}  & \vket{X_{j_1} \otimes Y_{j_2}} & \vket{X_{0} \otimes Y_{j_2}}  \cr
        \vbra{X_{0} \otimes Y_{0}} &  1 & \bm{0} & \bm{0} & \bm{0}  \cr
        \vbra{X_{i_1} \otimes Y_{0}}  & \mathbf{x}_{A \rightarrow A} & T_{A \rightarrow A} & T_{AB \rightarrow A} & T_{B \rightarrow A} \cr
        \vbra{X_{i_1} \otimes Y_{i_2}}  & \mathbf{x}_{AB \rightarrow AB} & T_{A \rightarrow AB} & T_{AB \rightarrow AB} & T_{B \rightarrow AB} \cr
        \vbra{X_{0} \otimes Y_{i_2}}  & \mathbf{x}_{B \rightarrow B} & T_{A \rightarrow B} & T_{AB \rightarrow B} & T_{B \rightarrow B} \cr
        }
    }\nonumber
\end{equation}
where $i_1=\{1, 2, ... , (d_A^2 -1)\}$, $i_2=\{1, 2, ... , (d_B^2 -1)\}$ and similarly for $j$. Here we break the entire $T$ matrix of the channel up according the the subsystem contributions where, for example, the term $T_{AB \rightarrow B}$ denotes the mapping of joint degrees of freedom of the input system $AB$ into the $B$ output subsystem.

With this notation in place, we can now define the general sub-unitarities of the bipartite channel.
\begin{definition}\label{all-sub-unitarities}
    For any quantum channel $\E_{AB}$ on a bipartite quantum system $AB$ the sub-unitarity $u_{X \rightarrow Y}$ of the channel is defined as:
    \begin{equation}\label{eqn:all-sub-unitarities}
        u_{X \rightarrow Y}(\mathcal{E}_{AB}) := \alpha_X \displaystyle \tr[T_{X \rightarrow Y}^\dagger T_{X \rightarrow Y}],
    \end{equation}
   for any $X,Y \in \{A,B,AB\}$ and with $\alpha_A = 1/(d_A^2 - 1)$, $\alpha_B = 1/(d_B^2 - 1)$ and $\alpha_{AB} = \alpha_A\alpha_B$.
\end{definition}
This coincides with our previous definitions for the single subsystem sub-unitarities, and also provides the form for the remaining other choices. Also note that under a local change of bases on the input and output subsystems we have
\begin{equation}
\E_{AB} \rightarrow (\V_A \otimes \V_B) \circ \E_{AB} \circ (\U^\dagger_A \otimes \U^\dagger_B),
\end{equation}
for local unitary channels denoted with $\V$ and $\U$. These changes of bases transform the sub-matrices $T_{X\rightarrow Y}$ under multiplication by orthogonal matrices. For example
\begin{equation}
T_{A\rightarrow A} \rightarrow \mathcal{O}_1 T_{A\rightarrow A} \mathcal{O}_2 ^T,
\end{equation}
for orthogonal matrices $\mathcal{O}_1, \mathcal{O}_2$, with e.g. $\mathcal{O}_2$ arising from the $\U_A(X_i) = \sum_{m=1}^{d_A^2-1} \mathcal{O}_{2;i,m} X_m $ (see Appendix \ref{append:u_c-properties}). This implies that all the sub-unitarity terms are invariant under local changes of bases.

It is straightforward to show (see Appendix \ref{append-sub-unit-properties}) that these sub-unitarities relate to the total unitarity of the quantum channel $\E_{AB}$ as follows.
\begin{theorem}
    The unitarity of a bipartite channel $\cal{E}_{AB}$ is obtained from the weighted sum of its sub-unitarities:
    \begin{equation}
        u(\cal{E}_{AB}) = \frac{1}{d^2-1} \sum_{X,Y \in \{A, B, AB\}}   \frac{u_{X \rightarrow Y}(\cal{E}_{AB})}{\alpha_X},
    \end{equation}
    where $d=d_A d_B$ is the dimension of the total system.
    \label{thm:unitarityexpnsion}
\end{theorem}

We shall make use of this decomposition of unitarity for our benchmarking protocol to estimate the correlated unitarity. But before discussing the protocol, we first give core properties of this measure that demonstrate its usefulness for assessing the correlation structure of a given channel.

\subsection{Properties of the correlated unitarity for a bipartite quantum channel}
The correlated unitarity $u_c$ is given in terms of sub-unitarities as $u_c(\E) = u_{AB \rightarrow AB} (\E) - u_{A\rightarrow A}(\E) u_{B\rightarrow B}(\E)$, and we now address the core properties of this measure. The following result shows that it obeys natural conditions.

\begin{theorem} For any bipartite quantum channel $\E_{AB}$, we have $u_c(\E_{AB}) \le 1$, and is invariant under local unitary transformations on either the input or output systems. Moreover $u_c(\E_{AB})=0$ for product channels and $u_c(\E_{AB})=1$ when $\E_{AB}$ is the \textit{SWAP} channel modulo local unitary changes of bases. 
\end{theorem}
A proof of this can be found in Appendices \ref{append-sub-unit--product-properties} \& \ref{section:uc-for-swap-channel}. Therefore, under this measure the \textit{SWAP} channel is the farthest from being a product channel, which is consistent with the fact that it perfectly transfers quantum information from one subsystem to the other.
\begin{figure}[t!]
    \centering
    \includegraphics[width=8.6cm]{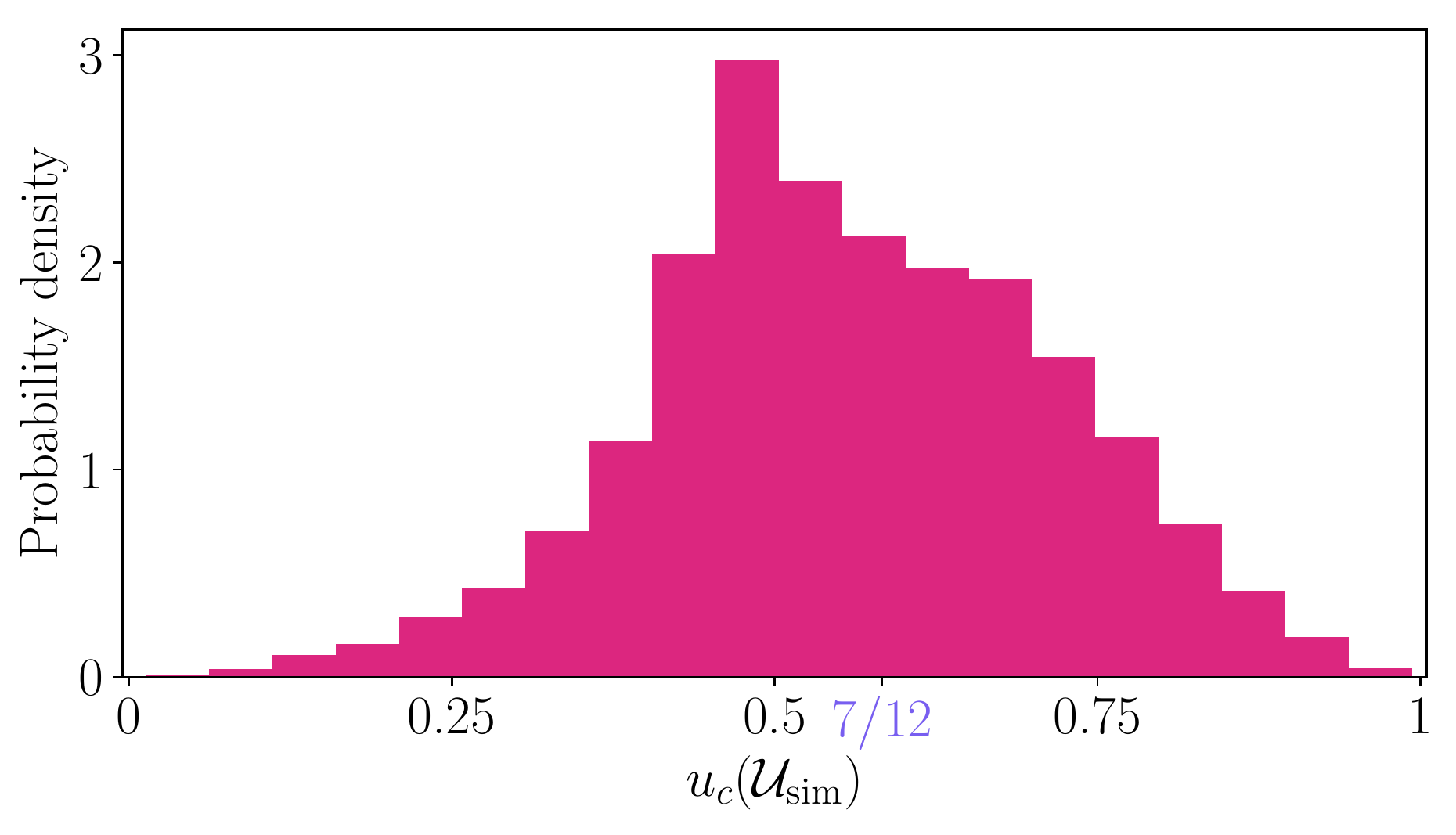}
    \caption{\textbf{Distribution of $u_c$ for $2$--qubit unitaries.} We plot the histogram of values of $u_c(\U_{\rm{sim}})$ for 20,000 random 2 qubit unitaries, $\U_{\rm{sim}}$. These correlated unitarities lie between $0$ and $1$, and take the value $u_c(\E_A \otimes \E_B)=0$ for product channels, and $u_c(SW\!AP)=1$ for the $SW\!AP$ channel. The value of $u_c$ is invariant under local unitary changes of basis. The upper bound for 2 qubit separable channels is $u_{\rm{c}}^{\rm{sep}}\leq 7/12$, and is also shown on the plot.  We sampled using the methods of \cite{bruzda2009random} and simulated using QuTip \cite{qutip2013}.}
    \label{fig:values-of-uc}
\end{figure}
However we can also consider intermediate regimes in which the bipartite channel is \emph{separable}, namely it can be written as
\begin{equation}
    \E_{AB} = \sum_k p_k \E_k \otimes \F_k,
\end{equation}
for some probability distribution $(p_k)$, and local channels $\E_k$ and $\F_k$ on $A$ and $B$ respectively. This class of channels are also known as Local Operations with Shared Randomness (LOSR) \cite{hsieh2020entanglement,gutoski2009properties}. The above definition generalizes that of separable states, and defines a convex subset of channels. A bipartite channel that is not separable is called \emph{non-separable}. It turns out that the correlated unitarity is strictly bounded over separable channels as the following establishes.
\begin{theorem}[Correlated unitarity is a witness of non-separability]\label{theorem:witness-bound}
    Given a bipartite quantum system $AB$ with subsystems $A$ \& $B$ of dimensions $d_A$ \& $d_B$ respectively, for a separable quantum channel $\E_{AB}$, we have that 
    \begin{equation}
        u_c(\E_{AB}) \le C(d_A, d_B) \leq \frac{17}{24} < 1,
    \end{equation}
    where
    \begin{equation}
        C(d_A , d_B) = \beta_A(1+\beta_B) (1-\frac{1}{\min(d_A^2,d_B^2)})   + \frac{1}{4}
    \end{equation}
    where $\beta_i = \frac{1}{d_i^2 -1}$ for $d_i = 2$ or $\beta_i = \frac{d_i}{d_i^2 -1}$ otherwise.
\end{theorem}
The proof of this bound is non-trivial, and we provide it in Appendix \ref{section:witness-of-non-sep}. This bound is not tight in general, and we provide sharper bounds in terms of the subsystem dimensions. The $d_A=d_B=3$ qutrit case provides the upper bound in $C(d_A , d_B)$ and could be improved, albeit via a non-trivial analysis of qutrit channels.

The consequence of the result is that if the correlated unitarity can be efficiently estimated, then obtaining values above the upper bound witnesses non-separability in the channel and so provides a practical way to certify quantum information transfer between $A$ and $B$.

The bound also relates to recent work on entanglement theory. Due to the limitations of the typical LOCC set of free channels when it comes to analysing Bell non-locality \cite{schmid2020standard}, it has been argued that LOSR channels provide a more sensible set. However, as LOSR channels are precisely the set of separable channels then any violation of the bound in Theorem \ref{theorem:witness-bound} implies the consumption of a resource state with respect to LOSR.


It is straightforward to compute $u_c$ for a range of channels. For example, consider the channel
\begin{equation}
    \E_{AB} = \sum_k p_k \U_k \otimes \V_k,
\end{equation}
where $\{\U_i\}_{i=1}^{d_A^2}$ and $\{\V_j\}_{j=1}^{d_B^2}$ are local unitary error bases \cite{knill1996non} on $A$ and $B$ respectively, namely unitaries on each subsystem that also form an orthonormal basis with respect to the Hilbert-Schmidt inner product. For this channel, $u_c(\E_{AB})$ then takes the form
\begin{equation}
    u_c(\E_{AB}) = \sum_k p_k^2 - (\sum_k p_k^2)^2.
\end{equation}

Further insight into $u_c(\E)$ can be obtained by formulating it in terms of two-point correlation measures. Suppose we have local observables $O_A$ and $O_B$ for system $A$ and $B$ respectively. We define the following correlation function
\begin{equation}
    \begin{split}
        F_{O_A, O_B}(\E,\psi_{AB}) := |\< O_A & \otimes O_B\>_{\E_{AB}(\psi_{AB})}|^2 \\
        & - |\<O_A\>_{\E_A(\psi_A)} |^2 |\<O_B\>_{\E_A(\psi_B)}|^2
    \end{split}
	\label{eqn:correlationfunction}
\end{equation}
where the channels $\E_A$ and $\E_B$  are local channels on $A$ respectively $B$ defined in Defn. (\ref{defn:sub-unitarity1}) and the input states $\psi_A$ and $\psi_B$ are marginals of $\psi_{AB}.$

The correlation function above becomes related to the covariance of classical random variables when considering classical states embedded in a quantum system
\begin{equation}
    \begin{split}
        F_{O_A, O_B}(id, \rho_{AB}) = \mathrm{cov}(O_A, O_B) [\<&O_A\otimes O_B\>_{\rho_{AB}} \\
        & + \<O_A\>_{\rho_A} \<O_B\>_{\rho_B} ],
    \end{split}
\end{equation}
 where $\rho_{AB} = \sum_{x,y} p(x,y) |x\>|y\>\<x|\<y|$ for $|x\>, |y\>$ computational basis states that diagonalize the hermitian operators $O_A$ and $O_B$ and $p(x,y)$ is a joint probability distribution with marginals $p(x)$ and $p(y)$.  In this case $\mathrm{cov}(O_A,O_B) =  \<O_A\otimes O_B\>_{\rho_{AB}} - \<O_A\>_{\rho_A} \<O_B\>_{\rho_B} $ and matches the covariance of classical random variables $X,Y$.

Then the correlated unitarity can be expressed as
\begin{equation}
	u_c(\E) = \alpha_{AB} \, d_{AB}^2\sum_{i,j,k,k'}  F_{P_i, P_j} (\E, \psi_{k,k'})
	\label{eqn:corrunitarityfunc}
\end{equation}
where $P_i$ \ are the traceless Pauli operators on each subsystem, and $\psi_{k,k'} = \frac{\ident_{AB} + P_k \otimes P_{k'}}{d_{AB}}$.

Overall, the correlated unitarity amounts to a working notion of correlation in a bipartite quantum channel, and we do not delve any further into its theoretical properties. In Appendix \ref{append:u-c-compared-to-norm} we also compare $u_c(\E_{AB})$ to a norm measure of correlation. While norm-based measures are mathematically more natural, our aim is to connect to benchmarking protocols, and so ultimately the utility of this measure should be judged by how useful it is in practice. We find that sub-unitarities arise very naturally in benchmarking protocols.

\section{Estimation of correlated unitarity via benchmarking protocols}\label{spam-robust-protocol-section}
In the previous section we developed a collection of tools, based around unitarity, to address sub-system features of a quantum channel. The introduction of sub-unitarities and the correlated unitarity allow us to quantify structures specific to bipartite quantum channels in a simple and direct manner. We now turn to the question of how such quantities may be estimated in practice in a protocol that is both efficient in the number of operations required and robust against SPAM errors. 

These quantities are generalizations of the unitarity, which can be efficiently estimated in benchmarking protocols, and it turns out similar methods work for sub-unitarities, however some complications do arise as we shall discuss. 

\subsection{Randomized Benchmarking Protocols}
The certification of quantum devices is a fundamental problem of quantum technologies, so as to verify that a physical device is actually performing with a sufficiently high fidelity. In the context of quantum computing it is desirable to provide a greater abstraction from the underlying physical implementation and talk of benchmarking a logical gate-set $\Gamma = \{ \U_1, \U_2, \dots,  \U_n\}$ of target unitary gates. 

The worst-case error rate is given by the diamond norm \cite{watrous2018theoryvec} distances $||\tilde{\U}_i - \U_i||_\diamond$, which is the relevant physical parameter for the fault tolerance theorem \cite{kitaev1997quantum}. However the diamond norm is a difficult quantity to measure, and so one must instead consider weaker measures, such as the average gate infidelity, given by
\begin{equation}
    r(\E) := 1 - \int d \psi \< \psi |\E (|\psi\>\<\psi|) |\psi\>,
\end{equation}
measuring the Haar-average deviation from the identity channel of a given channel $\E$. The average gate infidelity then provides bounds on the diamond distance of the form shown in equation (\ref{eqn:r-diamond-bound}).
The problem with this route is that the bounds cannot be tightened, and for $\E$ corresponding to a non-Pauli error there is a weak link between $r(\E)$ and the diamond norm \cite{wallman2014randomized, carignan2019bounding,kueng2016comparing}.

Randomized benchmarking techniques can be used to estimate $r(\E)$ and circumvent the exponential complexity of tomography, and the unavoidable SPAM errors. The core components of a randomized benchmarking protocol generally involves the noisy preparation of some initial quantum state $\rho$, which is then subject to a number $k$ of physical gates $\tilde{\U}_i$ that approximate target unitaries $\U_i \in \Gamma$, before a final imperfect measurement is performed for some binary outcome measurement $\{M, \I-M\}$. If the gates applied correspond to a (noisy) $2$-design, such as $\Gamma$ being the Clifford group, then it can be shown that \cite{kliesch2020theory} the resulting statistics are exponentially decreasing in $k$, namely $\mathbb{E}[m(k)]= c_1 +c_2 \lambda^k$, for constants $c_1$ and $c_2$ that contain the state preparation and measurement details. The decay constant $\lambda$ is then a measure of the noisiness of the physical gate-set $\tilde{\Gamma} := \{\tilde{\U}_i\}$ employed.

In the simplified model of \emph{gate-independent} noise, in which each channel can be decomposed as $\tilde{\U}_i = \E\circ \U_i$ for some $\E$ that is independent of $i$, then it can be shown that $\lambda \propto 1- r(\E)$, where $r(\E)$ is the average gate infidelity of the noise channel $\E$. In the more realistic case of \emph{gate-dependent} noise the relationship between the decay parameter $\lambda$ and the physics of the set $\tilde{\Gamma}$ is subtle, due to gauge degrees of freedom in the representation of the physical components \cite{proctor2017randomized}. However, despite these details the decay parameter can still be related to the physical gate-set and essentially corresponds to the average gate-set infidelity \cite{wallman2018randomized}.

At a more abstract level, a randomized benchmarking scheme admits a compact description in terms of convolutions of channels $\tilde{\U}_i$ with respect to the Clifford group \cite{merkel2018randomized}. The decay law is then viewed in a Fourier-transformed basis where the channel compositions become matrix multiplication over different irreps \cite{helsen2020general}. The resultant protocol then provides a benchmark for the degree to which the physically realized channels $\{\tilde{\U}_i\}$ form an approximate representation of the Clifford group \cite{gowers2017inverse, francca2018approximate}.

In the next section we expand on the components of the benchmarking scheme for the case of unitarity benchmarking.
\subsection{Unitary 2-designs \& Unitarity Benchmarking Protocols}
We now provide an outline of how the unitarity of a quantum channel can be estimated in a benchmarking protocol.

Recall that by $\mathbfcal{U}$ we denote the Liouville representation of a unitary channel $\U(X) = UXU^\dagger$, and therefore it takes the explicit form, $\mathbfcal{U} = U \otimes U^*$. A probability measure $\mu$ over the set of unitaries $U(d)$ is called a \textit{unitary 2-design} if we have that
\begin{equation}
    \int \!\!d \mu(U)\, \mathbfcal{U}^{\otimes 2} = \int \!\!d \mu_{\mbox{\tiny Haar}}(U)\, \mathbfcal{U}^{\otimes 2},
\end{equation}
where $\mu_{\mbox{\tiny Haar}}$ is the Haar measure over the group $U(d)$. In practice we are interested in unitary 2-designs which are finite, discrete distributions of unitaries. In particular the uniform distribution over the Clifford group $\mathcal{C}$ of unitaries is a $2$-design (in fact it is a $3$-design \cite{zhu2017multiqubit}), and therefore
\begin{equation}\label{unit2design}
    \frac{1}{|\mathcal{C}|} \sum_{U\in\mathcal{C}} \mathbfcal{U}^{\otimes 2} = \int \!\!d \mu_{\mbox{\tiny Haar}}(U)\, \mathbfcal{U}^{\otimes 2} =: P
\end{equation}
where $|\mathcal{C}|$ is the number of elements in the Clifford group and we denote the resultant operator by $P$. This operator acts on the vectorized form of $\B(\H)\otimes \B(\H)$, and using Schur-Weyl duality, it can be shown that $P$ is the projector onto the subspace
\begin{equation}\label{subspace}
    S:=\rm{span} \{ \vket{\I^{\otimes 2}}, \vket{\mathbb{F}}\},
\end{equation}
where $\mathbb{F}$ is the unitary that transposes vectors in the two subsystems, $ |\phi_1\> \otimes |\phi_2\> \rightarrow |\phi_2 \> \otimes |\phi_1\>$.

We can define an effective noise channel $\E$ via $\E := \U^\dagger \circ \tilde{\U}$, and moreover in what follows we shall assume for simplicity that each gate $U\in \Gamma$ is subject to the same effective noise channel (but again this assumption can be weakened and gate-dependent noise can be assessed via interleaved benchmarking \cite{magesan2012efficient}).

The unitarity of this noise channel can then be estimated in the following way. We prepare a quantum state $\rho$ of the system and choose the Clifford group as the gate-set. We now define
\begin{equation}
    \U_s := \U_{(s_1,s_2,\dots, s_k)} := \U_{s_1} \circ  \U_{s_2} \circ \cdots \circ  \U_{s_k},
\end{equation}
where $ \U_{s_i} \in \Gamma$ for all $i$, and $s_i$ labels the particular choice of unitary in the gate-set. We also denote by $\tilde{\U}_s$ the corresponding noisy implementation of the above sequence $s=(s_1, s_2, \dots, s_k)$ of unitaries. For any sequence $s$ and some hermitian observable $M$ we estimate the quantity
\begin{equation}\label{eqn:single-circuit}
    m(s) := \tr [ M \tilde{\U}_s (\rho) ],
\end{equation}
and then by randomly sampling over the Clifford group for each step in the sequence estimate $\mathbb{E}_s [ m(s)^2] := \frac{1}{|\Gamma|^k} \sum_{s} m(s)^2$. By exploiting the fact that the Clifford group is a $2$-design, and specifically equations (\ref{unit2design}) and (\ref{subspace}), it was shown in ~\cite{wallman2015estimating} that
\begin{equation}\label{fitting_fn_eqn_2}
    \mathbb{E}_s [ m(s)^2] = c_1 + c_2 u(\E)^{k-1},
\end{equation}
for constants $c_1$ and $c_2$ that contain any errors due to state-preparation or measurement. Therefore, by repeating this estimation for sequences of varying length we may extract an estimation of $u(\E)$ as a decay constant for the quantity in an efficient and SPAM-robust manner.

\subsection{Estimation of channel sub-unitarities via local \& global twirls}

The unitarity arose from considering a global twirl using a $2$-design, it turns out that the sub-unitarities arise in a similar fashion, but now by considering local twirls for a bipartite quantum system. Specifically, we now have a bipartite quantum system $AB$ with local gate-sets $\Gamma_A$ and $\Gamma_B$, which we assume are $2$-designs, and a global gate-set $\Gamma_{AB}$. Then, we may consider the independent twirls
\begin{equation}\label{eqn:twirl-each-subsystem}
    \frac{1}{|\Gamma_A| |\Gamma_B|} \sum_{U_A \in \Gamma_A, U_B \in \Gamma_B}  \mathbfcal{U}_A^{\otimes 2} \otimes  \mathbfcal{U}_B^{\otimes 2}=P_A \otimes P_B,
\end{equation}
where we now have local projections of channels at $A$ and $B$ onto subspaces $S_A$ and $S_B$, where
\begin{equation}
    S_A = \rm{span} \{ \vket{\I_A\otimes \I_{A'}}, \vket{\mathbb{F}_{AA'}}\},
\end{equation}
where $A'$ is isomorphic to $A$, and we have a similar expression for $S_B$.

In the context of benchmarking we have the problem of determining the \emph{addressability} of qubits and the existence of \emph{crosstalk} between qubits. For example, we want to implement some target unitary $\U_i \otimes id$ on one qubit, while leaving all others unaffected. However, in reality the physical channel performed $\tilde{\U}_i$ will involve an effective noise channel $\E$ that does not factorize neatly with noise only on the target qubit. Instead, the noise channel will act non-trivially on each subsystem of the bipartite split and could involve correlations that include the leakage of quantum information.

In what follows we again consider the averaged noise channel over the gate-set, and so at the simplest level of analysis assume that we have gate-independent noise. A more general analysis involving gate-dependent noise should be possible by following perturbative approaches such as in \cite{wallman2018randomized} and by making use of interleaved benchmarking \cite{magesan2012efficient}. We also note that the channel under consideration need not be a noise channel in such a scheme, but could be a target channel on which we wish to do robust tomography. For this context it would be possible to exploit recent methods that make use of randomized benchmarking to do tomography of quantum channels such as in \cite{kimmel2014robust}. We leave this kind of analysis for later investigation.

Under this average noise model assumption, we now perform a unitarity benchmarking scheme by randomly sampling from $\Gamma_A \otimes \Gamma_B$ and obtain a circuit of depth $k$, with sequence indexed via $s = (s_A, s_B)$ with $s_A = (a_1, a_2, \dots, a_k)$ and $s_B = (b_1, b_2, \dots, b_k)$ specifying the particular target unitary in the local gate-sets. As before, we estimate the quantity $m(s) := \tr[ M \tilde{\U}_s (\rho)]$ and also $\mathbb{E}_s [ m(s)^2]$ for circuits of depth $k$. However, for these local twirls, this quantity now has a different decay profile. As we show in Appendix \ref{appendix-for-pAB-protocol} this quantity behaves as
\begin{equation}\label{triple-decay}
    \mathbb{E}_s [ m(s)^2] = c_{00} + c_{01} \lambda_1^{k-1} +c_{10} \lambda_2^{k-1}+c_{11} \lambda_3^{k-1},
\end{equation}
where $(\lambda_1, \lambda_2, \lambda_3)$ are the singular values \footnote{This assumes a non-degenerate form of a Jordan matrix decomposition. Degenerate cases give rise to similar expressions. See Appendix \ref{section:jordan-decomp-pab} for details.} of the matrix of sub-unitarities
\begin{equation}
    \S =
    \resizebox{.8\hsize}{!}{$
    \begin{pmatrix}
        u_{A \to A}(\E) & \frac{1}{\sqrt{\alpha_B}} u_{AB \to A}(\E) & \sqrt{\frac{\alpha_A}{\alpha_B}} u_{B \to A}(\E) \cr
        \sqrt{\alpha_B} u_{A \to AB}(\E) & u_{AB \to AB}(\E) & \sqrt{\alpha_A} u_{B \to AB}(\E) \cr
        \sqrt{\frac{\alpha_B}{\alpha_A}} u_{A \to B}(\E) & \frac{1}{\sqrt{\alpha_A}} u_{AB \to B}(\E) & u_{B \to B}(\E) \cr
    \end{pmatrix}
    $},
\end{equation}
with $\alpha_X = \frac{1}{d_X^2 -1}$, and the constants $c_{00}, \dots, c_{11}$ contain the SPAM-errors. Therefore, the sub-unitarities arise in the context of this benchmarking, albeit in a more non-trivial form to the global protocol. For example, we have that
\begin{equation}
    \tr(\S)=\sum_i \lambda_i = u_{A \to A} (\E)+u_{AB \to AB} (\E)+u_{B \to B} (\E),
\end{equation}
with similar relations existing for the other coefficients of the characteristic polynomial of $\S$ \cite{horn2012matrix}. Note that $\sum_i \lambda_i =3$ if and only if $\E$ is a product of unitaries, and so this sum of eigenvalues gives a blunt handle on how much $\E$ deviates from this regime.

By estimating the decay constants in equation (\ref{triple-decay}) it is possible to obtain an estimate of channel correlations that coincides with the correlated unitarity for a family of channels. It is easily checked that for a product noise channel $\E = \E_A \otimes \E_B$ we have the matrix of sub-unitarities given by
\begin{equation}\label{product-a-matrix}
    \S =
    \resizebox{.8\hsize}{!}{$
    \begin{pmatrix}
        u(\E_A) & 0 & 0 \cr
        \sqrt{\alpha_B} u(\E_A) x_B &  u(\E_A)  u(\E_B) &  \sqrt{\alpha_A} u(\E_B) x_A \cr
        0 & 0 & u(\E_B) \cr
    \end{pmatrix}
    $},
\end{equation}
where $x_A$ and $x_B$ are constants related to deviations from unitality (see Appendix \ref{section:pab-prod}). This implies that eigenvalues of $\S$ are given by
\begin{equation}
    \{ \lambda_i \} = \{u(\E_A), u(\E_B), u(\E_A)  u(\E_B) \}.
\end{equation}  
It can be checked that this simple link with sub-unitarities extends to arbitrary \emph{separable} channels, for which $\lambda_1, \lambda_2, \lambda_3$ are exactly equal to the sub-unitarities $u_{A\rightarrow A}, u_{B\rightarrow B},u_{AB\rightarrow AB}$. This provides a way to compute the correlated unitarity.
More precisely, given $\lambda_1 \geq \lambda_2 \geq \lambda_3$, we may compute the quantity
\begin{equation}\label{eqn:simple-correlation-measure}
    C = \abs{\lambda_3 - \lambda_1 \cdot \lambda_2},
\end{equation}
where we use the fact that sub-unitarities are upper bounded by one to distinguish $\lambda_3$ from the other two. 

For non-separable channels the deviation of the eigenvalues from each of the subunitarities can be bounded by using the Girshgorin Circle Theorem or Brauer's Theorem  \cite{horn2012matrix}. For example, we obtain the bounds
\begin{equation}
    |\lambda_1 - u_{A \rightarrow A} (\E)| \le \frac{1}{\sqrt{\alpha_B}}u_{AB \to A}(\E) + \sqrt{\frac{\alpha_A}{\alpha_B}} u_{B \to A}(\E).
\end{equation}
Using identities for sub-unitarities, we can further show that
\begin{equation}
    |\lambda_1 - u_{A \rightarrow A} (\E)| \le \frac{1}{\sqrt{\alpha_B}}[1 - u_{A \rightarrow A} (\E) ].
\end{equation}
These two inequalities are generally weak, due to the factors of $\alpha_B$ and $\alpha_A$, but they do imply that the approximation is very good when either the off-diagonal elements are small or when the local unitarities are large. In such regimes this protocol will return a good estimate of the correlated unitarity, as shown in Figure \ref{fig:convergence-of-measures}.
\begin{figure}[t!]
    \centering
    \includegraphics[width=8.6cm]{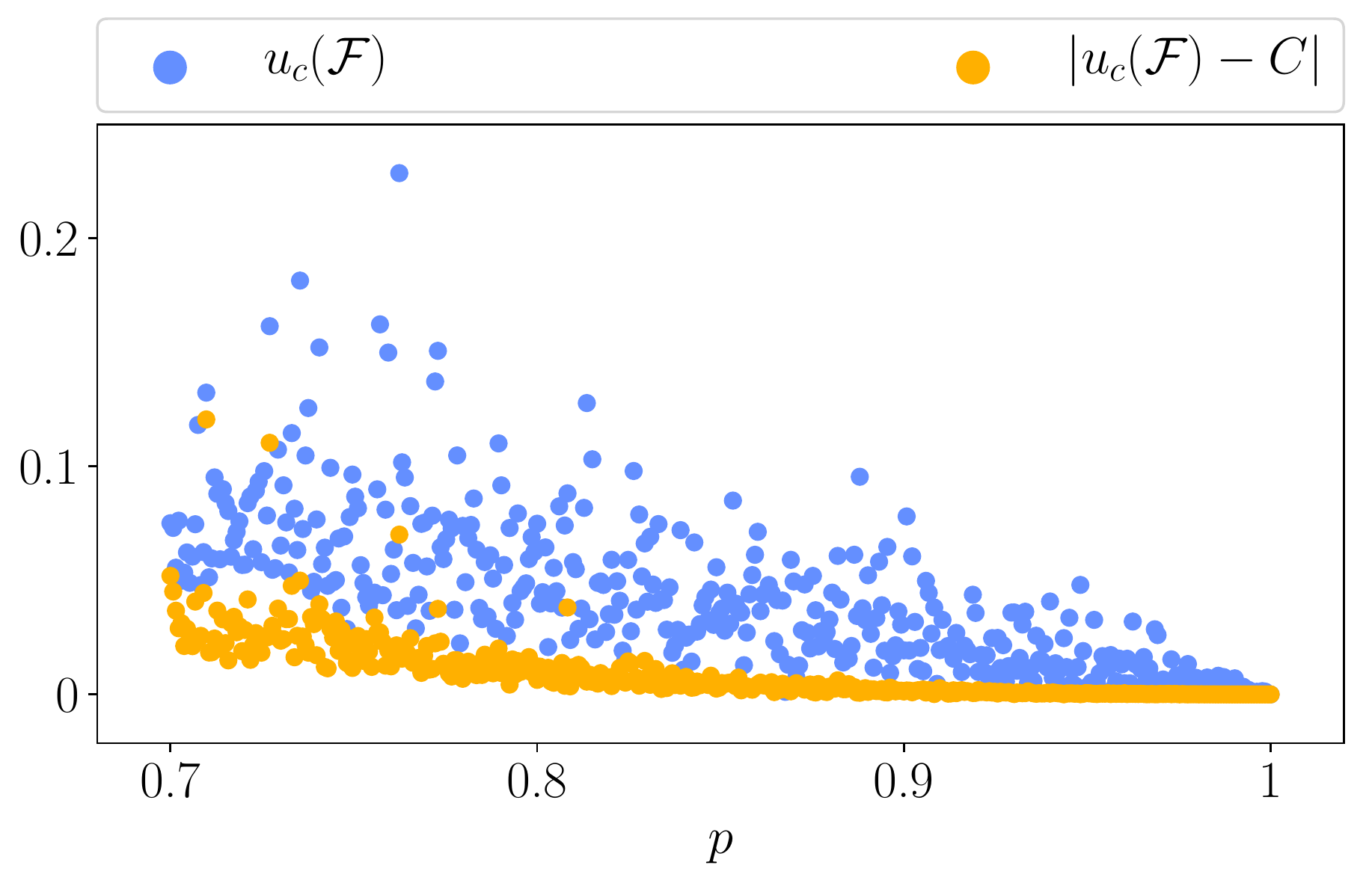}
    \caption{\textbf{SPAM error robust estimation of $u_c$ for generic quantum channels.} The convergence of the values of correlated unitarity and $C$ as gate noise takes a product form, for a 2 qubit simulation. We show $\abs{u_c - C}$ over $p$, where $\mathcal{F} = p \mathcal{E}_A \otimes \mathcal{E}_B + (1-p)\mathcal{G}$. The channels $\mathcal{E}_A$, $\mathcal{E}_B$ and $\mathcal{G}$ are sampled using the methods of \cite{bruzda2009random} and simulated using QuTip \cite{qutip2013}.}
    \label{fig:convergence-of-measures}
\end{figure}

Estimation of the three decay constants requires fitting noisy multi-exponential data which is non-trivial, but a range of methods have been developed to tackle this problem \cite{helsen2020general}. To assist with fitting, and moreover identify the sub-unitarity $u_{AB \to AB}$, we may supplement the local twirling with a global estimate of unitarity, and then make use of the decomposition of unitarity into sub-unitarities. Specifically, for the case of \emph{unital} separable channels, with $d_A=d_B=d$, we have that
\begin{equation}
    u(\E) = \frac{u_{A \to A}(\E) + u_{B \to B}(\E) + (d^2-1)u_{AB \to AB}(\E) }{d^2 + 1},
\end{equation}
and therefore we have the relation
\begin{equation}\label{uab-from-u}
    u_{AB \to AB}(\E) = \frac{ (d^2 + 1) u(\E) - \sum_i \lambda_i}{d^2 -2}.
\end{equation}
This means that separate estimations of $u(\E)$ and the decay constants $(\lambda_i)$ provide an estimate of $u_{AB \to AB}(\E)$, and so provides additional independent information on the terms entering the correlated unitarity.  In practice, this will require careful consideration as the average noise channel associated with $\Gamma_{AB}$ (employed in the estimation of unitarity) might be different than that associated with $\Gamma_{A}\otimes \Gamma_{B}$.

We note that by using \emph{randomized compiling} \cite{wallman2016noise, hashim2020randomized} for the implementation of a quantum circuit we may reduce the noise channel to being a Pauli channel. Since a general noise channel will not have $\lambda_i$ coinciding precisely with the sub-unitarities, by running the local twirling protocol  with and without randomized compiling one could witness the presence of non-Pauli noise.

\begin{protocol}[SPAM error robust, $\cal{C}\times\cal{C}$]
    \label{protocol:CxC}
    \begin{enumerate}[wide, labelwidth=!, labelindent=0pt]
        \setlength{\itemsep}{2pt}
        \setlength{\parskip}{0pt}
        \setlength{\parsep}{0pt}
        \item \textbf{Prepare} the system in a state $\rho$.
        \item \textbf{Select} a sequence of length $k$ of simultaneous random noisy Clifford gates locally on subsystems $A$ and $B$, starting with $k=1$. E.g. for each gate $\mathcal{U}_{AB,i}=\mathcal{U}_{A,i_1} \otimes \mathcal{U}_{B,i_2}$.
        \item \textbf{Estimate} the square $(m)^2$ of an expectation value of an observable $M$, for this particular sequence of gates.
        \item \textbf{Repeat 1, 2 \& 3} for many random sequences of the same length, finding the average estimation $\mathbb{E}[(m)^2]$ of $(m)^2$.
        \item \textbf{Repeat 1, 2, 3 \& 4} increasing the length of the sequence $k$ by 1.
        \item \textbf{Fit} the data $\mathbb{E}[(m)^2]$ against $k$ and obtain decay parameters as in Equation (\ref{triple-decay}). 
    \end{enumerate}
\end{protocol}

\subsection{Estimation of sub-unitarities for non-separable channels with low re-setting errors}
While the local twirling protocol provides a means to estimate the correlated unitarity in the case of any separable channel, we would like to be able to estimate such correlations for general non-separable channels. The obstacle here is to determine sub-unitarities such as $u_{A \to A} (\E_{AB})$. However, this requires preparing the maximally mixed state on subsystem $B$ and benchmarking the unitarity of the effective channel output on $A$. This presents a problem of how accurately such a re-set can be performed. Current devices, including ion-traps \cite{pino2020demonstration} and IBM's superconducting qubits \cite{corcoles2021exploiting}, allow for mid-circuit measurements and resets. These dynamical circuits capabilities can be accessed through hardware-agnostic SDKs \cite{sivarajah2020t, Qiskit}.
 
This is challenging to do in a fully SPAM-robust way, however from the form of Equation (\ref{eqn:sub-A}) we see that if it is possible to do a resetting of sub-system close to the maximally mixed state then one can obtain an estimate of the sub-unitarity $u_{A\to A}(\E_{AB})$, and similarly for other single-subsystem cases, by estimating the unitarity of the marginal channel $ \E_A = \tr_B \circ \E_{AB} \circ \R_B$, where $\R_B(\rho) = \frac{1}{d} \I_B$. Within the benchmarking circuit this would mean performing a noisy re-set $\tilde{\R}_B$ on $B$ after each $\tilde{\U}_i$ on $A$, with the aim of having $\tilde{\R}_B \approx \R_B$. This is a non-trivial assumption, and so in general the protocol will not be fully robust against re-set errors. However, if these errors are substantially smaller than the addressability errors one wishes to estimate then the protocol returns an approximate estimate.

We can summarize this sub-unitarity protocol as follows:
\begin{protocol}[SPAM error effected, $\cal{C}\times 1$]
    \label{protocol:Cx1}
    \begin{enumerate}[wide, labelwidth=!, labelindent=0pt]
        \setlength{\itemsep}{2pt}
        \setlength{\parskip}{0pt}
        \setlength{\parsep}{0pt}
        \item \textbf{Prepare} the system in the state $\rho$.
        \item \textbf{Select} a sequence of length $k$ of random noisy Clifford gates on subsystem $A$, starting with $k=1$. E.g. for each gate $\mathcal{U}_{A,i} \otimes id_B$
        \item \textbf{Estimate} the square $(m_A)^2$, of the expectation value of an observable $M_A$ on subsystem $A$ for this particular sequence of gates, while performing a \textbf{reset} $\cal{R}_B$ of the $B$ subsystem after every gate.
        \item \textbf{Repeat 1, 2 \& 3} for many random sequences of the same length, finding the average estimation $\mathbb{E}[(m_A)^2]$ of $(m_A)^2$.
        \item \textbf{Repeat 1, 2, 3 \& 4} increasing the length of the sequence $k$ by 1.
        \item \textbf{Fit} the data $\mathbb{E}[(m_A)^2]$ against $k$ and obtain decay parameters as in Equation (\ref{fitting_fn_eqn_2}).
    \end{enumerate}
\end{protocol}

\begin{figure}[t!]
	\centering
	\includegraphics[width=8.6cm]{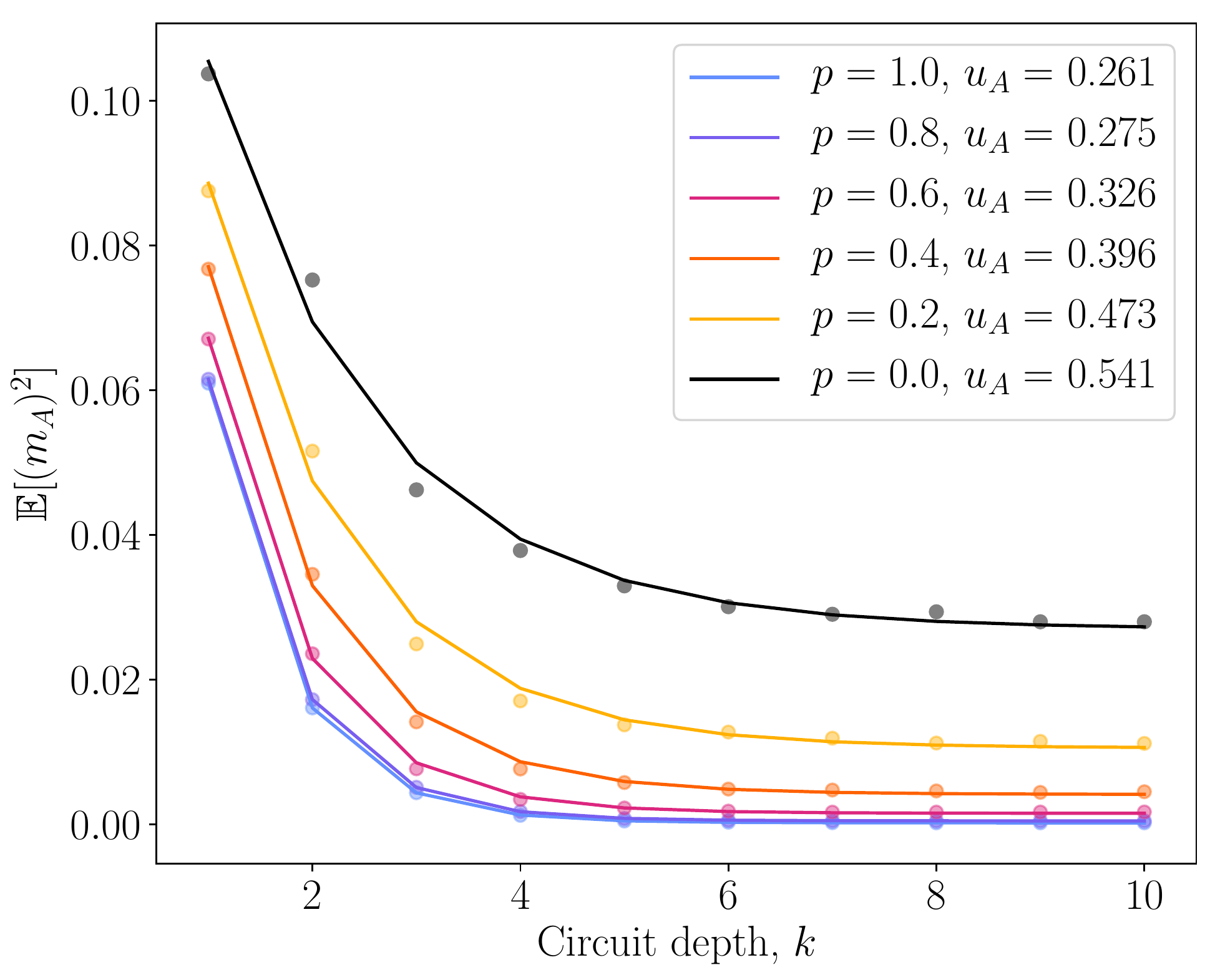}
	\caption{\textbf{Sub-unitarity estimation with re-set error.} Shown is a simulation of Protocol 2 to estimate the sub-unitarity $u_A \equiv u_{A \to A}(\E_{AB})$, modelling the re-set error associated $B$ as in Equation (\ref{eqn:coherent-reset}). This re-set error is shown for different levels of depolarization $p$, including $p=0$ i.e. no reset. The channel $\E_{AB}$ in this case has a theoretical value of $u_{A \to A}(\E_{AB})=0.261$. The protocol returns an estimate of the sub-unitarity accurate to $\sim 90$\% for re-set errors up to $\sim 20$\%.}
	\label{fig:uaa-spam-coherent-error}
\end{figure}

Given approximate estimates of $u_{A \to A}(\E_{AB})$ and $u_{B \to B}(\E_{AB})$ we may then exploit the fact that $\sum_i \lambda_i = u_{A \to A}(\E_{AB})+u_{B \to B}(\E_{AB})+u_{AB \to AB}(\E_{AB})$ to infer the value of $u_{AB \to AB}(\E_{AB})$ and thus compute the correlated unitarity for the channel $\E_{AB}$. Therefore, under the assumption of sufficiently small re-setting errors we may estimate the correlated unitarity for an arbitrary channel. Note that in the context of the local Clifford gate-sets the effective channel need not be the same in each protocol since Protocol 2 uses a different gate-set. However, we can use the same gate-set in Protocol 2 as in 1, since the application of non-trivial Clifford gates on $B$ does not change matters if $\tilde{\R}_B \approx \R_B$.

It is straightforward to numerically test how sensitive the above protocol is to re-setting errors. For example, one can model such re-set errors as depolarizing
\begin{equation}\label{eqn:coherent-reset}
\tilde{\R}_B = id_A \otimes (p \R_B + (1-p) id_B),
\end{equation}
where $p \in [0,1]$. In Figure \ref{fig:uaa-spam-coherent-error} we plot the benchmarking decay curves and find that for re-setting errors up to $\sim 20\%$ the protocol returns an estimate of the sub-unitarity $u_{A \rightarrow A}(\E)$ accurate to $\sim 90\%$. Note that such a channel will not in general destroy correlations between $A$ and $B$, in contrast to a stronger, more simplistic error model of
\begin{equation}\label{eqn:incoherent-reset}
\tilde{\R}_B (\rho_{AB}) = \rho_A \otimes ( \frac{1}{2}( \I + \mathbf{b} \cdot \boldsymbol{\sigma})),
\end{equation}
where one assumes a re-set to a local qubit state with non-zero Bloch vector $\mathbf{b}$. Under this stronger model assumption a simulation shows that such a scenario returns a good estimate for the sub-unitarity for $|\mathbf{b}| \le 0.2$.

\begin{figure}[t!]
	\centering
	\includegraphics[width=8.0cm]{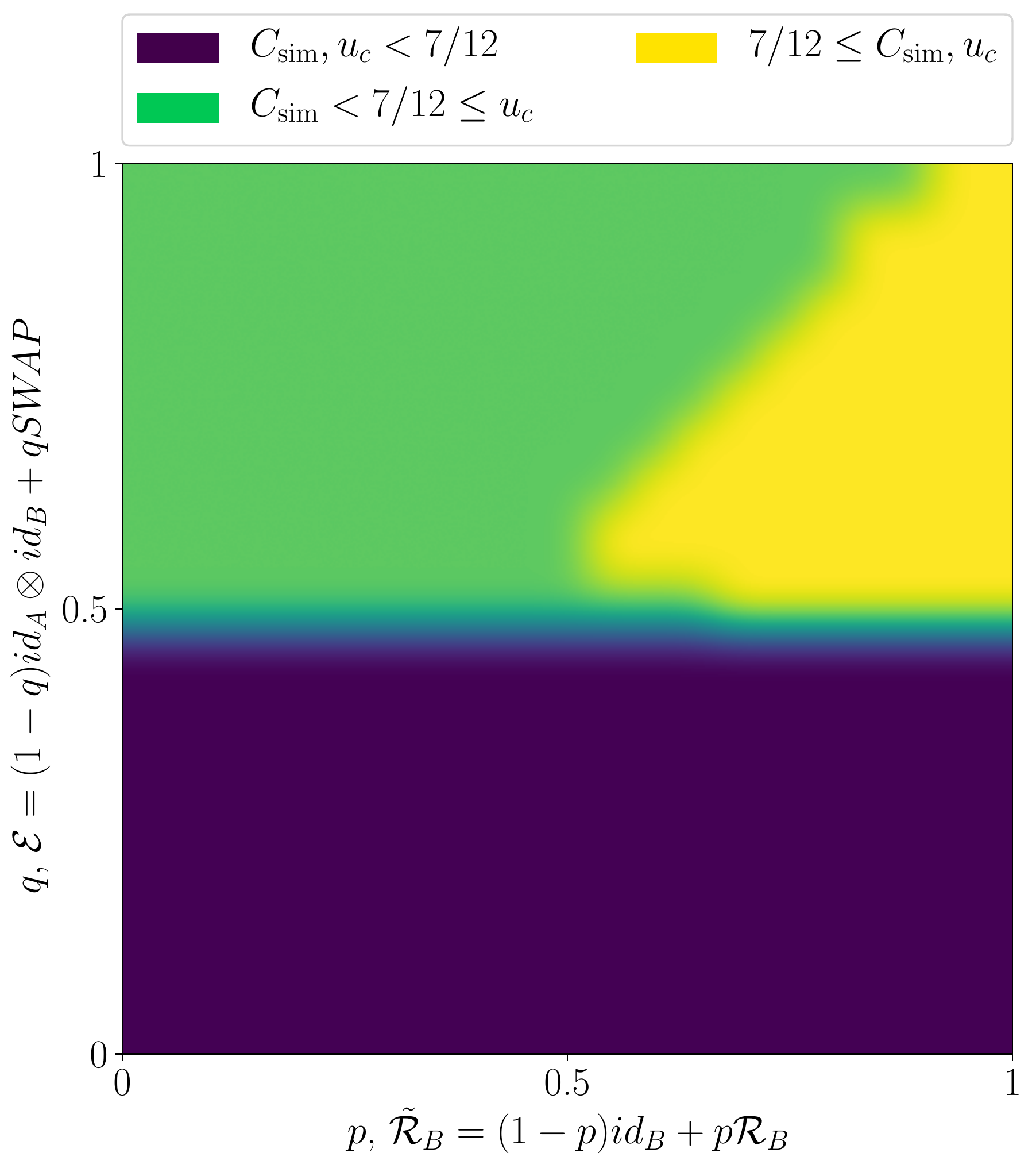}
	\caption{\textbf{Witnessing channel non-separability.} Given a quantum channel $\E_{AB}$ we consider the ability to efficiently witness its non-separability via correlated unitarity in the presence of re-setting noise. This could be realized, for example, in the context of robust tomography using randomized benchmarking \cite{kimmel2014robust}. We consider a $1$--parameter family of $2$-qubit channels obtained from a convex mixture of the maximally non-separable \textit{SWAP} channel and the identity channel (a product channel). The contour plot compares the true value of correlated unitarity $u_{c}(\E_{AB})$ with the correlation measure $C_{\rm sim}\approx C$ estimating Equation (\ref{eqn:simple-correlation-measure}) in the presence of re-set errors. For two qubits, non-separability occurs if $u_c(\E_{AB}) > 7/12$. We simulate both Protocol 1 and 2, and we find that for a wide range of re-set errors we may witness non-separability for $p,q \gtrsim 0.5$. The region of green where $p,q \ge 1/2$ is an artifact of our particular method, and with a more refined algorithm we expect detection of non-separability also in this region.}
	\label{fig:swap-witness}
\end{figure}

There are further variants around the above protocol. For example, if re-setting to non-maximally mixed states have very low errors then this provides another means to estimate $u_{A\to A}(\E_{AB})$. For example, if a low-error re-set to the pair of states $\frac{1}{2}( \I  \pm \mathbf{b} \cdot \boldsymbol{\sigma})$ is possible for some $\mathbf{b}$ then it can be shown that the average unitarity of the output on $A$ over the pair is always an upper bound on $u_{A\to A}(\E_{AB})$ (see Appendix \ref{RBforsubunit} for details), and so would provide a lower bound on the correlated unitarity. Therefore, this would allow witnessing of non-separability under the preceding assumptions.

In theory, another source of information that could be exploited is the unitarity of the channel from $AB$ to $A$, given by
\begin{equation}
    \E_{AB\to A}(\rho) = \tr_B \circ \E_{AB}(\rho).
\end{equation}
In terms of sub-unitarities this quantity can be decomposed as
\begin{align}
    u(\E_{AB\to A}) &= \frac{1}{(d_Ad_B)^2-1} \left (\frac{1}{\alpha_A}u_{A \to A}(\E_{AB}) + \right .\nonumber \\
    & \left . +\frac{1}{\alpha_B}u_{B \to A}(\E_{AB}) +\frac{1}{\alpha_A\alpha_B}u_{AB \to A}(\E_{AB}) \right).
\end{align}
However, while this provides an expression in terms of sub-unitarities without requiring re-setting, the standard benchmarking protocol will not work here due to the input and output systems being of different dimensions,  and therefore a more involved protocol would be required.

\subsection{Addressability of qubits and sub-unitarities}
Several methods have recently been developed for detection \cite{sarovar2020detecting}, characterization \cite{gambetta2012characterization, mckay2020correlated} and mitigation \cite{winick2020simulating} of unwanted correlations between subsystems (specifically cross-talk) in a quantum device from a hardware-agnostic and model independent perspective.  Our work adds to this toolkit new methods to characterize non-separable correlations and provides information about noise channels that is independent from features captured by previous works.

Simultaneous randomized benchmarking (SimRB) \cite{gambetta2012characterization} compares the increase in error rates when both subsystems are simultaneously and independently driven vs when one subsystem is driven and the other is kept idle. This quantifies the amount of new errors experienced by a subsystem as a result of simultaneously applying Clifford gates on the other. As it is the case for Protocol 2, due to the local independent Clifford twirl on one subsystem, SimRB is also affected by SPAM, and strong errors may be detected by deviations from exponential decay~\cite{gambetta2012characterization}. 

To compare with the information obtained from sub-unitarities, a quantity to detect correlations can be determined from the simultaneous Clifford twirl as in ~\cite{gambetta2012characterization}. We denote this quantity by
\begin{equation}
	a(\cal{E}_{AB}) := e_{AB} - e_A \cdot e_B.
\end{equation}
where $\E_{AB}$ is an effective noise channel associated to the Clifford gate-set acting locally on each subsystem $A$ and $B$.  The three decay parameters $e_{AB}$, $e_A$ and $e_B$ are extracted from the randomized benchmarking protocol that applies simultaneous local Clifford gates to subsystems $A$ and $B$ and are given in terms of the Liouville data for the channel as
\begin{equation}
	\begin{split}
		e_A &=  \alpha_A \tr[T_{A \to A}], \\
		e_B &=  \alpha_B \tr[T_{B \to B}], \\
		e_{AB} &=  \alpha_A  \alpha_B \tr[T_{AB \to AB}],
	\end{split}
\end{equation}
with the coefficients $\alpha_X$ as defined earlier.

\begin{figure}[t!]
	\centering
	\includegraphics[width=8cm]{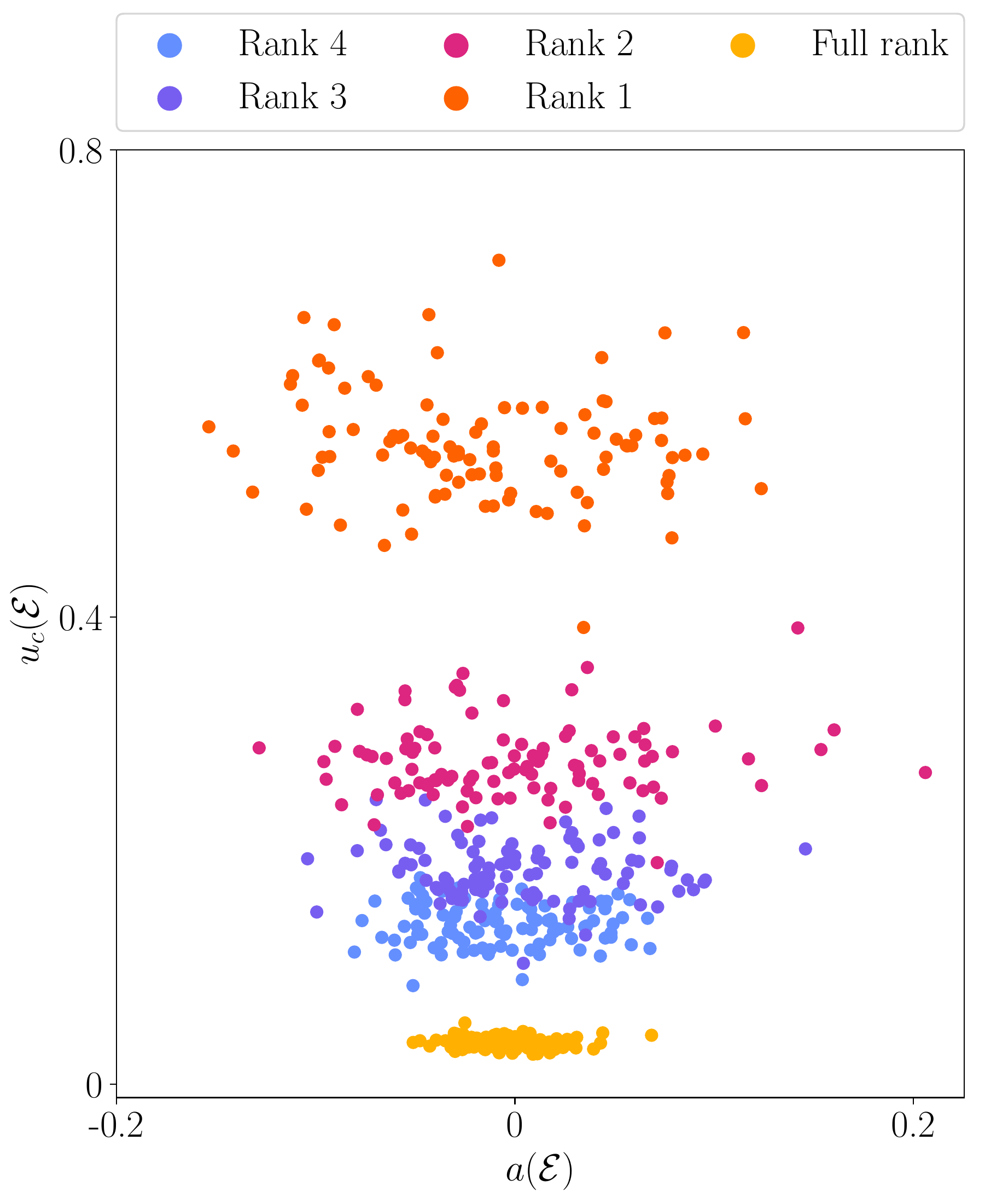}
	\caption{\textbf{Correlated unitarity vs. addressability.} Correlated unitarity is largely independent from existing addressability measures, while Kraus rank is a better indicator of the the value of $u_c$, which is consistent with it capturing the non-separable correlations between subsystems. This suggests the measure might be suitable for benchmarking $2$-qubit gates where  the unitary transfer of quantum information between subsystems is required. The above plot is for random channels of different ranks from the distributions of Bruzda \textit{et al.} and simulated using QuTip \cite{bruzda2009random, qutip2013}.}
	\label{fig:comparing-measures}
\end{figure}
For a product channel, $T_{AB \to AB} = T_{A \to A} \otimes T_{B \to B}$ and therefore $a(\cal{E}_A \otimes \cal{E}_B)=0$. In this manner, any deviation of $a(\cal{E})$ from zero is taken as detection of correlated behaviour. Note that in contrast to sub-unitarities, these measures are not invariant under local basis changes which makes it more problematic to interpret as a strict correlation measure.

It is easy to verify that the correlated unitarity provides independent information to a SimRB protocol, for example the CNOT gate is undetected by the addressability correlation measure; however it is detected by correlated unitarity. Figure \ref{fig:comparing-measures} shows this is generic for bipartite channels, and we find that there are regions where the addressability correlation measure is zero or close to it, but the correlated unitarity varies greatly. 

\section{Outlook}
Our starting point in this work was to develop simple, yet effective measures of correlations in quantum channels and means to assess sub-structures of such channels. The approach was motivated and guided by the idea of introducing measures that can be both efficiently estimated through RB-type of techniques and interpreted operationally as to quantify non-separable correlations.

Certain sub-unitarities of a general bipartite channel can be interpreted as unitarities of locally acting channels induced by state preparation and discarding on one subsystem.  Furthermore, we showed that they satisfy a set of inequalities that express an information-disturbance relation. This opens up new directions to analyse non-classical features of quantum channels directly from their robust tomographic description ~\cite{kimmel2014robust}.

In the context of benchmarking of quantum devices, it will be of interest to develop hardware implementations of the protocols here and determine how effective and useful they are in practice. Such analysis will closely investigate the effects of re-set errors for the subsystem unaddressed by target gates. Our simulations show that our second protocol, while not fully robust can still allow small re-set errors to estimate magnitudes of correlated noise, but ultimately whether this is a reasonable assumption must be assessed for the system at hand.

Throughout this work we consider the induced error to be time-independent and gate-independent and averaged for the gate-set considered. As such, relaxing these constraints would be a natural line to develop~\cite{wallman2018randomized, ball2016effect, qi2021randomized}.

Our primary protocol relies on fitting a multi-exponential decay to noisy data. In general this is a hard problem, and there will be many fits that will approximate the decay curve. The protocol could be substantially improved by exploiting recent statistical techniques~\cite{harper2019statistical}, algorithms for multi-exponential fitting~\cite{helsen2020general, xue2019benchmarking} and other approaches such as spectral tomography \cite{helsen2019spectral}.

\section*{Acknowledgments}
We would like to thank Matteo Lostaglio for helpful discussions. MG is funded by a Royal Society Studentship.
DJ is supported by the Royal Society and also a University Academic Fellowship.

\clearpage
\onecolumngrid 
\appendix
\section{Quantum operations and a review of notation}

\subsection{Review of notation}\label{notation}
Throughout these Appendices we consistently use the same notation as the main text, which we review here.

We consider an open bipartite quantum system with an associated an Hilbert space $\mathcal{H}_A \otimes \mathcal{H}_B$ and dimension $d = d_A d_B$.  Quantum channels act on the system such that $\E_{AB}$: $\cal{B}(\mathcal{H}_A \otimes \mathcal{H}_B) \to \cal{B}(\mathcal{H}_A \otimes \mathcal{H}_B)$, and unless otherwise stated we assume for simplicity that the input and output systems are identical. We denote all vectorized quantities in boldface, $|\boldify{M}\> := |vec(M)\>$ for any operator $M \in \B(\H_A \otimes \mathcal{H}_B)$ and similarly, we denote the Liouville representation $\boldsymbol{\E_{AB}}: = \mathcal{L}(\E_{AB})$ for any channel $\E_{AB}$, as detailed in the main text.

For subsystem $A$, we choose an orthonormal basis of operators $X_\mu = (X_0 =  \frac{1}{\sqrt{d_A}}\I_A, X_i)$, where $d_A$ is dimension of the subsystem $A$, and $\tr[X_\mu^\dagger X_\nu] = \delta_{\mu \nu}$. Similarly for $B$ an orthonormal basis $Y_\mu = (Y_0 = \frac{1}{\sqrt{d_B}}\I_B, Y_i)$. Together these provide a basis for the full system which is given in the Liouville representation as
\begin{equation}
    \vket{X_\nu \otimes Y_\mu} := \vket{X_\nu} \otimes \vket{Y_\mu}.
\end{equation}
Furthermore  $\{ \vket{X_\nu \otimes Y_\mu} \}_{\nu,\mu}$ forms a complete orthonormal basis for $\cal{H}_{A}\otimes \cal{H}_{A} \otimes \cal{H}_{B}\otimes \cal{H}_{B}$, and with respect to this basis, the Liouville representation of $\E_{AB}$ corresponds to a matrix $\boldsymbol{\E_{AB}}$ whose entries satisfy
\begin{equation}
	\vbra{X_\nu \otimes Y_\mu}  \boldsymbol{\E_{AB}}  	\vket{X_{\nu'} \otimes Y_{\mu'}} = \Tr(X_{\nu} \hc\otimes Y_{\mu}\hc \, \E_{AB}(X_{\nu'} \otimes Y_{\mu'})).
\end{equation}

For simplicity, where there is no ambiguity on the local labels $\mu$ and $\nu$ we will sometimes use a single-label notation $  \vket{Z_{\omega}} = 	\vket{X_\nu \otimes Y_\mu} $. In particular, we denote $  \vket{Z_{0}}  =  \vket{X_0} \otimes \vket{Y_0}$.

We highlight that we shall use the greek-labels ($\mu,\nu,\dots$) for sums that run over \emph{all} basis operators and Latin-labels ($i,j, \dots$) notation to run over just the \emph{trace-less} basis operators.

Consider a quantum channel $\E_{AB \rightarrow A' B'} : \B(\H_A \otimes \H_B) \rightarrow \B(\H_{A'} \otimes \H_{B'})$. We define a \emph{product channel} as one that takes the form
\begin{equation}
\E_{AB \rightarrow A' B'} = \E_{A \rightarrow A'} \otimes \E_{B \rightarrow B'},
\end{equation}
for channels $\E_{A \rightarrow A'} : \B(\H_A ) \rightarrow \B(\H_{A'})$ and $\E_{B \rightarrow B'} : \B(\H_B ) \rightarrow \B(\H_{B'})$. The choice of labeling of the output subsystems is for convenience, as a joint channel of the form $ \E_{A \rightarrow B'} \otimes \E_{B \rightarrow A'}$ can be cast in the above form simply by relabeling $A' \leftrightarrow B'$. A \emph{separable channel} is defined as a convex mixture of product channels, namely
\begin{equation}
\E_{AB \rightarrow A' B'} = \sum_k p_k\E^k_{A \rightarrow A'} \otimes \E^k_{B \rightarrow B'},
\end{equation}
for some distribution $p_k$ and local channels between $(A,A')$ and $(B,B')$. A channel that is not separable is defined to be \emph{non-separable.} 

\subsection{Quantum operations in the vectorized operator basis}
Using this notation we now give some useful quantum operations in the Liouville representation that we use through out this work with proofs following. Firstly, the channel to trace out ($\tr$) the system,
and a channel we define to prepare ($prep: prep(1)=\ident/d$) a new system in the maximally mixed state
\begin{equation}\label{eqn:vectorised-trace}
   \bm{tr} = \sqrt{d} \ \vbra{Z_0} \ \  \text{and} \ \   \bm{prep} = \vket{Z_0}/\sqrt{d}.
\end{equation}
A direct consequence of this is that the completely depolarizing channel $\D(\rho) := \ident/d$ is given by 
\begin{equation}\label{eqn:vectorised-depolarization}
    \bm{\D} = \bm{prep} \cdot  \bm{tr} = \vketbra{Z_0}{Z_0}.
\end{equation}
The identity channel ($id(\rho)=\rho$) also allows a very simple form in the Liouville representation: $\bm{id} = \ident^{\otimes 2}$. From these definitions we can build bipartite channels, such as the partial trace of subsystem $B$
\begin{equation}\label{eqn:vectorised-partial-trace}
    \bm{id}_A \otimes \bm{tr}_B = \sqrt{d_B} \ \bm{id}_A \otimes \vbra{Y_0}.
\end{equation}
where $\bm{id}_A$ is the Liouville representation of the identity channel on subsystem $A$. Similarly, combination of this with the preparation channel on $B$ leads to the complete depolarization channel for the $B$ subsystem
\begin{equation}\label{eqn:vectorised-depolar-B}
    \bm{id}_A \otimes \bm{\D}_B = \bm{id}_A \otimes ( \bm{prep}_B \cdot  \bm{tr}_B) = \bm{id}_A \otimes \vketbra{Y_0}{Y_0}.
\end{equation}
Finally for $d_A=d_B$ we can express the unitary operation, $SW\!AP$, that swaps the states of both subsystems compactly in the Liouville representation as
\begin{equation}\label{eqn:vectorised-swap}
    \bm{SW\!AP} = \sum_{\nu=0,\mu=0}^{d_A^2-1, d_B^2-1} \vketbra{X_\nu \otimes Y_\mu}{X_\mu \otimes Y_\nu}.
\end{equation}
Proofs for the preceding Liouville operators are now given:
\begin{proof}(of eqn. (\ref{eqn:vectorised-trace}))
   For the first part we have $ \tr[\rho_{AB}] = \tr[\ident \rho_{AB}] = \sqrt{d} \vbraket{Z_0}{\rho_{AB}}$ as $Z_0 = Z_0^\dagger$. We can then vectorize both sides and apply the definition of the Liouville representation of a channel $\ket{\tr[\boldify{\rho_{AB}}]} = \sqrt{d} \vbraket{Z_0}{\rho_{AB}}$ and $\bm{tr}\vket{\rho_{AB}} = \sqrt{d} \vbraket{Z_0}{\rho_{AB}}$. Therefore $\bm{tr} = \sqrt{d} \vbra{Z_0}$. For the second part, definitionally, $\ident / d = Z_0/\sqrt{d}$, $prep(1) = Z_0/\sqrt{d}$ and the vectorization of 1 leaves it unchanged $\vket{1}=1$. Therefore $ \vket{prep(1)} = \vket{Z_0}/\sqrt{d}$ and $\bm{prep} \vket{1} = \vket{Z_0}/\sqrt{d} \vket{1}$. As 1 is the only valid state of the trivial system, we read off $\bm{prep} = \vket{Z_0}/\sqrt{d}$ completing the proof.
\end{proof}

\begin{proof}(of eqn. (\ref{eqn:vectorised-depolarization}))
    We have $\boldify{\D}\vket{\rho}=\vket{\ident/d}=\vket{Z_0}/\sqrt{d}$. As $\vbraket{Z_0}{\rho}=1/\sqrt{d}$ for any quantum state $\rho$ we can write $\boldify{\D} = \vketbra{Z_0}{Z_0}$.
\end{proof}

\begin{proof}(of eqn. (\ref{eqn:vectorised-partial-trace}))
    This follows from eqn. (\ref{eqn:vectorised-trace}), with the identity channel on subsystem $A$.
\end{proof}

\begin{proof} (of $\bm{prep_B} = \bm{id}_A \otimes \vket{Y_0}/\sqrt{d_B}$)
    This follows from eqn. (\ref{eqn:vectorised-trace}), with the identity channel on subsystem $A$.
\end{proof}

\begin{proof}(of eqn. (\ref{eqn:vectorised-swap}))
    From definition, we can write any bipartite state in the form $\rho := \sum_{\nu,\mu} \lambda_{\nu \mu} \ X_\nu \otimes Y_\mu$. The $SW\!AP$ channel then acts on this state such that $ SW\!AP(\rho) := \sum_{\nu,\mu} \ \lambda_{\mu \nu} X_\nu \otimes Y_\mu$.
    Therefore, from inspection, the Liouville super operator of the channel is $\boldify{SW\!AP} = \sum_{\nu,\mu} \vketbra{X_\mu \otimes Y_\nu }{X_\nu \otimes Y_\mu}$.
\end{proof}

\section{Properties of subunitarity}\label{appendsubunit}

\subsection{Elementary properties of subunitarity}\label{append-sub-unit-properties}

\begin{lemma}\label{lemma:unitarity-zero-iff-depolar}
    Given a quantum channel $\E$, we have that $u(\E)=0$ if and only if $\E$ is a completely depolarizing channel. 
\end{lemma}
\begin{proof}
    We have that $u(\E)=0$  if and only if $\tr[T^\dagger T] =||T||_2^2=0 $ but this occurs if and only if $T=0$. Therefore the only possible non-zero data in the channel's Liouville representation is the $\mathbf{x}$ vector. This is a completely depolarizing channel to a fixed quantum state as required.
\end{proof}

\begin{theorem}\label{thm:local-subunitarity}
    For $\E_{A}(\rho):= \tr_{B}[ \E_{AB}(\rho \otimes \frac{\ident_B}{d_B})]$ we have $u(\E_A) = u_{A \to A}(\E_{AB})$, the unitarity $u$ of the local channel equal to the subunitarity $u_{A \to A}$ of the full channel.
\end{theorem}
\begin{proof}
    From definition the sub-unital block $T_{A\to A} = \vbra{X_i\otimes Y_0} \boldsymbol{\E_{AB}} \vket{X_j\otimes Y_0} = \tr[ X_i^{\dagger}\otimes \frac{\ident_B}{\sqrt{d_B}} \E_{AB}(X_j\otimes \frac{\ident_B}{\sqrt{d_B}})]  = \tr_{A}[ X_{i}^{\dagger} \tr_{B} [\E_{AB}(X_j\otimes \frac{\ident_B}{d_B})] = \tr_{A}[X_i^{\dagger} \E_{A}(X_j)]$, which gives the unital block $T$ of $\E_{A}$. As $T_{A \to A, \E_{AB}} = T_{\E_A}$ from definition $u_{A \to A}(\E_{AB})=u(\E_A)$. Similarly $u(\E_B) = u_{B \to B}(\E_{AB})$ for $\E_{B}(\rho):= \tr_{A}[ \E_{AB}( \frac{\ident_A}{d_A} \otimes \rho )]$.
\end{proof}

\begin{theorem}
    The unitarity of a channel $\cal{E}$ can be written as the weighted sum of its sub-unitarities
    \begin{equation}
        u(\cal{E}) = \frac{1}{d^2-1} \sum_{i,j = (A, B, AB)}   \frac{u_{i \to j}(\cal{E})}{\alpha_i},
    \end{equation}
    where $d=d_A . d_B$, $\alpha_A = 1/(d_A^2 - 1)$, $\alpha_B = 1/(d_B^2 - 1)$ and $\alpha_{AB} = \alpha_A\cdot\alpha_B$.
\end{theorem}
\begin{proof}
    This simply follows from block-matrix multiplication, giving $ \tr[T^\dagger T] = \sum_{n,m = (A, B, AB)} \tr[T_{n \to m}^\dagger T_{n \to m}]$. Therefore (see eqn. (\ref{eqn:unitarity-T})) the unitarity is $u(\cal{E}) = \frac{1}{d^2-1} \sum_{n,m = (A, B, AB)} \tr[T_{n \to m}^\dagger T_{n \to m}]$. Rearranging the dimensional constants (see eqn. (\ref{eqn:all-sub-unitarities})) completes the proof.
\end{proof}

\subsection{Properties of subunitarity for product channels}\label{append-sub-unit--product-properties}
\begin{lemma}
    For a product channel, $\E_A\otimes\E_B$, the sub unital block $T_{A \to A}= T_A \otimes \vketbra{Y_0}{Y_0}$ where $T_{A} := \sum_{i,j} \vketbra{\E_{A}(X_{j})}{X_i}$. Similarly $T_{B \to B}= \vketbra{X_0}{X_0} \otimes T_B$ where $T_{B} := \sum_{i,j} \vketbra{\E_{B}(Y_{j})}{Y_i}$.
\end{lemma}
\begin{proof}
    From definition, $T_{A \to A,ij}=\vbra{X_{i}} \otimes \vbra{Y_0} \boldify{\E_{A}} \otimes \boldify{\E_{B}} \vket{X_{j}} \otimes \vket{Y_0} =\vbra{X_{i}} \boldify{\E_{A}} \vket{X_{j}} \tr[\E_B(\ident/d_A)]$. For any trace preserving channel $\tr[\E(\ident/d)]=1$ so $ T_{A \to A}= \sum_{i,j} \vbra{X_{i}} \boldify{\E_{A}} \vket{X_{j}} \vketbra{X_i}{X_j} \otimes \vketbra{Y_0}{Y_0}$. The proof for $T_{B \to B}$ follows similarly.
\end{proof}

\begin{lemma}\label{Tblocks-productchannel}
    For a product channel, $\E_A\otimes\E_B$, we have $T_{AB \to AB} = T_A \otimes T_B$.
\end{lemma}
\begin{proof}
   From definition, $T_{AB \to AB} = \sum_{ij}^{d_A^2 -1} \sum_{nm}^{d_B^2 -1} \vbra{X_{i}} \otimes \vbra{Y_{n}} \boldify{\E_{A}} \otimes \boldify{\E_{B}} \vket{X_{j}} \otimes \vket{Y_{m}} \vketbra{X_i}{X_j} \otimes \vketbra{ Y_n}{Y_m} = \sum_{ij}^{d_A^2 -1} \sum_{nm}^{d_B^2 -1} \vketbra{\E_{A}(X_{j})}{X_i} \otimes \vketbra{\E_{B}(Y_{m})}{Y_n} = T_{A}\otimes T_{B}.$
\end{proof}

\begin{theorem}\label{product-lemma}
    For a product channel, $\E = \E_A\otimes\E_B$, $u_{AB \to AB}(\E)=u_{A \to A}(\E) \cdot u_{B \to B}(\E)$.
\end{theorem}
\begin{proof}
    From Lemma \ref{Tblocks-productchannel} we can write $u_{AB \to AB}(\E)  = \alpha_{A}\cdot\alpha_{B} \tr[T_{A}^\dagger \otimes T_{B}^\dagger T_{A} \otimes T_{B}] = \alpha_{A}\cdot\alpha_{B} \tr[T_{A}^\dagger T_{A}] \tr[ T_{B}^\dagger T_{B}] = u(\E_A)u(\E_B)$. As $u(\E_A)=u_{A \to A}(\E)$ and $u(\E_B)=u_{B \to B}(\E)$ for any channel this completes the proof.
\end{proof}

\begin{corollary}
    The correlated unitarity $u_c$ of a product channel $\E_A \otimes \E_B$ is $u_{c}(\E_A \otimes \E_B) =0$.
\end{corollary}
\begin{proof}
    This follows directly from Theorem \ref{product-lemma}.
\end{proof}

\begin{lemma}
    The sub-unitarity $u_{A \to AB}(\cal{E}_A \otimes \cal{E}_B)$ for a bipartite product channel $\cal{E}_A \otimes \cal{E}_B$, decomposes as
    \begin{equation}
        u_{A \to AB}(\cal{E}_A \otimes \cal{E}_B) = u_{A \to A}(\cal{E}_A \otimes \cal{E}_B) x_{B},
    \end{equation}
    where $x_{B} := \mathbf{x}_{B \to B}^\dagger \mathbf{x}_{B \to B}$ for the non-unital vector of the subsystem $B$ of the channel $\E_B$.
\end{lemma}
\begin{proof}
    From the definition of $u_{A \to AB}$ we have
    \begin{equation}
        \begin{split}
            u_{A \to AB}(\cal{E}_A \otimes \cal{E}_B) &= \alpha_A \sum_{k,j,n=1}^{(d_A^2-1)(d_B^2-1)} \vbra{X_j \otimes Y_n}  \mathbfcal{E} \vket{X_k \otimes Y_0} \vbra{X_k \otimes Y_0} \mathbfcal{E}^\dagger \vket{X_j \otimes Y_n}, \\
            &= \alpha_A \sum_{k,j,n=1}^{(d_A^2-1)(d_B^2-1)} \vbra{X_k} \mathbfcal{E}_A^\dagger \vket{X_j} \vbra{X_j} \mathbfcal{E}_A \vket{X_k} \vbra{Y_0} \mathbfcal{E}_B^\dagger \vket{Y_n}  \vbra{Y_n}  \mathbfcal{E}_B \vket{Y_0}, \\
            &=  u_{A \to A}(\cal{E}_A \otimes \cal{E}_B) \sum_{n=1}^{d_B^2-1} \vbra{Y_0} \mathbfcal{E}_B^\dagger \vket{Y_n}  \vbra{Y_n}  \mathbfcal{E}_B \vket{Y_0}, \\
            &=  u_{A \to A}(\cal{E}_A \otimes \cal{E}_B) x_{B}
        \end{split}
    \end{equation}
    which completes the proof.
\end{proof}
Swapping the subsystem labels we also have $ u_{B \to AB}(\cal{E}_A \otimes \cal{E}_B) = u_{B \to B}(\cal{E}_A \otimes \cal{E}_B) x_{A}$,  where $x_{A} := \mathbf{x}_{A \to A}^\dagger \mathbf{x}_{A \to A}$ for the non-unital vector of the subsystem $A$ of the channel.

\subsection{Properties of subunitarity for separable channels}\label{append-sub-unit-seperable-properties}
\begin{lemma}\label{lemma:u-aba-zero-for-seperable}
    The sub-unitarity $u_{AB \to A}(\E)$ for a  bipartite separable channel  $\E := \sum_i^r p_i \E_{A,i} \otimes \E_{B,i}$ is zero.
\end{lemma}
\begin{proof}
    From the definition of $u_{AB \to A}$ we have
    \begin{equation}
        \begin{split}
            u_{AB \to A}(\E) &= \alpha_A \alpha_B \tr[T_{AB \to A}^\dagger T_{AB \to A}], \\
            &= \alpha_A \alpha_B \sum_{k,j,n=1}^{(d_A^2-1)(d_B^2-1)} \vbra{X_j \otimes Y_n}  \mathbfcal{E}^\dagger \vket{X_k \otimes Y_0} \vbra{X_k \otimes Y_0} \mathbfcal{E} \vket{X_j \otimes Y_n}, \\
            &= \alpha_A \alpha_B \sum_{k,j,n=1}^{(d_A^2-1)(d_B^2-1)} \sum_{i,j}^r p_i p_j \vbra{X_j} \mathbfcal{E}_{A,i}^\dagger \vket{X_k} \vbra{X_k} \mathbfcal{E}_{A,j} \vket{X_j}   \vbra{Y_n}  \mathbfcal{E}_{B,i}^\dagger \vket{Y_0} \vbra{Y_0} \mathbfcal{E}_{B,j} \vket{Y_n}.
        \end{split}
    \end{equation}
    For the channel to be trace preserving we must have $\vbra{Y_0} \mathbfcal{E}_{B,j} \vket{Y_n} = 0$ for all $n$ \& $j$. Therefore $u_{AB \to A}(\E)=0$.
\end{proof}
Additionally, swapping the subsystem labels, $u_{AB \to B}(\E) = 0$ for any separable bipartite channel $\E$.

\begin{lemma}\label{lemma:u-ab-zero-for-seperable}
    The sub-unitarity $u_{A \to B}(\E)$ for a  bipartite separable channel  $\E := \sum_i^r p_i \E_{A,i} \otimes \E_{B,i}$ is zero.
\end{lemma}
\begin{proof}
    From definition
    \begin{equation}
        \begin{split}
            u_{A \to B}(\E) &= \alpha_A\tr[T_{A \to B}^\dagger T_{A \to B}], \\
            &= \alpha_A \sum_{k,j=1}^{(d_A^2-1)(d_B^2-1)} \vbra{X_j \otimes Y_0}  \mathbfcal{E}^\dagger \vket{X_0 \otimes Y_k} \vbra{X_0 \otimes Y_k} \mathbfcal{E} \vket{X_j \otimes Y_0}, \\
            &= \alpha_A \sum_{k,j=1}^{(d_A^2-1)(d_B^2-1)} \sum_{i,j}^r p_i p_j \vbra{X_j} \mathbfcal{E}_{A,i}^\dagger \vket{X_0} \vbra{X_0} \mathbfcal{E}_{A,j} \vket{X_j}   \vbra{Y_0}  \mathbfcal{E}_{B,i}^\dagger \vket{Y_k} \vbra{Y_k} \mathbfcal{E}_{B,j} \vket{Y_0}.
        \end{split}
    \end{equation}
    For the channel to be trace preserving we must have $\vbra{X_0} \mathbfcal{E}_{A,j} \vket{X_j} = 0$ for all $j$. Therefore $u_{A \to B}(\E)=0$.
\end{proof}
Additionally, swapping the subsystem labels, $u_{B \to A}(\E) = 0$ for any separable bipartite channel $\E$.

\begin{lemma}
   For a unital bipartite separable channel  $\E := \sum_i^r p_i \E_{A,i} \otimes \E_{B,i}$ where $\E_{X,i}$ are local unital channels, the sub-unitarity $u_{A \to AB}(\E)$  is zero.
\end{lemma}
\begin{proof}
    From the definition of $u_{A \to AB}$ we have
    \begin{equation}
        \begin{split}
            u_{A \to AB}(\E) &= \alpha_A \tr[T_{A \to AB}^\dagger T_{A \to AB}], \\
            &= \alpha_A \sum_{k,j,n=1}^{(d_A^2-1)(d_B^2-1)}  \vbra{X_k \otimes Y_0} \mathbfcal{E}^\dagger \vket{X_j \otimes Y_n} \vbra{X_j \otimes Y_n}  \mathbfcal{E} \vket{X_k \otimes Y_0}, \\
            &= \alpha_A  \sum_{k,j,n=1}^{(d_A^2-1)(d_B^2-1)} \sum_{i,j}^r p_i p_j \vbra{X_k} \mathbfcal{E}_{A,j}^\dagger \vket{X_j} \vbra{X_j} \mathbfcal{E}_{A,i} \vket{X_k}  \vbra{Y_0} \mathbfcal{E}_{B,j}^\dagger \vket{Y_n}  \vbra{Y_n}  \mathbfcal{E}_{B,i} \vket{Y_0} .
        \end{split}
    \end{equation}
    For the channel to be unital we must have $\vbra{Y_n} \mathbfcal{E}_{B,i} \vket{Y_0} = 0$ for all $n$. Therefore $u_{A \to AB}(\E)=0$.
\end{proof}
Additionally, swapping the subsystem labels, $u_{B \to AB}(\E) = 0$ for any unital separable bipartite channel $\E$.

\subsection{Properties of subunitarity for Pauli channels}\label{properties-of-pauli-channels-section}    
Consider the Pauli operators $P_{\alpha}$ acting on $n$ qubits. These will form a complete orthonormal basis (so that normalization will be included in the definition of $P_{\alpha}$) so as $\Tr(P_{\alpha} P_{\beta})  = \delta_{\alpha,\beta}$ and $P_{\alpha}\hc = P_{\alpha}$. We will also label $P_{0} : = \ident /\sqrt{2^n}$, the identity operator. Moreover, for simplicity we consider bipartite systems formed of $A$ and $B$ each of $n$ qubits, so they have dimensions $d_{A} = d_B = 2^n$.

\begin{lemma}
Let $\E(\rho ) = \sum_{i} p_i P_i\rho P_i$ be a Pauli channel with $\sum_i p_i = d$  where the Pauli operators are normalized so that $\Tr(P_i P_j) = \delta_{ij}$ with $\E$ acting on a system of dimension $d$. Then it follows that $\boldsymbol{\E}$, its Liouville matrix has entries
\begin{equation}
    \< \boldsymbol{P_{j} }|\boldsymbol{\E} |\boldsymbol{P_{k} }\> = \delta_{jk}  \sum_{i} (-1)^{\eta(P_i,P_k)} p_i 
\end{equation}
where $\eta(P_i, P_k)$ is 0 if $P_i$ and $P_k$ commute and 1 if they anti-commute.
The unitarity of $\E$ is given by \footnote{Alternatively this formula can be calculated from the definition involving Haar measure}
\begin{align}
    u(\E)  =  \frac{1}{(d^2 - 1)} \left( \left(\sum_{i} p_i^2 \right)- 1 \right)
\end{align}
\label{lem:PauliUnitarity}
\end{lemma} 
\begin{proof}
    Check directly $\<\boldsymbol{P_j} |\boldsymbol{\E}|\boldsymbol{P_{k} }\> = \<\boldsymbol{P_{j} }|\boldsymbol{\E(P_k)}\> = \Tr(P_j \E(P_k)) = \sum_i p_i \Tr(P_j P_i P_k P_i) = \frac{1}{d} \sum_i p_i (-1)^{\eta(P_i, P_k)} \Tr(P_jP_k)  =\frac{1}{d} \delta_{jk} \sum_i p_i (-1)^{\eta(P_i, P_k)}$. This is a diagonal Liouville matrix, and the unitarity is determined in terms of its non-unital block $T_{\E}$ as 
    \begin{align}
        u(\E) &= \frac{1}{d^2-1} \Tr(T_{\E}\hc T_{\E}) \\
        & = \frac{1}{d^2-1} \sum_{j\neq 0} \<\boldsymbol{P_{j} }|\boldsymbol{\E}|\boldsymbol{P_{j} }\>^2.
    \end{align}
    Note that $\<\boldsymbol{P_{0} }| \boldsymbol{\E} |\boldsymbol{P_{0} }\> = \frac{1}{d}\sum_{i} p_i = 1$. Notice the orthogonality relation $\sum_{j} (-1)^{\eta(P_i,P_j)} (-1)^{\eta{(P_{i'} P_j})} = d^2\delta_{ii'}$ so that
    \begin{align}
        \sum_{j} \<\boldsymbol{P_{j} }|\boldsymbol{\E}|\boldsymbol{P_{j} }\>^2 &=\frac{1}{d^2} \sum_{j, i, i'}  p_i p_{i'} (-1)^{\eta(P_i, P_j)} (-1)^{\eta(P_{i'}, P_j)}\\
        & =  \sum_{i} p_i^2.
    \end{align}
    Therefore, we have
    \begin{align}
        u(\E) &= \frac{1}{d^2-1}\sum_{j\neq 0}|\<\boldsymbol{P_{j} }|\boldsymbol{\E}|\boldsymbol{P_{j} }\>|^2\\
        &=  \frac{1}{(d^2 - 1)} \left( \left(\sum_{i} p_i^2 \right)- 1 \right).
    \end{align}
\end{proof}
A bipartite Pauli channel  on two $n$ qubits systems will take the following form
\begin{equation}
    \E(\rho_{AB} ) = \sum_{\alpha,\beta} p_{\alpha,\beta}  P_{\alpha}\otimes P_{\beta} \, \rho_{AB} \,  P_{\alpha}\otimes P_{\beta}
\end{equation}
and trace preserving condition requires $\sum_{\alpha,\beta} p_{\alpha,\beta}  = 4^n$.  We denote $d = d_A = d_B = 2^n$. The Liouville representation, with respect to a Pauli basis will be a diagonal matrix.
The local channel at $A$
\begin{align}
    \E_{A} (\rho_A) : &= \Tr_{B} \E(\rho_A\otimes \ident/d) \\
    & = \sum_{\alpha}  q_{\alpha,0} P_{\alpha}\rho_A P_{\alpha} 
\end{align}
where the $q_{\alpha,0} := \frac{1}{d}\sum_{\beta} p_{\alpha,\beta}$.
Similarly at $B$:
\begin{align}
    \E_{B} (\rho_B) : &= \Tr_{B} \E(\ident/d \otimes \rho_B) \\
    & = \sum_{\alpha}  q_{0,\beta} P_{\beta}\rho_B P_{\beta} 
\end{align}
where the $q_{0,\beta} := \frac{1}{d} \sum_{\alpha} p_{\alpha,\beta}$.
The subunitarities at $A$ and $B$ are given by  $u_{A\to A} = u(\E_A)$ and $u_{B\to B} = u(\E_B)$. Therefore we get the following result.
\begin{lemma}
    Let $d = 2^n$, the dimension of system A and respectively system B, then we have that
    \begin{align}
        u_{A\to A}& = \frac{1}{d^2-1}\left( \sum_{\alpha}  q_{\alpha,0}^2 - 1\right),\\
        u_{B\to B } &= \frac{1}{d^2-1}\left( \sum_{\beta} q_{0,\beta}^2 -1 \right),\\
        u_{AB\to AB} & = \frac{d^2 + 1}{d^{2}-1}u(\E) -\frac{1}{d^2-1} (u_{A\to A} + u_{B\to B}).
    \end{align}		
\end{lemma}
\begin{proof}
    The relations for $u_{A\to A}$ and $u_{B\to B}$ follow directly from Lemma \ref{lem:PauliUnitarity}. The relation for $u_{AB\to AB} $ follows from the fact that the Liouville representation of $\E$ is diagonal so that $T_{\E} = T_{A\to A} \oplus T_{AB\to AB} \oplus T_{B\to B}$ and thus 
    \begin{align*}
        \Tr(T_{\E} \hc T_{\E}) = & \Tr(T_{A\to A}\hc T_{A\to A}) + \Tr(T_{B\to B}\hc T_{B\to B}) + \\ &+ \Tr(T_{AB\to AB}\hc T_{AB\to AB})
    \end{align*}
    In terms of the unitarities, $u(\E)  = \frac{1}{4^{2n}-1}\Tr(T_{\E}\hc T_{\E})$ and  $u_{AB\to AB} = \frac{1}{(2^{2n}-1)^2} \Tr(T_{AB\to AB}\hc T_{AB\to AB})$ then the above is equivalent to
    \begin{equation}
            u_{AB\to AB}  = \frac{4^{2n}-1}{(2^{(2n)} - 1)^2} u(\E)  -  \frac{1}{2^{2n} -1} \left(u_{A\to A} + u_{B\to B} \right).
    \end{equation}
\end{proof}

\begin{lemma}
The correlated unitarity for Pauli noise channel on a bipartite system $AB$ with $\rm{dim}(A) = \rm{dim}(B) = d = 2^n$ is given by
\begin{align}
    u_{c} &=\frac{1}{(d^2-1)^2}\left( \sum_{\alpha,\beta} p_{\alpha,\beta}^2  - (\sum_{\alpha} q_{\alpha,0}^2) \sum_{\beta} q_{0,\beta}^2 \right).
\end{align}
\end{lemma}
\begin{proof}
    Directly from above.
\end{proof}

\section{Properties of correlated unitarity}\label{append:u_c-properties}

\subsection{Comparison of correlated unitarity with norm measures}\label{append:u-c-compared-to-norm}
We can compare the choice of definition for correlated unitarity with a norm, which sheds light on its structure and limitations. Consider the Hilbert-Schmidt norm expression
\begin{equation}
    \Delta^2 := ||T_{AB} - T_A \otimes T_B||^2
\end{equation}
where $T_{AB} \equiv T_{AB\rightarrow AB}$ and similarly for $T_A$ and $T_B$. As this is a norm we have $\Delta = 0$ if and only if $T_{AB} = T_A \otimes T_B$, namely if and only if the channel is a product channel. We can expand this expression in terms of the Hilbert-Schmidt inner product to obtain
\begin{align}
    \Delta^2 &= \< T_{AB} - T_A\otimes T_B,  T_{AB} - T_A\otimes T_B\> \nonumber \\
    &= \< T_{AB}, T_{AB}\> + \<T_A\otimes T_B,  T_A\otimes T_B\> -\<T_{AB} , T_A \otimes T_B\> - \<T_A \otimes T_B, T_{AB}\> \nonumber\\
    &= ||T_{AB}||^2 + ||T_A||^2 ||T_B||^2 - 2 Re\left[ \<T_{AB},T_A \otimes T_B\> \right] \nonumber \\
    &= ||T_{AB}||^2 + ||T_A||^2 ||T_B||^2 - 2  ||T_{AB}|| \ ||T_A|| \ ||T_B|| \cos \theta  \nonumber \\
    \Delta^2 &= t_{AB}^2 +t_A^2 t_B^2 - 2 t_{AB} t_A t_B \cos \theta  \nonumber,
\end{align}
where we have defined an angular variable $\theta$ via the inner product between $T_{AB}$ and $T_A \otimes T_B$ and replaced the norm values with $t_{AB},t_A, t_B$ in the obvious way. Now the correlated unitarity is given by $u_c = \alpha_{AB} (t_{AB}^2 - t_A^2 t_B^2),$ with the dimensional prefactor $\alpha_{AB} =\frac{1}{(d_A^2 -1)(d_B^2 -1)}$. 
Substituting for $t_{AB}$ into $\Delta^2$ we have that
\begin{equation}
    \Delta^2 = \frac{u_c}{\alpha_{AB}} + 2 (t_At_B)^2 - 2 \sqrt{\frac{u_c}{\alpha_{AB}} + (t_At_B)^2} (t_A t_B) \cos \theta.
\end{equation}
This implies a few things. Firstly, for $u_c =0$ we have
\begin{equation}
    \Delta^2 = 2(t_At_B)^2 (1-\cos \theta),
\end{equation}
and so we see that $u_c$ vanishing does not imply a product channel unless one of the $t_A, t_B$ vanishes or if $\theta =0$. The expression also implies that $\theta$ is an independent parameter that will in general vary the norm distance. Note that the benchmarking protocol gives us both $(t_At_B)$ and $u_c$ but does not give us $\theta$. Therefore our existing benchmarking does not return enough to determine norm distance measure.

The above highlights relevant data at quadratic order that our approach is not sensitive to, but note that the $\cos \theta$ term is bounded and so it still is the case that $u_c$ is acting as a ``distance'' from being a product channel. Specifically, we have
\begin{equation}
    \frac{u_c}{\alpha_{AB}} + 2 (t_At_B)^2 - 2 \sqrt{ \frac{u_c}{\alpha_{AB}} + (t_At_B)^2} (t_A t_B)  \le \Delta^2 \text{ \ and \ } \Delta^2 \le  \frac{u_c}{\alpha_{AB}} +2 (t_At_B)^2 + 2 \sqrt{ \frac{u_c}{\alpha_{AB}} + (t_At_B)^2} (t_A t_B) 
\end{equation}
This implies that estimating $u_c$ and $t_At_B$ allows us to estimate the norm distance $\Delta$.

\subsection{Operational interpretation of $u_c$}\label{operational-uc}
\begin{proof} (Of Eqn ~\ref{eqn:corrunitarityfunc})
    Using the definition of correlated unitarity,
    \begin{align} 
        u_c(\E_{AB})&=\alpha_A \alpha_B \left( \Tr(T^{\dagger}_{AB\to AB} T_{AB\to AB}) - \Tr(T^{\dagger}_{A\to A} T_{A\to A}) \Tr(T^{\dagger}_{B\to B} T_{B\to B})\right)\\
        & = \alpha_A \alpha_B \left( \sum_{n,m,a,b \neq 0} |\vbra{ X_n\otimes Y_m} T_{AB\to AB}\vket{X_a\otimes Y_b}|^2 -  |\vbra{X_n}T_{A \to A}\vket{X_a}|^2 |\vbra{Y_m} T_{B \to B}\vket{Y_b}|^2\right)\\
        &=\alpha_A \alpha_B \left( \sum_{n,m,a,b \neq 0} |\Tr(X_n\otimes Y_m \E_{AB}(X_a\otimes Y_b))|^2 - |\Tr(X_n \E_{A}(X_a)) |^2 |\Tr(Y_m \E_{B}(Y_b))|^2 \right).
    \end{align}
\end{proof}

\subsection{$u_c$ is invariant under local unitaries}
\begin{corollary}\label{cor:local-subs-invariant}
    The local subunitarities of any channel $u_{A \to A}(\E)$ \& $u_{B \to B}(\E)$ are invariant under local unitaries.
\end{corollary}
\begin{proof}
    From definition $u_c := u_{AB \to AB} - u_{A \to A} \cdot u_{B \to B}$.  It is easy to show that each term is invariant under local unitaries.
    
    Firstly, the local subunitarities of any channel $u_{A \to A}(\E)$ \& $u_{B \to B}(\E)$ are invariant under local unitaries. This is because from Theorem \ref{thm:local-subunitarity} we have that $u_{A \to A}(\E)=u(\E_A)$, therefore sandwiching with any product unitaries $\E' = \cal{U}_{1,A} \otimes \cal{U}_{1,B} \ \circ \ \E \ \circ \ \cal{U}_{2,A} \otimes \cal{U}_{2,B}$ we have $u_{A \to A}(\E')=u(\cal{U}_{1,A} \circ \E_A \circ \cal{U}_{2,A})$. From the invariance of unitarity under unitaries \cite{wallman2015estimating}, $u_{A \to A}(\E') = u(\E_A)$.
    
    It remains to prove that $u_{AB}$ is invariant. We can write the Liouville representation of any product unitary in the our basis as $ \boldify{\cal{U}_{i,A} \otimes \cal{U}_{i,B}} = (1 \oplus \cal{O}_{i,A}) \otimes (1 \oplus \cal{O}_{i,B})$ where $\cal{O}_{i,X}$ are unitary matrices of dimension $(d_X^2 -1) \cross (d_X^2 -1)$ obeying $\cal{O}_{i,X} \cal{O}_{i,X}^\dagger = \ident_{T_{X}}$. Product channels have the additional property that $T_{AB,\cal{U}_i} = T_{A,\cal{U}_i} \otimes T_{B,\cal{U}_i} = \cal{O}_{i,A} \otimes \cal{O}_{i,B}$.
 
    We define a channel $\E' = \cal{U}_{1,A} \otimes \cal{U}_{1,B} \ \circ \ \E \ \circ \ \cal{U}_{2,A} \otimes \cal{U}_{2,B}$: namely, the channel with product unitaries before and after. The product unitaries will have block diagonal unital blocks which can be seen from considering their only non-zero subunitarities are $u_{A \to A}$, $u_{B \to B}$, \& $u_{AB \to AB}$. Because of this simple structure the sub-unital block $T_{AB}$ of $\E'$ is
    \begin{equation}\label{eqn:tab-product-sandwich}
        T_{AB,\E'} = T_{AB,\cal{U}_1} T_{AB,\E} T_{AB,\cal{U}_2} = \cal{O}_{1,A} \otimes \cal{O}_{1,B} T_{AB,\E} \cal{O}_{2,A} \otimes \cal{O}_{2,B}.
    \end{equation}
    We can now calculate the required subunitarity $ u_{AB}(\E') = \alpha_{AB} \tr[T_{AB, \E'}^\dagger T_{AB, \E'}]$, and from the cyclical properties of the trace,
    \begin{equation}
        \begin{split}
            u_{AB}(\E') &= \alpha_{AB} \tr[ T_{AB,\E}^\dagger  T_{AB,\E} \cal{O}_{2,A}^\dagger \otimes \cal{O}_{2,B}^\dagger \cal{O}_{2,A} \otimes \cal{O}_{2,B}], \\
            &= \alpha_{AB} \tr[ T_{AB,\E}^\dagger  T_{AB,\E}] = u_{AB}(\E).
        \end{split}
    \end{equation}
 This implies $u_c$ is invariant under local unitarities.
\end{proof}

\subsection{Maximal value of correlated unitarity}\label{section:uc-for-swap-channel}

It is readily seen that the $SW\!AP$ channel has correlated unitarity,
\begin{equation}
    u_c(SW\!AP) = 1.
\end{equation}
This follows since, from Equation (\ref{eqn:vectorised-swap}) we have $\boldify{SW\!AP} = \sum_{\nu,\mu} \vketbra{X_\mu \otimes Y_\nu }{X_\nu \otimes Y_\mu}$. In our basis, this makes the unital block $T$ a matrix with 1 along the minor diagonal and zero everywhere else. We can then simply read off that $u_{AB \to AB}=u_{A \to B}=u_{B \to A}=1$ and all other subunitarities are zero. Correlated unitarity is then $u_c = u_{AB \to AB} - u_{A \to A} \cdot u_{B \to B} = 1$. The following shows the converse, that if the sub-unitarities for $AB\rightarrow AB$, $A\leftrightarrow B$ are maximized then the channel must be a $SW\!AP$ channel, modulo local changes of basis.

\begin{lemma}\label{lemma:swap-up-to-local-unitaries}
    Any channel $\E$ with $u_{AB \to AB}(\E)=u_{A \to B}(\E)=u_{B \to A}(\E)=1$ is equivalent to the $SW\!AP$ channel up to local unitaries.
\end{lemma}
\begin{proof}
    From Theorem \ref{thm:unitarityexpnsion} under the given conditions the channel is unitary, and all other subunitarities are zero. We can use that $u_{A \to B}(\E) = u_{A \to A}(SW\!AP \circ \E) = 1$ and similarly $u_{B \to B}(SW\!AP \circ \E) = 1$. Since the unitarity equals $1$ only for a unitary we deduce that $SW\!AP \circ \E$ must be a product channel $\U_A \otimes \U_B$ of local unitaries on each subsystem. Since $SW\!AP^2=id$, this implies that $\E = SW\!AP \circ \U_A \otimes \U_B$.
\end{proof}

\subsection{Proof of $u_c$ as Witness of non-separability}\label{section:witness-of-non-sep}
The proof of the upper bound on separable channels turns out to be non-trivial, and relies on bounds on the inner product of $T$--matrices for quantum channels. We first establish basic ingredients we need for the analysis.
\begin{lemma}\label{lemma:liouivlle-to-choi}
	For a channel $\E:$ $\B(\H) \to \B(\H')$ the Choi–Jamiołkowski state can be expressed in the $X_\nu$ basis as
	\begin{equation}
		\mathcal{J}(\E) = \frac{1}{d} \sum_{\nu=0}^{d^2} \cal{E}(X_\nu) \otimes X_\nu^*,
	\end{equation}
	with the complete orthonormal basis of both $\H$ \& $\H'$ as $X_\mu = (X_0=\ident/\sqrt{d}, X_i)$ with $\dim(\H)=\dim(\H')=d$.
\end{lemma}
\begin{proof}
	$ \mathcal{J}(\E):=\cal{E} \otimes id (\frac{1}{d} \vketbra{\ident_{d}}{\ident_{d}})$, but one can directly show  $\vketbra{\ident_{d}}{\ident_{d}} = \sum_{\nu} X_\nu \otimes X_{\nu}^{*}$. This follows from the fact that $\Tr( X_{\nu}\hc\otimes X_\mu\hc  \vketbra{\ident_{d}}{\ident_{d}}) = \vbraket{\ident}{X_{\nu} \hc  X_\mu^{*}} = \Tr X_{\nu}\hc X_{\mu}^{*} $ and therefore $\vketbra{\ident_{d}}{\ident_{d}} = \sum_{\mu,\nu}\Tr(X_\nu\hc X_\mu^{*})  X_\nu\otimes X_{\mu} $. However $ \sum_\mu \Tr(X_\nu\hc X_\mu^{*}) X_\mu =  \sum_\mu \Tr( X_\mu\hc X_\nu^{*}) X_\mu= X_\nu^{*} $ since $\Tr(A^T)= \Tr(A) $, and the result follows.
\end{proof}
We now have the following estimates.
\begin{lemma}\label{lemma:range-of-t-inner-product}
	Given two channels $\E_{1}$ and $\E_{2}$ with unital blocks in the Liouville representation $T_{1}$ and $T_2$, we have
	\begin{equation}
	    -d\leq \<T_{1}, T_2\> \leq d^2 - 1,
    \end{equation}
    where $d$ is the dimension of the Hilbert space.
\end{lemma}
We shall use this lemma to establish the upper bound on correlated unitarity for separable channels. However, we conjecture a stronger result that for any two quantum channels $\E_1$, $\E_2$ that $ \<T_1 , T_2\> \ge -1$. This, for example implies the bound for optimal inversion of the coherence vector of a quantum state \cite{byrd2003characterization,rungta2001universal} as a special case. The analyse to establish this sharper bound appears to be non-trivial. Since it is not essential to our work we leave it as an open problem. We do, however, establish this lower beyond for a subset of channels (see Lemma \ref{lemma:t-lower-bound-qubits} below).
\begin{proof}
    In the Choi representation we have 
    \begin{equation}
        \mathcal{J}(\E_1) = \frac{1}{d}\sum_{\mu}^{d^2}  \E_1(X_\mu) \otimes X_\mu^{*} \ \ \ \ \text{and} \ \ \ \
        \mathcal{J}(\E_2) = \frac{1}{d}\sum_{\mu}^{d^2}  \E_2(X_\mu) \otimes X_\mu^{*}\\
    \end{equation}
    with $X_\mu = (X_0=\ident/\sqrt{d}, X_i)$.
    Therefore we have that
    \begin{equation}
        \Tr( \mathcal{J}(\E_1)^{\dagger} \mathcal{J}(\E_2) ) = \frac{1}{d^2} \sum_{\mu,\nu}^{d^2} \Tr (\E_1(X_\mu)^{\dagger} \E_2(X_\nu)) \Tr(X_\mu^{T} X_\nu^{*}).
    \end{equation}
    Since Choi matrices are positive semidefinite, then so is the above quantity. Furthermore, $\Tr(X_\mu^{T} X_\nu^{*})  = \delta_{\mu\nu}$ and so
    \begin{equation}
        \Tr( \mathcal{J}(\E_1)^{\dagger} \mathcal{J}(\E_2) ) = \frac{1}{d^2}\sum_{\mu}^{d^2} \Tr (\E_1(X_\mu)^{\dagger} \E_2(X_\mu)) \geq 0,
    \end{equation}
    and therefore we have
    \begin{equation}
        \sum_{\mu}^{d^2}  \vbraket{\E_1(X_\mu)}{\E_2(X_\mu)} = \sum_{\mu}^{d^2} \Tr (\E_1(X_\mu)^{\dagger} \E_2(X_\mu)) \geq 0.
    \end{equation}
    Now we look at $\<T_1, T_2\>  = \Tr (T_1^{\dagger} T_2)$ and expand with respect to same basis. 
    \begin{equation}
        \<T_1, T_2\> = \sum_{i=1}^{d^2-1}  \vbra{X_i} T_1^\dagger T_2 \vket{X_i} = \sum_{i=1}^{d^2-1}  \vbraket{\E_1(X_i)}{\E_2(X_i)}  = \sum_{\mu}^{d^2}  \vbraket{\E_1(X_\mu)}{\E_2(X_\mu)} - \vbraket{\E_1(X_0)}{\E_2(X_0)}.
    \end{equation}
    Then it follows that 
    \begin{equation}
        \<T_1, T_2\> \geq - \vbraket{\E_1(X_0)}{\E_2(X_0)}.
    \end{equation}
    However,
    \begin{equation}
        |\vbraket{\E_1(X_0)}{\E_2(X_0)}|^2 \leq \vbraket{\E_1(X_0)}{\E_1(X_0)}\vbraket{\E_2(X_0)}{\E_2(X_0)}.
    \end{equation} 
    and since $\vbraket{\E_i\left(\frac{\ident}{d}\right)}{\E_i\left(\frac{\ident}{d}\right)} \le 1$, we deduce that
    \begin{equation}
        |\vbraket{\E_1\left(\frac{\ident}{\sqrt{d}}\right)}{\E_2\left(\frac{\ident}{\sqrt{d}}\right)}| \leq d,
    \end{equation} 
    and so we obtain the lower bound of
    \begin{equation}
        -d \le \<T_1, T_2\>.
    \end{equation}
    The upper bound follows directly form Holder's inequality
    \begin{equation}
    \<T_1, T_2 \> \leq  ||T_1||_{\infty} ||T_2||_1 \leq (d^2-1)
    \end{equation}
    where we have used in the above that the eigenvalues of $T_1$ and $T_2$ have modulus at most 1, and their rank is at most $d^2-1$.
\end{proof}
We also have the following lower bound on the inner product of two $T$-matrices for subsets of quantum channels.
\begin{lemma}\label{lemma:t-lower-bound-qubits}
    Let $\E_1$ and $\E_2$ be two quantum channels. If we have that either
    \begin{enumerate}
        \item One of the channels is unital,
        \item The channels are arbitrary $d=2$ qubit channels,
    \end{enumerate}
    then it follows that $-1 \le \<T_1, T_2\> \le d^2-1$.
\end{lemma}
The proof of this is as follows.
\begin{proof}
    If one of the channels, $\E_1$ say, is unital then
    \begin{equation}
        \<T_1, T_2\> \geq - \vbraket{\E_1(X_0)}{\E_2(X_0)} = -\vbraket{X_0}{\E_2(X_0)} - \vbraket{X_0 }{X_0} = -1,
    \end{equation}
    where we use the orthonormality $\vbraket{X_0}{X_i}$ for all $i = 1, \dots d^2-1$ and the fact that if $\E_1$ is unital then $\E_1(X_0) = X_0$.

    Now suppose that both $\E_1$ and $\E_2$ are qubit channels. Given any qubit channel $\E$, the corresponding Choi state take the form
    \begin{align}
        \mathcal{J}(\E) &= \frac{1}{4} ( \I + \mathbf{x} \cdot \boldsymbol{\sigma} \otimes \I + \sum_{i,j} T_{ij} \sigma_i \otimes \sigma_j),
    \end{align}
    where $\{\sigma_i\}$ are the Pauli matrices. As shown in \cite{horodecki1996information} it is possible to perform local unitary changes $U_A \otimes U_B$ of basis so that
    \begin{align}
        \U_A \otimes \U_B [\mathcal{J}(\E)] &= \frac{1}{4} ( \I + \mathbf{x} \cdot \boldsymbol{\sigma} \otimes \I + \sum_i t_i \sigma_i \otimes \sigma_i),
    \end{align}
    and so the channel is described, modulo local choices of basis, by the two vectors $\mathbf{x}$ and $\mathbf{t} = (t_1, t_2, t_3)$. The link between $T_{ij}$ and $\mathbf{t}$ is that $ T = O_A \rm{diag}(t_1, t_2, t_3) O_B^T$ for orthogonal matrices $O_A, O_B$ corresponding to the local unitary rotations. It can be shown that if $\mathcal{J}(\E)$ is a valid quantum state (and so $\E$ a valid quantum channel) the vector $\mathbf{x}$ lies in the Bloch sphere, and $\mathbf{t}$ lies in a particular tetrahedron $\mathcal{T}$ in $\mathbf{R}^3$. Moreover, if $\mathbf{x} = \mathbf{0}$ then every $\mathbf{t} \in \mathcal{T}$ corresponds to a valid quantum state. Since $\mathbf{x}$ corresponds to the non-unitality of the quantum channel $\E$, this implies that if $\E$ is a quantum channel with non-unitality vector $\mathbf{x}$ and $T$--matrix $T$ then there exists another quantum channel $\E_u$ with the same $T$--matrix, but which is unital. This implies that for the inner product $\<T_1, T_2\>$ we can without loss of generality assume that one channel is unital, and thus from the previous part of our proof we obtain $-1 \le \<T_1, T_2\>$. The upper bound for $\<T_1,T_2\>$ is unchanged from the previous lemma.
\end{proof}
\begin{lemma}\label{lemma:uc-for-separable-channel}
    For a bipartite separable channel  $\E := \sum_i^r p_i \E_{A,i} \otimes \E_{B,i}$ the correlated unitarity $u_c(\E)$ can be decomposed as
    \begin{equation}
        u_c(\E) = \alpha_A \alpha_B(\sum_{i,j}^{r,r} p_{i} p_{j} \ev{ T_{A}^{i}, T_{A}^{j}} \ev{ T_{B}^{i}, T_{B}^{j}} - \sum_{i,j}^{r,r} p_{i} p_{j} \ev{ T_{A}^{i}, T_{A}^{j}} \sum_{m,n}^{r,r} p_{m} p_n \ev{ T_{B}^{m}, T_{B}^{n}} )
    \end{equation}
    where $T_{A}^{i}$ is the unital block in the Liouville representation of  $\E_{A,i}$ and $T_{B}^{i}$ is the unital block of  $\E_{B,i}$.
\end{lemma}
\begin{proof}
    From definition the correlated unitarity is
	\begin{equation}
	u_c(\E) = \alpha_A \alpha_B( \ev{T_{AB \to AB}, T_{AB \to AB}} -\ev{T_{A \to A}, T_{A \to A}}\ev{T_{B \to B}, T_{B \to B}} ).
	\end{equation}
    Since $\E$ is separable, in the Liouville representation linearity implies
    \begin{equation}
       \vket{\E(\rho)} = \vket{\sum_i^r p_i \E_{A,i} \otimes \E_{B,i}(\rho)}  = \sum_i^r p_i \boldify{\E_{A,i}} \otimes \boldify{\E_{B,i}}\vket{\rho} = \boldify{\E} \vket{\rho}
    \end{equation}
    therefore it follows that the relevant subunital blocks of the channel are simply the weighted sum of the subunital blocks of each product channel:
	\begin{equation}
	T_{AB\to AB} = \sum_{i}^r p_{i} T_{A}^{i} \otimes T_{B}^{i}, \ \ T_{A\to A}  = \sum_{i}^r p_{i} T_{A}^{i}, \ \ T_{B \to B}  = \sum_{i}^r p_{i} T_{B}^{i},
	\end{equation}
	where $T_{A}^{i}$ is the unital block in the Liouville representation of  $\E_{A,i}$ and $T_{B}^{i}$ is the unital block in the Liouville representation of  $\E_{B,i}$.
	Thus the correlated unitarity is
    \begin{equation}
        \begin{split}
            u_c(\E_{AB}) &= \alpha_A \alpha_B( \sum_{i,j}^{r,r} p_{i} p_{j} \ev{ T_{A}^{i} \otimes T_{B}^{i}, T_{A}^{j} \otimes T_{B}^{j}} - \sum_{i,j}^{r,r} p_{i} p_{j} \ev{ T_{A}^{i}, T_{A}^{j}} \sum_{m,n}^{r,r} p_{m} p_n \ev{ T_{B}^{m}, T_{B}^{n}} ), \\
            &= \alpha_A \alpha_B(\sum_{i,j}^{r,r} p_{i} p_{j} \ev{ T_{A}^{i}, T_{A}^{j}} \ev{ T_{B}^{i}, T_{B}^{j}} - \sum_{i,j}^{r,r} p_{i} p_{j} \ev{ T_{A}^{i}, T_{A}^{j}} \sum_{m,n}^{r,r} p_{m} p_n \ev{ T_{B}^{m}, T_{B}^{n}} ).
        \end{split}
    \end{equation}
    Which completes the proof.
\end{proof}

\begin{theorem}[Correlated unitarity is a witness of non-separability]
    Given a bipartite quantum system $AB$ with subsystems $A$ \& $B$ of dimensions $d_A$ \& $d_B$ respectively, for a separable quantum channel $\E_{AB}$, we have that 
    \begin{equation}
        u_c(\E_{AB}) \le C(d_A, d_B) \leq \frac{17}{24} < 1,
    \end{equation}
    where
    \begin{equation}
        C(d_A , d_B) = \beta_A(1+\beta_B) (1-\frac{1}{\min(d_A^2,d_B^2)})   + \frac{1}{4}
    \end{equation}
    where $\beta_i = \frac{1}{d_i^2 -1}$ for $d_i = 2$ or $\beta_i = \frac{d_i}{d_i^2 -1}$ otherwise.
\end{theorem}
\begin{proof}
    From Lemma \ref{lemma:uc-for-separable-channel}, 
    \begin{equation}
        u_c(\E_{AB}) = \alpha_A \alpha_B(\sum_{i,j}^{r,r} p_{i} p_{j} \ev{ T_{A}^{i}, T_{A}^{j}} \ev{ T_{B}^{i}, T_{B}^{j}} - \sum_{i,j}^{r,r} p_{i} p_{j} \ev{ T_{A}^{i}, T_{A}^{j}} \sum_{m,n}^{r,r} p_{m} p_n \ev{ T_{B}^{m}, T_{B}^{n}} )
    \end{equation}
    where $T_{A}^{i}$ is the unital block in the Liouville representation of  $\E_{A,i}$ and $T_{B}^{i}$ is the unital block of  $\E_{B,i}$. To simplify notation we label the normalized inner products
    \begin{equation}
        t_{ij} := \alpha_{A }\ev{ T_{A}^{i}, T_{A}^{j}} \text{ and } s_{ij} := \alpha_{B}\ev{ T_{B}^{i}, T_{B}^{j}},
    \end{equation}
    and define $A := \sum_{ij}^{r,r} p_i p_j t_{ij}$ and $B := \sum_{ij}^{r,r} p_i p_j s_{ij}$. In this notation the correlated unitarity of the separable channel is just
    \begin{equation}\label{eqn:uc-for-sep-channel}
        u_c(\E_{AB}) = \sum_{ij}^{r,r} p_i p_j t_{ij} s_{ij} - AB.
    \end{equation}
    From Lemma \ref{lemma:range-of-t-inner-product} the range of any particular $t_{ij}$ is
    \begin{equation}
        -\beta_A \leq t_{ij} \leq 1
    \end{equation}
    where $\beta_A = d_A \alpha_A$ applies to all channels and $\beta_A = \alpha_A $ holds for the case of qubit channels or if one of the channels is unital. Additionally from the non-negativity of the Hilbert Schmidt inner product $t_i \equiv t_{ii} \geq 0$. Similarly for the $B$ subsystem: $-\beta_B \leq s_{ij} \leq 1$ and $s_i \equiv s_{ii} \geq 0$.
    
    We now bound the first term in equation (\ref{eqn:uc-for-sep-channel}) in relation to the second. Out of the $r^2$ possible terms in the first term there are $r$ terms that are equal to $p_i^2 t_{i}s_{i}$ (namely when $i=j$). Now suppose that out of the $r^2-r$ remaining terms there are $k$ terms where $t_{ij}$ is negative: $t_{-,m}$, ($m=\{0,1,...,k-1,k\}$), and $r^2-(r+k)$ other terms where $t_{ij}$ is positive: $t_{+,n}$, ($n=\{0,1,...,r^2-(r+k)-1,r^2-(r+k)\}$). We can then write the correlated unitarity as
    \begin{equation}
        \begin{split}
            u_c(\E_{AB}) &= \sum_i^r p_i^2 t_{i}s_{i} + \sum_{i \neq j}^{r^2 - r} p_i p_j t_{ij} s_{ij} -AB, \\
            &=  \sum_i^r p_i^2 t_{i}s_{i} + \sum_{m=(ij), i \neq j}^k p_i p_j t_{-,m} s_{m} + \sum_{n=(ij), i\neq j}^{r^2 -r -k} p_i p_j t_{+,n} s_n -  AB, \\
            &=  \sum_i^r p_i^2 t_{i}s_{i} - \sum_{m=(ij), i \neq j}^k p_i p_j  \abs{t_{-,m}} s_{m} + \sum_{n=(ij), i\neq j}^{r^2 -r -k} p_i p_j  t_{+,n} s_n -  AB. \\
        \end{split}
    \end{equation}
    We now bound the summation of positive and negative $t_{i \neq j}$ terms. As all $t_{-,m} \leq 0$ then since $|t_{-m}|\leq \beta_A$ we can bound the summation of negative terms as
    \begin{equation}\label{eqn:bound-negative-a}
        \sum_{m=(ij), i \neq j}^k p_i p_j \abs{t_{-,m}} \leq \sum_{m=(ij), i \neq j}^k \beta_A p_i p_j \leq \sum_{i \neq j}^{r^2 -r} \beta_A p_i p_j = \beta_A (1-\sum_i^r p_i^2)
    \end{equation}
    where we have maximized $k$ to include all $r^2 -r$ possible terms, and used the simple relation that $\sum_i^r p_i^2 + \sum_{ij}^{r^2-r} p_i p_j =1$. From definition, $ A = \sum_i^r p_i^2 t_{i} + \sum_{i \neq j}^{r^2 - r} p_i p_j t_{ij}$ therefore the whole summation of cross terms can be written as
    \begin{equation}
        \sum_{i \neq j}^{r^2 - r} p_i p_j t_{ij} =  \sum_{n=(ij), i\neq j}^{r^2 -r -k} p_i p_j  t_{+,n} - \sum_{m=(ij), i \neq j}^k p_i p_j  \abs{t_{-,m}}  = A - \sum_i^r p_i^2 t_{i}.
    \end{equation}
    From this we can bound the summation of the positive terms using the previous bound in eqn. (\ref{eqn:bound-negative-a}):
    \begin{equation}
        \begin{split}
            \sum_{n=(ij), i\neq j}^{r^2 -r -k} p_i p_j  t_{+,n} &= A - \sum_i^r p_i^2 t_{i} + \sum_{m=(ij), i \neq j}^k p_i p_j \abs{t_{-,m}}, \\
            \sum_{n=(ij), i\neq j}^{r^2 -r} p_i p_j  t_{+,n} & \leq A - \sum_i^r p_i^2 t_{i} +  \beta_A (1-\sum_i^r p_i^2).
        \end{split}
    \end{equation}
    Since both $t_{+,n}\geq 0$ and $\abs{t_{-,m}} \geq 0$ and all elements $-\min(\beta_B, \sqrt{s_i s_j}) \leq s_{i \neq j} \leq \sqrt{s_i s_j} \leq 1$, then we can bound the summation containing $t_{+,n}  s_n$ elements as
    \begin{equation}
        \sum_{n=(ij), i\neq j}^{r^2 -r -k} p_i p_j  t_{+,n} s_n \leq \sum_{n=(ij), i\neq j}^{r^2 -r -k} p_i p_j  t_{+,n} \leq  A - \sum_i^r p_i^2 t_{i} +  \beta_A (1-\sum_i^r p_i^2)
    \end{equation}
    and the summation containing $t_{-,m} s_m$ elements (assuming $\sqrt{s_i s_j} \geq \beta_B$)
    \begin{equation}
        - \sum_{m=(ij), i \neq j}^k p_i p_j  \abs{t_{-,m}} s_{m} \leq \beta_B \sum_{m=(ij), i \neq j}^k p_i p_j  \abs{t_{-,m}} \leq \beta_B (\beta_A (1-\sum_i^r p_i^2)).
    \end{equation}
    Putting all this together we get a bound on the correlated unitarity of
    \begin{equation}
        \begin{split}
            u_c(\E_{AB}) &\leq \sum_i^r p_i^2 t_{i}s_{i} + \beta_B (\beta_A (1-\sum_i^r p_i^2)) + A - \sum_i^r p_i^2 t_{i} +  \beta_A (1-\sum_i^r p_i^2) -  AB, \\
            &\leq \sum_i^r p_i^2 t_{i}(s_{i}-1) + \beta_A(1+\beta_B) (1-\sum_i^r p_i^2)  + A(1-B).
        \end{split}
    \end{equation}
    With no loss of generality we can set $A \leq B$ as $A$ and $B$ are interchangeable. Therefore we have that $A(1-B) \leq B(1-B)$. As $0 \leq B \leq 1$, this is maximized when $B=1/2$. Additionally as $s_i \leq 1$ then $s_i - 1 \leq 0$ and the whole first term is negative. Therefore
    \begin{equation}
        \begin{split}
            u_c(\E_{AB}) &\leq \beta_A(1+\beta_B) (1-\sum_i^r p_i^2)  + \frac{1}{4}.
        \end{split}
    \end{equation}
    Further from the Cauchy-Schwartz inequality $\sum_i^r p_i^2 \geq \frac{1}{r} \geq \frac{1}{\min(d_A^2,d_B^2)}$. Putting this together we have:
        \begin{equation}
            \begin{split}
                u_c(\E_{AB}) &\leq \beta_A(1+\beta_B) (1-\frac{1}{\min(d_A^2,d_B^2)})   + \frac{1}{4}.
            \end{split}
        \end{equation}
        where we have $\beta_i = \frac{d_i}{d_i^2 -1}$ for $d_i>2$ and $\beta_i = \frac{1}{d_i^2 -1}$ for $d_i = 2$. Firstly, for $d_A = d_B=2$ we find that
        \begin{equation}
            \begin{split}
                u_c(\E_{AB}) &\leq \frac{d_B^2}{(d_A^2 -1)(d_B^2-1)} (1-\frac{1}{d_B^2})  + \frac{1}{4} = \frac{1}{3} + \frac{1}{4} = \frac{7}{12}.
            \end{split}
        \end{equation}
        We now eliminate the two other cases with a qubit subsystem. Firstly, ($d_A =2, d_B > 2$) yields
        \begin{equation}
            \begin{split}
                u_c(\E_{AB}) &\leq \frac{1}{d_A^2 -1}(1+\frac{d_B}{d_B^2 -1}) (1-\frac{1}{d_A^2})  + \frac{1}{4}, \\
                &\leq \frac{1}{3}(1+\frac{d_B}{d_B^2 -1}) (\frac{3}{4})  + \frac{1}{4}, \\
            \end{split}
        \end{equation}
        which is maximised for $d_B = 3$ giving $  u_c(\E_{AB}) \leq 17/32 \approx 0.53$. Secondly, ($d_A > 2, d_B = 2$) yields
        \begin{equation}
            \begin{split}
                u_c(\E_{AB}) &\leq \frac{d_A}{d_A^2 -1}(1+\frac{1}{d_B^2 -1}) (1-\frac{1}{d_B^2})  + \frac{1}{4}, \\
                &\leq \frac{d_A}{d_A^2 -1}(1+\frac{1}{3}) (\frac{3}{4})  + \frac{1}{4}, \\
                &\leq \frac{d_A}{d_A^2 -1}(\frac{4}{3}) (\frac{3}{4})  + \frac{1}{4}, \\
                &\leq \frac{d_A}{d_A^2 -1}  + \frac{1}{4}, \\
            \end{split}
        \end{equation}
        which is maximised for $d_A = 3$ giving $  u_c(\E_{AB}) \leq 5/8 \approx 0.63$.
        
        Now we consider the two broader cases. Firstly, ($d_A, d_B > 2$ with  $d_B \geq d_A$) which yields
        \begin{equation}
            \begin{split}
                u_c(\E_{AB}) &\leq \frac{d_A}{d_A^2 -1}(1+\frac{d_B}{d_B^2 -1}) (1-\frac{1}{d_A^2})  + \frac{1}{4}, \\
                &\leq \frac{1}{d_A}(1+\frac{d_B}{d_B^2 -1})  + \frac{1}{4}, \\
            \end{split}
        \end{equation}
        which is maximised for $d_A = d_B = 3$ giving $  u_c(\E_{AB}) \leq 17/24 (\approx 0.71)$. Secondly, ($d_A, d_B > 2$ with  $d_A > d_B$) which yields
        \begin{equation}
            \begin{split}
                u_c(\E_{AB}) &\leq \frac{d_A}{d_A^2 -1}(1+\frac{d_B}{d_B^2 -1}) (1-\frac{1}{d_B^2})  + \frac{1}{4}, \\
            \end{split}
        \end{equation}
        which is minimised for  $d_A =4, d_B = 3$ giving $  u_c(\E_{AB}) \leq 311/540 \approx 0.58$.
        This completes the proof.
\end{proof}

\section{Analysis of local independent twirls on $A$ and $B$}\label{appendix-for-pAB-protocol}
\subsection{Definition of subspace projectors}
\begin{lemma}\label{lemma:projector-P}
    The operator
    \begin{equation}
        P := \int \!\!d \mu_{\mbox{\tiny Haar}}(U)\ \mathbfcal{U}^{\otimes 2} = \int \!\!d \mu_{\mbox{\tiny Haar}}(U)\ (U \otimes U^{*})^{\otimes2},
    \end{equation}
    on $\mathcal{H}^{\otimes4}=V \oplus V^\perp$ is a projector into the subspace $V=\textrm{span}(\vket{\ident^{\otimes 2}}, \vket{\mathbb{F}})$, where $\mathbb{F}$ is the Flip operator on the sub-systems, and therefore $P=0$ on $V^\perp$.
\end{lemma}
\begin{proof}
    (Of Lemma \ref{lemma:projector-P})
    For any group $G$ with an invariant measure (i.e. finite or compact) and a representation V, the averaging over all elements of the group gives a projector,
    \begin{equation}
        P=\int V(g) \ dg,
    \end{equation}
    onto the invariant subspace $ \{ |\psi \>: V(g)\ket{\psi}=\ket{\psi} \forall \ g \in  G\}.$ To find the invariant subspace for $V(U)=(U \otimes U^{*})^{\otimes2}$ it is easier to look at $V'(U) = U \otimes U \otimes U^* \otimes U^*$. According to the definition of the invariant subspace, we must find $X$ such that 
    \begin{equation}
        V'(U)\vket{X}=\vket{X},
    \end{equation}
    or equivalently $[X,U \otimes U]=0$. 

    We can decompose $U \otimes U$ into irreducible representations of $U(d)$. There are 2 of them: the symmetric subspace and the alternating subspace. This is related to the fact that symmetric group on two elements has two irreducible representations: the trivial one ($\ident$) and the alternating one ($\mathbb{F}$). 

    Using Schur's lemma \footnote{Schur's Lemma states that the only matrices that commute with all elements of an irreducible representation of a group are scalar multiples of $\ident$.} the operator $X$ must be a multiple of the identity when restricted to either of these two subspaces. Putting everything together,  (up to reordering of spaces) the invariant subspace is spanned by $\vket{\ident^{\otimes 2}}$ and $\vket{\mathbb{F}}$.
\end{proof}

\begin{lemma}\label{lemma:eigenvectorsofv}
    A normalized basis for the invariant vector space  $V=\textrm{span}(\vket{\ident^{\otimes 2}}, \vket{\mathbb{F}})$ is given by
    \begin{equation}
        \begin{split}
            \ket{0} &= \vket{X_0} \otimes \vket{X_0}, \\
            \ket1 &= \frac{1}{\sqrt{d^2-1}} \sum_{k=1}^{d^2-1} \vket{X_k} \otimes \vket{X_k^\dagger},
        \end{split}
    \end{equation}
    where $X_\mu=(X_0=\ident/\sqrt{d}, X_i)$.
\end{lemma}
\begin{proof}
    We defined the tensor product of two vectorized matrices as:
    \begin{equation}
        \vket{A \otimes B} := \vket{A} \otimes \vket{B},
    \end{equation}
    Applying this definition to the $1^\st$ vector that spans the space $\vket{\ident^{\otimes 2}} = \vket{\ident} \otimes \vket{\ident} = d \vket{X_0} \otimes \vket{X_0}$. Normalizing, the first eigenvector is therefore $\ket{0} := \vket{X_0} \otimes \vket{X_0}$.

    The Flip operator in our basis is given by considering the permutation of computational basis states:
    \begin{equation}
        \mathbb{F} := \sum_{i,j}^{d,d} \ketbra{j}{i} \otimes \ketbra{i}{j} = \sum_{i,j}^{d,d} \ketbra{j}{i} \otimes (\ketbra{j}{i})^\dagger = \sum_{\mu}^{d^2} X_\mu \otimes X_\mu^\dagger
    \end{equation}
    up to a dimensional factor. Therefore $\vket{\mathbb{F}} =  \sum_{\mu=0}^{d^2-1} \vket{X_\mu} \otimes \vket{X_\mu^\dagger}$.
    From inspection the $2^\nd$ normalized eigenvector that spans this subspace is
    \begin{equation}
            \ket{1} = \frac{1}{\sqrt{d^2-1}} \sum_{k=1}^{d^2-1} \vket{X_k} \otimes \vket{X_k^\dagger}.
    \end{equation}
\end{proof}

We can now write the decomposition of the projector $ P := \ketbra{0}{0} + \ketbra{1}{1}$ as
\begin{equation}
    P = \vketbra{X_0}{X_0} \otimes \vketbra{X_0}{X_0} + \frac{1}{d^2-1} \sum_{i,j}^{d^2-1}  \vketbra{X_i}{X_j} \otimes \vketbra{X_i^\dagger}{X_j^\dagger}.
\end{equation}

\begin{definition}
    The projector
    \begin{equation}
        P_{AB} := \int \!\!d \mu_{\mbox{\tiny Haar}}(U_A)\ \int \!\!d \mu_{\mbox{\tiny Haar}}(U_B)\ (\mathbfcal{U}_A \otimes \mathbfcal{U}_B )^{\otimes 2},
    \end{equation}
    for the tensor product of two copies of a bipartite system with subsystems $A$ \& $B$.
\end{definition}
Since the integrals are independent, it is readily seen that,
\begin{equation}
    P_{AB} = P_A \otimes P_B = \sum_{i,j} \ketbra{ij}{ij}
\end{equation}
where $P_A$ is the projector $P$ on subsystem $A$, and similarly for $B$. We can now calculate the action of the projector $P_{AB}$ on two copies of the Liouville representation of a bipartite channel $P_{AB} \mathbfcal{E}^{\otimes 2} P_{AB}$.

\subsection{Calulation of elements of $P_{AB} \mathbfcal{E}^{\otimes 2} P_{AB}$ \& the matrix of sub-unitarities $\S$.}\label{appendpab}
We now show that the operator $P_{AB}\mathbfcal{E}^{\otimes 2} P_{AB}$ can be viewed as encoding the quadratic order invariants of the quantum channel, and in particular the traceless components form a $3\times 3$ \emph{matrix of sub-unitarities} $\S$ for the bipartite quantum channel. A basis of four eigenvectors of $P_{AB}$ can be written in the basis $(\vket{X_{\mu}} \otimes \vket{Y_{\nu}})^{\otimes 2}$ to match the order of the subspaces of $\mathbfcal{E}^{\otimes 2}$. This gives
\begin{equation}
    \begin{split}
        \ket{00} &= \vket{X_0} \otimes \vket{Y_0} \otimes \vket{X_0} \otimes \vket{Y_0}, \\
        \ket{10} &= \sqrt{\alpha_A} \sum_{n=1}^{d_A^2-1} \vket{X_n} \otimes \vket{Y_0} \otimes \vket{X_n^\dagger} \otimes \vket{Y_0}, \\
        \ket{01} &= \sqrt{\alpha_B} \sum_{m=1}^{d_B^2-1} \vket{X_0} \otimes \vket{Y_m} \otimes \vket{X_0} \otimes \vket{Y_m^\dagger}, \\
        \ket{11} &= \sqrt{\alpha_A \alpha_B} \sum_{n,m=1}^{d_A^2-1,d_B^2-1} \vket{X_n} \otimes \vket{Y_m} \otimes \vket{X_n^\dagger} \otimes \vket{Y_m^\dagger}, \\
    \end{split}
\end{equation}
where $\alpha_i = 1/(d_i^2 -1)$.
We can now calculate the matrix elements of $ P_{AB}\mathbfcal{E}^{\otimes 2} P_{AB}$ in this basis. Firstly, for each subsystem, as the $\mu=0$  elements are proportional to the identity we have $X_0^\dagger=X_0$ \& $Y_0^\dagger=Y_0$. Secondly, as the channel $\E$ is a CPTP map, we have that $\E( (X_{\mu} \otimes Y_{\nu})^\dagger) = (\E(X_{\mu} \otimes Y_{\nu}))^\dagger$ for any elements of the basis, and so
\begin{equation}
    \begin{split}
        \vbra{X_\mu^\dagger \otimes Y_\nu^\dagger} \mathbfcal{E} \vket{X_\sigma^\dagger \otimes Y_\omega^\dagger} &= \vbraket{X_\mu^\dagger \otimes Y_\nu^\dagger}{\E(X_\sigma^\dagger \otimes Y_\omega^\dagger)}, \\
        &= \tr[(X_\mu^\dagger \otimes Y_\nu^\dagger)^\dagger \E(X_\sigma^\dagger \otimes Y_\omega^\dagger)], \\
        &= \tr[\E(X_\sigma^\dagger \otimes Y_\omega^\dagger) X_\mu \otimes Y_\nu], \\
        &= \tr[\E(X_\sigma \otimes Y_\omega)^\dagger X_\mu \otimes Y_\nu], \\
        &=  \vbra{\E(X_\sigma \otimes Y_\omega)} \vket{X_\mu \otimes Y_\nu}, \\
        &=  \vbra{X_\sigma \otimes Y_\omega} \mathbfcal{E}^\dagger \vket{X_\mu \otimes Y_\nu},
    \end{split}
\end{equation}
where $\cal{E}\hc$ corresponds to the adjoint of $\E$ that is defined via $\Tr(A\E(B)) = \Tr(\E\hc(A) B)$. Futhermore note that if the non-unital block of $\E$ is $T$, then the non-unital block of $\E\hc$ is $T\hc$. 

We can now calculate the 16 possible combinations $\bra{a} \mathbfcal{E}^{\otimes 2}\ket{b}$. One element is simply equivalent to the trace preserving property of a quantum channel $\bra{00} \mathbfcal{E}^{\otimes 2} \ket{00} = (\tr[ \frac{\ident}{\sqrt{d}} \E(\frac{\ident}{\sqrt{d}}) ])^2 = 1$. The remaining elements can be divided into 3 sub-blocks to be defined
\begin{equation}
    P_{AB} \mathbfcal{E}^{\otimes 2} P_{AB} =
    \bordermatrix{~ 
    & \ket{00} & \ket{ij}  \cr
    \bra{00} & 1 & \bm{0} \cr
    \bra{ij} & \bm{x} & \S \cr
    }
    \text{ where } ij \in \{01,10,11\}.
\end{equation}

Consider a diagonal $\bra{10} \mathbfcal{E}^{\otimes 2} \ket{10}$ element in the matrix $\S$, from the above properties it follows that
\begin{equation}
    \begin{split}
        \bra{10} \mathbfcal{E}^{\otimes 2} \ket{10} &= \alpha_A \sum_{i,j=1}^{d_A^2-1} \vbra{X_i} \otimes \vbra{Y_0} \mathbfcal{E} \vket{X_j} \otimes \vket{Y_0}  \vbra{X_i^\dagger} \otimes \vbra{Y_0} \mathbfcal{E} \vket{X_j^\dagger} \otimes \vket{Y_0}, \\
        &= \alpha_A \sum_{i,j=1}^{d_A^2-1} \vbra{X_i \otimes Y_0} \mathbfcal{E} \vket{X_j \otimes Y_0}  \vbra{X_j \otimes Y_0} \mathbfcal{E}^\dagger \vket{X_i \otimes Y_0}, \\
        &= \alpha_A \tr[T_{A \to A} T_{A \to A}^\dagger] =\alpha_A \tr[T_{A \to A}^\dagger T_{A \to A}] = u_{A \to A}(\E).
    \end{split}
\end{equation}
and similarly $\bra{01} \mathbfcal{E}^{\otimes 2} \ket{01} = u_{B \to B}(\E)$ \& $\bra{11} \mathbfcal{E}^{\otimes 2} \ket{11} = u_{AB \to AB}(\E)$.
Off diagonal elements in $A$ can be calculated with an additional dimensional factor. For example, following the same line
\begin{equation}
    \begin{split}
        \bra{01} \mathbfcal{E}^{\otimes 2} \ket{10} &= \sqrt{\alpha_A \alpha_B} \sum_{i,j=1}^{(d_B^2-1)(d_A^2-1)} \vbra{X_0 \otimes Y_i} \mathbfcal{E} \vket{X_j \otimes Y_0} \vbra{X_j \otimes Y_0} \mathbfcal{E}^\dagger \vket{X_0 \otimes Y_i}, \\
        &= \sqrt{\alpha_A \alpha_B} \tr[T_{A \to B} T_{A \to B}^\dagger] = \sqrt{\frac{\alpha_B}{\alpha_A}} u_{A \to B}(\cal{E}). \\
    \end{split}
\end{equation}
Further we have elements such as
\begin{equation}\label{lop-eqn}
    \begin{split}
        \bra{11} \mathbfcal{E}^{\otimes 2} \ket{10} =& \alpha_A \sqrt{\alpha_B} \sum_{k,j,n=1}^{(d_A^2-1)(d_B^2-1)} \vbra{X_j \otimes Y_n}  \mathbfcal{E} \vket{X_k \otimes Y_0} \vbra{X_k \otimes Y_0} \mathbfcal{E}^\dagger \vket{X_j \otimes Y_n}, \\
        = & \alpha_A \sqrt{\alpha_B} \tr[T_{A \to AB}^\dagger T_{A \to AB}] = \sqrt{\alpha_B} u_{A \to AB}(\cal{E})
    \end{split}
\end{equation}
and $\bra{10} \mathbfcal{E}^{\otimes 2} \ket{11} = \alpha_A \sqrt{\alpha_B} \tr[T_{AB \to A}^\dagger T_{AB \to A}] = \frac{1}{\sqrt{\alpha_B}} u_{AB \to A}(\cal{E})$. The remaining elements of $\S$ can be found by swapping the labeling of the subsystems.
Putting this together we have the matrix of sub-unitarities given by,
\begin{equation}
    \S =
    \bordermatrix{~ 
        & \ket{10}  & \ket{11} & \ket{01}  \cr
    \bra{10}   & u_{A \to A}(\cal{E}) & \frac{1}{\sqrt{\alpha_B}} u_{AB \to A}(\cal{E}) & \sqrt{\frac{\alpha_A}{\alpha_B}} u_{B \to A}(\cal{E}) \cr
    \bra{11}   & \sqrt{\alpha_B} u_{A \to AB}(\cal{E}) & u_{AB \to AB}(\cal{E}) & \sqrt{\alpha_A} u_{B \to AB}(\cal{E}) \cr
    \bra{01}  & \sqrt{\frac{\alpha_B}{\alpha_A}} u_{A \to B}(\cal{E}) & \frac{1}{\sqrt{\alpha_A}} u_{AB \to B}(\cal{E}) & u_{B \to B}(\cal{E}) \cr
    }.
\end{equation}

The three elements $\bra{ij} \mathbfcal{E}^{\otimes 2} \ket{00}$ with $ij \in\{01,10,11\}$ quantify the non-unitality of the channel for each subsystem to quadratic order, through the H-S inner product of the generalized Bloch vector $\mathbf{x}$ for each subsystem. We can define $ x_{i} := \mathbf{x}_{i \to i}^\dagger \mathbf{x}_{i \to i}$. Therefore we have
\begin{equation}
    \begin{split}
        \bra{10} \mathbfcal{E}^{\otimes 2} \ket{00} &= \sqrt{\alpha_A} \sum_{i=1}^{d_A^2-1} \vbra{X_i \otimes Y_0} \mathbfcal{E} \vket{X_0 \otimes Y_0}  \vbra{X_0 \otimes Y_0} \mathbfcal{E}^\dagger \vket{X_i \otimes Y_0}, \\
        &= \sqrt{\alpha_A}  \mathbf{x}_{A \to A}^\dagger \mathbf{x}_{A \to A} = \sqrt{\alpha_A} x_{A},
    \end{split}
\end{equation}
similarly $\bra{11} \mathbfcal{E}^{\otimes 2} \ket{00} = \sqrt{\alpha_A \alpha_B} x_{AB}$, $\bra{01} \mathbfcal{E}^{\otimes 2} \ket{00} = \sqrt{\alpha_B} x_{B}$. Therefore $\bm{x}^T = (\sqrt{\alpha_A} x_{A}, \sqrt{\alpha_A \alpha_B} x_{AB}, \sqrt{\alpha_B} x_{B})$.

The final three elements $\bra{00} \mathbfcal{E}^{\otimes 2} \ket{ij}$ with $ij \in \{01,10,11\}$ are required to the zero from the trace preserving properties of a quantum channel. For example, considering $\bra{00} \mathbfcal{E}^{\otimes 2} \ket{10}$ for $\E$ to be a valid TP map we must have $\vbra{X_0 \otimes Y_0} \mathbfcal{E} \vket{X_i \otimes Y_0} = 0$ for all $i$. Therefore
\begin{equation}
    \bra{00} \mathbfcal{E}^{\otimes 2} \ket{10} = \sqrt{\alpha_A} \sum_{i=1}^{d_A^2-1} \vbra{X_0 \otimes Y_0} \mathbfcal{E} \vket{X_i \otimes Y_0}  \vbra{X_i \otimes Y_0} \mathbfcal{E}^\dagger \vket{X_0 \otimes Y_0} = 0.
\end{equation}
Through the same argument $\bra{00} \mathbfcal{E}^{\otimes 2} \ket{01} = \bra{00} \mathbfcal{E}^{\otimes 2} \ket{11} =0$.

Finally, putting all elements together we have,
\begin{equation}
    P_{AB} \mathbfcal{E}^{\otimes 2} P_{AB} =
        \bordermatrix{~ 
        & \ket{00} & \ket{10}  & \ket{11} & \ket{01}  \cr
        \bra{00} &  1 & 0 & 0 & 0  \cr
        \bra{10}  & \sqrt{\alpha_A} x_{A} & u_{A \to A}(\cal{E}) & \frac{1}{\sqrt{\alpha_B}} u_{AB \to A}(\cal{E}) & \sqrt{\frac{\alpha_A}{\alpha_B}} u_{B \to A}(\cal{E}) \cr
        \bra{11}  & \sqrt{\alpha_A \alpha_B} x_{AB} & \sqrt{\alpha_B} u_{A \to AB}(\cal{E}) & u_{AB \to AB}(\cal{E}) & \sqrt{\alpha_A} u_{B \to AB}(\cal{E}) \cr
        \bra{01}  & \sqrt{\alpha_B} x_{B} & \sqrt{\frac{\alpha_B}{\alpha_A}} u_{A \to B}(\cal{E}) & \frac{1}{\sqrt{\alpha_A}} u_{AB \to B}(\cal{E}) & u_{B \to B}(\cal{E}) \cr
        }.
\end{equation}
Comparing this with decomposition of the Liouville representation of a bipartite channel $\cal{E}$ in eqn. (\ref{eqn:liouville-bipartitechannel}), we see that $P_{AB}$ produces the normalized purity of every sub-block of $\cal{E}$. As sub-unitarities are the normalized purity of sub-blocks of the unital block $T$, these values are extracted, as well as the absolute value of the non-unital vector for both sub-systems. Using the form of the top row of $P_{AB} \mathbfcal{E}^{\otimes 2} P_{AB}$, it is easily seen that
\begin{equation}
    \det(P_{AB} \mathbfcal{E}^{\otimes 2} P_{AB} - \lambda \ident) = (1-\lambda)\det(\S - \lambda \ident)
\end{equation}
and therefore for any channel $\E$ the 4 eigenvalues of $P_{AB} \mathbfcal{E}^{\otimes 2} P_{AB}$ will be $\lambda_0 = 1$ and the 3 eigenvalues of $\S$.

\subsection{The matrix components for separable channels}\label{section:pab-prod}
For a product channel $\E = \cal{E}_A \otimes \cal{E}_B$ the sub-unitarity matrix $\S$ takes a particularly simple form. Since the channel is separable quantum information does not flow between $A$ and $B$ and Theorem \ref{product-lemma} tells us that $u_{A \to B}(\cal{E}_A \otimes \cal{E}_B)=u_{B \to A}(\cal{E}_A \otimes \cal{E}_B)=0$ and $u_{AB \to AB}(\cal{E}_A \otimes \cal{E}_B)=u_{A \to A}(\cal{E}_A \otimes \cal{E}_B) \cdot u_{B \to B}(\cal{E}_A \otimes \cal{E}_B)$. Thus, for a product channel $\E = \cal{E}_A \otimes \cal{E}_B$
\begin{equation}
    P_{AB} \mathbfcal{E}^{\otimes 2} P_{AB} =
        \bordermatrix{~ 
        & \ket{0 0} & \ket{1 0}  & \ket{1 1} & \ket{0 1}  \cr
        \bra{0 0} &  1 & 0 & 0 & 0  \cr
        \bra{1 0}  & \sqrt{\alpha_A} x_{A} & u(\cal{E}_A) & 0 & 0 \cr
        \bra{1 1}  & \sqrt{\alpha_A \alpha_B} x_A x_B & \sqrt{\alpha_B} u(\cal{E}_A) x_B & u(\cal{E}_A) u(\cal{E}_B) & \sqrt{\alpha_A} u(\cal{E}_B) x_A \cr
        \bra{0 1}  & \sqrt{\alpha_B} x_{B} & 0 & 0 & u(\cal{E}_B) \cr
        }.
\end{equation}
From this it is readily seen that the eigenvalues for a product channel are $\{1, u(\E_A), u(\E_B), u(\E_A) u(\E_B)\}$. More generally, for the case of a \emph{separable} channel $\E_{AB}$ from Lemmas \ref{lemma:u-aba-zero-for-seperable} \& \ref{lemma:u-ab-zero-for-seperable} we find instead that
\begin{equation}
    P_{AB} \mathbfcal{E}^{\otimes 2} P_{AB} =
        \bordermatrix{~ 
        & \ket{0 0} & \ket{1 0}  & \ket{1 1} & \ket{0 1}  \cr
        \bra{0 0} &  1 & 0 & 0 & 0  \cr
        \bra{1 0}  & \sqrt{\alpha_A} x_{A} & u_{A \rightarrow A}(\cal{E}_{AB}) & 0 & 0 \cr
        \bra{1 1}  & \sqrt{\alpha_A \alpha_B} x_{A B} & \sqrt{\alpha_B} u_{A \to AB}(\cal{E}_{AB}) & u_{AB \rightarrow AB} (\E_{AB}) & \sqrt{\alpha_A} u_{B \to AB}(\cal{E}_{AB}) \cr
        \bra{0 1}  & \sqrt{\alpha_B} x_{B} & 0 & 0 & u_{B \rightarrow B}(\cal{E}_{AB}) \cr
        }.
\end{equation}
and so now the eigenvalues are $\{1, u_{A \rightarrow A}(\cal{E}_{AB}),u_{B \rightarrow B}(\cal{E}_{AB}),u_{AB \rightarrow AB}(\cal{E}_{AB})\}$., Therefore $(P_{AB} \mathbfcal{E}^{\otimes 2} P_{AB})^m$, will have eigenvalues
\begin{equation}\label{eqn:eigens-pab-product-channel}
    \{ \lambda_i \} = \{1,u_{A \rightarrow A}(\cal{E}_{AB})^m,u_{B \rightarrow B}(\cal{E}_{AB})^m,u_{AB \rightarrow AB}(\cal{E}_{AB})^m \},
\end{equation}  
which implies that the sub-unitarities are the decay constants for the benchmarking protocol. 

More generally we do not have such a simple link between the eigenvalues and sub-unitarities. Indeed, it may be the case that the matrix cannot be diagonalized fully, and so one must instead use a Jordan decomposition to determine the decay law for the protocol. We provide the details for the fully general case in the next section.
\subsection{Jordan decomposition for arbitrary bipartite channels}\label{section:jordan-decomp-pab}
For a general bipartite channel $\E$ we can use the Jordan normal form of the matrix $P_{AB} \mathbfcal{E}^{\otimes 2} P_{AB}$  to study the structure scales with a power, $(P_{AB} \mathbfcal{E}^{\otimes 2} P_{AB})^m$.
\begin{definition}
    Using the Jordan matrix decomposition of any square matrix $M$, we can find the Jordan normal form such that
    \begin{equation}
        M = S^{-1} J S,
    \end{equation}
    where $S$ is a invertible matrix, and $J$ is a block diagonal matrix of Jordan blocks \cite{horn2012matrix}.
\end{definition}

\begin{corollary}\label{jordanpower}
    The Jordan matrix decomposition of a square matrix $M$ to the power $n$ follows
    \begin{equation}
        M^n = S^{-1} J^n S.
    \end{equation}
\end{corollary}
\begin{proof}
    This follows simply from $S S^{-1} = \ident$.
\end{proof}
This implies that if  write $P_{AB} \mathbfcal{E}^{\otimes 2} P_{AB}$, in a Jordan normal form, $J$, then the decay law of $(P_{AB} \mathbfcal{E}^{\otimes 2} P_{AB})^m$ will be determined entirely by $J^m$. There are 3 possibilities that could occur:
\begin{equation}
    J = 
    \begin{pmatrix}
        1 & 0 & 0 & 0  \cr
        0 & \lambda_1 & 0 & 0 \cr
        0 & 0 & \lambda_3 & 0 \cr
        0 & 0 & 0 & \lambda_2 \cr
    \end{pmatrix}
    , \ J =
    \begin{pmatrix}
        1 & 0 & 0 & 0  \cr
        0 & \lambda_1 & 1 & 0 \cr
        0 & 0 & \lambda_1 & 0 \cr
        0 & 0 & 0 & \lambda_2 \cr
    \end{pmatrix}
    , \ J =
    \begin{pmatrix}
        1 & 0 & 0 & 0  \cr
        0 & \lambda_1 & 1 & 0 \cr
        0 & 0 & \lambda_1 & 1 \cr
        0 & 0 & 0 & \lambda_1 \cr
    \end{pmatrix}
    ,
\end{equation}
where $\lambda_i$ are the eigenvalues of the block $\S$. Which form the Jordan decomposition takes depends on the degeneracy of $\lambda_i$ and whether the geometric and algebraic multiplicities of each $\lambda_i$ coincide \cite{horn2012matrix}. 

For $J$ diagonal, we have that
\begin{equation}
    (P_{AB} \mathbfcal{E}^{\otimes 2} P_{AB})^m = S^{-1} J^m S = S^{-1}
    \begin{pmatrix}
        1 & 0 & 0 & 0  \cr
        0 & \lambda_1^m & 0 & 0 \cr
        0 & 0 & \lambda_3^m & 0 \cr
        0 & 0 & 0 & \lambda_2^m \cr
    \end{pmatrix}
    S,
\end{equation}
where $\{ \lambda_i \}$ are the eigenvalues of $\S$.
Therefore,
\begin{equation}
    (P_{AB} \mathbfcal{E}^{\otimes 2} P_{AB})^m = S^{-1} ( \ketbra{00} +  \lambda_1^m \ketbra{10} +  \lambda_2^m \ketbra{01} +  \lambda_3^m \ketbra{11} ) S.
\end{equation}

If the Jordan decomposition of $P_{AB} \mathbfcal{E}^{\otimes 2} P_{AB}$ is not completely diagonal, then $(P_{AB} \mathbfcal{E}^{\otimes 2} P_{AB})^m$ still scales with the eigenvalues of $\S$ but in a slightly more complex manner. From above, the 2 remaining options are
\begin{equation}
    (P_{AB} \mathbfcal{E}^{\otimes 2} P_{AB})^m = S^{-1}
    \begin{pmatrix}
        1 & 0 & 0 & 0  \cr
        0 & \lambda_1 & 1 & 0 \cr
        0 & 0 & \lambda_1 & 0 \cr
        0 & 0 & 0 & \lambda_2 \cr
    \end{pmatrix}^{\!\!m}
    S = S^{-1}
    \begin{pmatrix}
        1 & 0 & 0 & 0  \cr
        0 & \lambda_1^m & m \lambda_1^{m-1} & 0 \cr
        0 & 0 & \lambda_1^m & 0 \cr
        0 & 0 & 0 & \lambda_2^m \cr
    \end{pmatrix}
    S,
\end{equation}
and
\begin{equation}
    (P_{AB} \mathbfcal{E}^{\otimes 2} P_{AB})^m = S^{-1}
    \begin{pmatrix}
        1 & 0 & 0 & 0  \cr
        0 & \lambda_1 & 1 & 0 \cr
        0 & 0 & \lambda_1 & 1 \cr
        0 & 0 & 0 & \lambda_1 \cr
    \end{pmatrix}^{\!\!m}
    S = S^{-1}
    \begin{pmatrix}
        1 & 0 & 0 & 0  \cr
        0 & \lambda_1^m & \lambda_1^{m-1} & \frac{m(m-1)}{2} \lambda_1^{m-2} \cr
        0 & 0 & \lambda_1^m & \lambda_1^{m-1} \cr
        0 & 0 & 0 & \lambda_1^m \cr
    \end{pmatrix}
    S.
\end{equation}
Therefore, in this more general scenario the decay law behaviour of $(P_{AB} \mathbfcal{E}^{\otimes 2} P_{AB})^m$ is still described by the constants $\{\lambda_i\}$.

\subsection{Analysis of the $\cal{C}\times\cal{C}$ unitarity benchmarking protocol }\label{pab-protocol}

We now show that the unitarity benchmarking protocol detailed in Protocol \ref{protocol:CxC} generates the claimed decay law for the noise channel associated to the gate-set $\Gamma_{AB}$.
\begin{lemma}\label{projector2appears}
    Over all sequences $\textbf{s}$, and for a gate-independent noise channel $\E$, the expectation value of a observable $M$ squared can be written as:
    \begin{equation}
        \mathbb{E}_s [ m(s)^2] = \vbra{M}^{\otimes 2} (P_{AB} \mathbfcal{E}^{\otimes 2} P_{AB})^{k-1} \vket{\E(\rho)}^{\otimes 2}.
    \end{equation}
    with circuit of depth $k$, and sequences indexed via $s = (s_A, s_B)$ with $s_A = (a_1, a_2, \dots, a_k)$ and $s_B = (b_1, b_2, \dots, b_k)$ specifying the particular target unitary in each of the local gate-sets $\Gamma_{AB} = \Gamma_A \otimes \Gamma_B$.
    \end{lemma}
    \begin{proof}
    From eqn.(\ref{eqn:single-circuit}) over all sequences we have
    \begin{equation}
        \begin{split}
        \mathbb{E}_s [ m(s)^2] :=& \frac{1}{|\Gamma_{AB}|^k} \sum_{s} m(s)^2 = \frac{1}{|\Gamma_{AB}|^k} \sum_{s} (\tr [ M \tilde{\U}_s (\rho) ])^2 \\
        =& \frac{1}{|\Gamma_{AB}|^k} \sum_{s} \vbraket{M}{\tilde{\mathcal{U}}_{\textbf{s}} (\rho)}^{2} = \frac{1}{|\Gamma_{AB}|^k} \sum_{s} \vbra{M}\tilde{\mathbfcal{U}}_{\textbf{s}}\vket{\rho}^{2}, \\
        =& \frac{1}{|\Gamma_{AB}|^k} \sum_{s}  \vbra{M} \tilde{\mathbfcal{U}}_{s_k} \tilde{\mathbfcal{U}}_{s_{k-1}} ... \ \tilde{\mathbfcal{U}}_{s_1} \vket{\rho}^{2}, \\
        =& \frac{1}{|\Gamma_{AB}|^k} \sum_{s}  \vbra{M} (\mathbfcal{U}_{s_k} \mathbfcal{E}) (\mathbfcal{U}_{s_{k-1}} \mathbfcal{E}) \ ... \ (\mathbfcal{U}_{s_1} \mathbfcal{E}) \vket{\rho}^{2}. \\
        \end{split}
    \end{equation}
    Which we can write equivalently as a bipartite system,
    \begin{equation}
        \mathbb{E}_s [ m(s)^2] =  \frac{1}{|\Gamma_{AB}|^k} \sum_{s} \vbra{M}^{\otimes 2} (\mathbfcal{U}_{s_k}^{\otimes 2} \mathbfcal{E}^{\otimes 2}) (\mathbfcal{U}_{s_{k-1}}^{\otimes 2} \mathbfcal{E}^{\otimes 2}) \ ... \ (\mathbfcal{U}_{s_1}^{\otimes 2} \mathbfcal{E}^{\otimes 2}) \vket{\rho}^{\otimes 2}.
    \end{equation}
    The summation over $\U_s$ for each gate $k$, can be expanded
    \begin{equation}
        \mathbb{E}_s [ m(s)^2] = \vbra{M}^{\otimes 2} ( \frac{1}{\abs{\Gamma_{AB}}} \sum_{U_{s_k}\in \Gamma_{AB}} \mathbfcal{U}_{s_k}^{\otimes 2} \mathbfcal{E}^{\otimes 2}) ( \frac{1}{\abs{\Gamma_{AB}}} \sum_{U_{s_{k-1}}\in \Gamma_{AB}} \mathbfcal{U}_{s_{k-1}}^{\otimes 2} \mathbfcal{E}^{\otimes 2}) ... ( \frac{1}{\abs{\Gamma_{AB}}} \sum_{U_{s_1}\in \Gamma_{AB}} \mathbfcal{U}_{s_1}^{\otimes 2} \mathbfcal{E}^{\otimes 2}) \vket{\rho}^{\otimes 2},
    \end{equation}
    recalling $\cal{U}_s=\mathcal{U}_{s_A} \otimes \mathcal{U}_{s_B}$ and the sequences expand as $\sum_s^{k^2} = \sum_{s_A,s_B}^{k,k}$. We now have the form to use the property of the unitarity 2-design gate-set on each sub-system from eqns. (\ref{unit2design}) \& (\ref{eqn:twirl-each-subsystem}),
    \begin{equation}
        \mathbb{E}_s [ m(s)^2] = \vbra{M}^{\otimes 2} (\int \!\!d \mu_{\mbox{\tiny Haar}}(U_A)\ \int \!\!d \mu_{\mbox{\tiny Haar}}(U_B)\ (\mathbfcal{U}_A \otimes \mathbfcal{U}_B )^{\otimes 2} \mathbfcal{E}^{\otimes 2}) \ ... \ ( \int \!\!d \mu_{\mbox{\tiny Haar}}(U_A)\ \int \!\!d \mu_{\mbox{\tiny Haar}}(U_B)\ (\mathbfcal{U}_A \otimes \mathbfcal{U}_B )^{\otimes 2} \mathbfcal{E}^{\otimes 2}) \vket{\rho}^{\otimes 2},
    \end{equation}
    there are now $k$ identical integrals over $U_{A}$ and $U_{B}$. So we can write
    \begin{equation}
        \mathbb{E}_s [ m(s)^2] = \vbra{M}^{\otimes 2} (\int \!\!d \mu_{\mbox{\tiny Haar}}(U_A)\ \int \!\!d \mu_{\mbox{\tiny Haar}}(U_B)\ (\mathbfcal{U}_A \otimes \mathbfcal{U}_B )^{\otimes 2}  \mathbfcal{E}^{\otimes 2})^k \vket{\rho}^{\otimes 2}.
    \end{equation}
    This is just the projector $P_{AB}$, where $P_{AB} = P_A \otimes P_B$ up to reordering of subsystems, given by 
    \begin{equation}
        \begin{split}
            \mathbb{E}_s [ m(s)^2] &= \vbra{M}^{\otimes 2} (P_{AB} \mathbfcal{E}^{\otimes 2})^k \vket{\rho}^{\otimes 2}, \\
            &= \vbra{M}^{\otimes 2} (P_{AB} \mathbfcal{E}^{\otimes 2})^{k-1} (P_{AB} \mathbfcal{E}^{\otimes 2}) \vket{\rho}^{\otimes 2}, \\
            &= \vbra{M}^{\otimes 2} (P_{AB} \mathbfcal{E}^{\otimes 2})^{k-1} (P_{AB}) \vket{\E(\rho)}^{\otimes 2},
        \end{split}
    \end{equation}
    where we have absorbed the first noise channel to the initial state of the system $\rho$. As $P_{AB} = (P_{AB})^2$, we are free to write all the intermediate projectors twice
    \begin{equation}
        \begin{split}
            \mathbb{E}_s [ m(s)^2] &= \vbra{M}^{\otimes 2} (P_{AB} \mathbfcal{E}^{\otimes 2})(P_{AB}^2 \mathbfcal{E}^{\otimes 2})^{k-2}(P_{AB}) \vket{\E(\rho)}^{\otimes 2}, \\
            &= \vbra{M}^{\otimes 2} (P_{AB} \mathbfcal{E}^{\otimes 2} P_{AB})(P_{AB} \mathbfcal{E}^{\otimes 2} P_{AB})^{k-2} \vket{\E(\rho)}^{\otimes 2}, \\
            &= \vbra{M}^{\otimes 2} (P_{AB} \mathbfcal{E}^{\otimes 2} P_{AB})^{k-1} \vket{\E(\rho)}^{\otimes 2}.
        \end{split}
    \end{equation}
    Which completes the proof.
\end{proof}

From Section \ref{section:jordan-decomp-pab}  if the Jordan decomposition is diagonal
\begin{equation}
    (P_{AB} \mathbfcal{E}^{\otimes 2} P_{AB})^{k-1} = S^{-1} ( \ketbra{00} +  \lambda_1^{k-1} \ketbra{10} +  \lambda_2^{k-1} \ketbra{01} +  \lambda_3^{k-1} \ketbra{11} ) S,
\end{equation}
where $\lambda_i$ are the eigenvalues of the matrix $\S$.
Therefore from Lemma \ref{projector2appears} we can write
\begin{equation}
    \begin{split}
        \mathbb{E}_s [ m(s)^2] &= \vbra{M}^{\otimes 2} (P_{AB} \mathbfcal{E}^{\otimes 2} P_{AB})^{k-1} \vket{\E(\rho)}^{\otimes 2}, \\
        &= \vbra{M}^{\otimes 2} S^{-1} J^{k-1} S \vket{\E(\rho)}^{\otimes 2}, \\
        &= \vbra{M}^{\otimes 2} S^{-1} ( \ketbra{00} +  \lambda_1^{k-1} \ketbra{10} +  \lambda_2^{k-1} \ketbra{01} +  \lambda_3^{k-1} \ketbra{11} ) S \vket{\E(\rho)}^{\otimes 2}.
    \end{split}
\end{equation}
The transformation matrix $S$ can be absorbed into the initial state of the system and the final measurement
\begin{equation}
    \begin{split}
        \mathbb{E}_s [ m(s)^2]  &= \vbra{S^{-1 \dagger}(M^{\otimes 2})} J^{k-1} \vket{S(\E(\rho)^{\otimes 2})}, \\
        &=  \vbra{S^{-1 \dagger}(M^{\otimes 2})} ( \ketbra{00} +  \lambda_1^{k-1} \ketbra{10} +  \lambda_2^{k-1} \ketbra{01} +  \lambda_3^{k-1} \ketbra{11} ) \vket{S(\E(\rho)^{\otimes 2})},
    \end{split}
\end{equation}
and further expanded as
\begin{equation}
    \begin{split}
        \mathbb{E}_s [ m(s)^2]  &= \braket{\boldify{S^{-1 \dagger}(M^{\otimes 2})}}{00} \braket{00}{\boldify{S(\E(\rho)^{\otimes 2})}} \\
        & \ \ \ \ \ \ \ \ \ \  + \lambda_1^{k-1} \braket{\boldify{S^{-1 \dagger}(M^{\otimes 2})}}{10}\braket{10}{\boldify{S(\E(\rho)^{\otimes 2})}} \\
        & \ \ \ \ \ \ \ \ \ \ \ \ \ \ \ \ \ \ \ \ + \lambda_2^{k-1} \braket{\boldify{S^{-1 \dagger}(M^{\otimes 2})}}{01}\braket{01}{\boldify{S(\E(\rho)^{\otimes 2})}} \\
        & \ \ \ \ \ \ \ \ \ \ \ \ \ \ \ \ \ \ \ \ \ \ \ \  \ \ \ \ + \lambda_3^{k-1} \braket{\boldify{S^{-1 \dagger}(M^{\otimes 2})}}{11}\braket{11}{\boldify{S(\E(\rho)^{\otimes 2})}}.
    \end{split}
\end{equation}
Or simply,
\begin{equation}
    \mathbb{E}_s [ m(s)^2]  = c_{00} + c_{10} \ \lambda_1^{k-1} + c_{01} \ \lambda_2^{k-1} + c_{11} \ \lambda_3^{k-1}.
\end{equation}
So if a channel $\E$ produces a diagonal Jordan decomposition $J$, the protocol will produce a fit of this form where $\lambda_i$ are the eigenvalues of $\S$.

If the Jordan decomposition is not diagonal, then there are 2 remaining options. Firstly,
\begin{equation}
    J^{k-1} =
    \begin{pmatrix}
        1 & 0 & 0 & 0  \cr
        0 & \lambda_1^{k-1} & (k-1) \lambda_1^{k-2} & 0 \cr
        0 & 0 & \lambda_1^{k-1} & 0 \cr
        0 & 0 & 0 & \lambda_2^{k-1} \cr
    \end{pmatrix}
    ,
\end{equation}
where the fit will take the following form: $\mathbb{E}_s [ m(s)^2] = c_{0} + c_{1} \ \lambda_1^{k-1} + c_{2} \ \lambda_2^{k-1}$, where $\lambda_i$ are the degenerate eigenvalues of $\S$, and constants $c_i$ are dependent on $M, \rho, S, S^{-1} \& \ \bm{x}$.  Secondly, 
\begin{equation}
    J^{k-1} =
    \begin{pmatrix}
        1 & 0 & 0 & 0  \cr
        0 & \lambda_1^{k-1} & \lambda_1^{k-2} & \frac{(k-1)(k-2)}{2} \lambda_1^{k-3} \cr
        0 & 0 & \lambda_1^{k-1} & \lambda_1^{k-2} \cr
        0 & 0 & 0 & \lambda_1^{k-1} \cr
    \end{pmatrix}
    ,
\end{equation}
where the fit will take the following form: $ \mathbb{E}_s [ m(s)^2] = c_{0} + c_{1} \ \lambda_1^{k-1}$, where $\lambda_1$ is the degenerate eigenvalue of $\S$, and for different constants $c_i$ dependent on $M, \rho, S, S^{-1}$ and $ \bm{x}$.

\section{Estimating sub-unitarities via mid-circuit re-set protocols}\label{RBforsubunit}
The local subunitarities $u_{A \to A}(\E_{AB})$ and $u_{B \to B}(\E_{AB})$ of any bipartite channel $\E_{AB}$ are measures of interest in their own right. However the exact estimation of the subunitarity of gate noise through unitarity benchmarking requires the repeated preparation of the maximally mixed state on the ancillary subsystem. As shown in \cite{combes2017logical}, this introduces additional noise from the imperfect depolarization. 

In the main text, figures were given of simulations of the estimation of local subunitarities under the assumption that any error in the preparation of the maximally mixed state was purely local to the ancillary subsystem. What follows is a discussion of possible methods to extract estimates of local subunitarities under more physically realistic assumptions about the nature of induced re-set errors and the quantum device in question.

\subsection{Estimating local sub-unitarities with re-set errors}
The manner in which the induced error is modelled determines the accuracy of the predicted estimate of the subunitarity. If we model the noisy re-set channel $\tilde{\R}_B$ as
\begin{equation}
    \tilde{\R}_B = \E_P \circ (id_A \otimes \R_B) \circ \E_M,
\end{equation}
where $\R_B$ is the exact reset, and where $\E_M$ and $\E_P$ are SPAM errors on whole system related to the imperfect reset of the sub-system $B$. Then it can be shown that Protocol \ref{protocol:Cx1} allows the estimation of the subunitarity of the combined channel
\begin{equation}\label{fitting_fn_eqn_3}
    \mathbb{E}_{s_A} [ m(s_{A})^2] = c_1 + c_2 u_{A \to A}(\E_M \circ \E \circ \E_P)^{k-1}
\end{equation}
for a sequence of length $k$ where $\E$ is the noise channel associated to the gate-set. The constants $c_1$ \& $c_2$ depend on the initial and final SPAM and non-unitality of the channel $\E$.

The Protocol \ref{protocol:Cx1} requires the preparation of the maximally mixed state  ($\vket{Y_0}/\sqrt{d_B}$ in our notation) on subsystem $B$, albeit noisily. However, we can consider an alternative, where we randomly re-set to one of the computational basis states. For two qubits, we can consider the Liouville representation of the preparation channel
\begin{equation}
    \bm{prep}_{B,\pm Z} := \bm{id}_A \otimes (\vket{Y_0}/\sqrt{2} \pm \vket{Y_Z}/\sqrt{2}),
\end{equation}
which prepares the state $\ketbra{0} = \frac{1}{2}(\ident_B \pm Z)$ on sub-system $B$. For a bipartite channel $\E$, the related channel $\cal{E}_{+Z}$ on qubit $A$ is defined as
\begin{equation}
    \boldify{\cal{E}}_{+Z} := \bm{tr}_B \cdot \boldify{\cal{E}} \cdot \bm{prep}_{B,+Z} = (\bm{id}_A \otimes \vbra{Y_0}) \ \boldify{\cal{E}} \ ( \bm{id}_A \otimes (\vket{Y_0} + \vket{Y_Z})),
\end{equation}
and similarly $\boldify{\cal{E}}_{-Z} := \bm{tr}_B \cdot \boldify{\cal{E}} \cdot \bm{prep}_{B,-Z}$.

We can calculate the structure of the unitarity of these channels using the Liouville representation. The definition of unitarity can be written in our basis as
\begin{equation}
    u(\cal{E}_A) = \frac{1}{d^2 - 1} \sum_{ij} \vbra{X_i} \boldify{\cal{E}}^\dagger \vket{X_j} \vbra{X_j} \boldify{\cal{E}} \vket{X_i},
\end{equation}
for some channel $\cal{E}_A$ that maps $\cal{B}(\mathcal{H}_A) \to \cal{B}(\mathcal{H}_A)$. The unitarity of the channel $\cal{E}_{+Z}$ can then be related to the local sub-unitarity of the channel $\cal{E}$ as
\begin{equation}
    \begin{split}
        u(\cal{E}_{+Z}) = u_{A \to A}(\cal{E}) + \frac{1}{3} \sum_{ij} & \vbra{X_i \otimes Y_Z} \boldify{\cal{E}}^\dagger \vket{X_j \otimes Y_0} \vbra{X_j \otimes Y_0} \boldify{\cal{E}} \vket{X_i \otimes Y_Z} \\
        & \ \ \ \ \ + \vbra{X_i \otimes Y_Z} \boldify{\cal{E}}^\dagger \vket{X_j \otimes Y_0} \vbra{X_j \otimes Y_0} \boldify{\cal{E}} \vket{X_i \otimes Y_0} \\
        & \ \ \ \ \ \ \ \ \ \ + \vbra{X_i \otimes Y_0} \boldify{\cal{E}}^\dagger \vket{X_j \otimes Y_0} \vbra{X_j \otimes Y_0} \boldify{\cal{E}} \vket{X_i \otimes Y_Z},
    \end{split}
\end{equation}
and similarly
\begin{equation}
    \begin{split}
        u(\cal{E}_{-Z}) = u_{A \to A}(\cal{E}) + \frac{1}{3} \sum_{ij} & \vbra{X_i \otimes Y_Z} \boldify{\cal{E}}^\dagger \vket{X_j \otimes Y_0} \vbra{X_j \otimes Y_0} \boldify{\cal{E}} \vket{X_i \otimes Y_Z} \\
        & \ \ \ \ \ - \vbra{X_i \otimes Y_Z} \boldify{\cal{E}}^\dagger \vket{X_j \otimes Y_0} \vbra{X_j \otimes Y_0} \boldify{\cal{E}} \vket{X_i \otimes Y_0} \\
        & \ \ \ \ \ \ \ \ \ \ - \vbra{X_i \otimes Y_0} \boldify{\cal{E}}^\dagger \vket{X_j \otimes Y_0} \vbra{X_j \otimes Y_0} \boldify{\cal{E}} \vket{X_i \otimes Y_Z}.
    \end{split}
\end{equation}
This follows from expansion of the definitions of the channels and the Liouville definition of unitarity. By taking the mean of the unitarity of these two channels we find
\begin{equation}\label{extraterm}
    \frac{1}{2}( u(\cal{E}_{+Z}) + u(\cal{E}_{-Z})) = u_{A \to A}(\cal{E}) + \frac{1}{3} \sum_{ij} \vbra{X_i \otimes Y_Z} \boldify{\cal{E}}^\dagger \vket{X_j \otimes Y_0} \vbra{X_j \otimes Y_0} \boldify{\cal{E}} \vket{X_i \otimes Y_Z}.
\end{equation}
As the $2^{\nd}$ term in eqn. (\ref{extraterm}) is strictly non-negative we can use this measure to bound the sub-unitarity of the target channel. Therefore, if we can better re-set to one of the computational basis state, then we can estimate the $A\rightarrow A$ sub-unitarity via the following. 
\begin{lemma}\label{local-unit-inequal}
    The local sub-unitarity of a bipartite channel $u_{A \to A}(\cal{E})$ can be bounded by the average unitarity of the channel with two specific initial conditions,
    \begin{equation}
        u_{A \to A}(\cal{E}) \leq \frac{1}{2}( u(\cal{E}_{+Z}) + u(\cal{E}_{-Z}))
    \end{equation}
    where $\cal{E}_{+Z}(\rho)=\tr_B[\cal{E}(\rho \otimes \ketbra{0})]$ and $\cal{E}_{-Z}(\rho)=\tr_B[\cal{E}(\rho \otimes \ketbra{1})]$.
\end{lemma}
\begin{proof}
    This follows from Lemma E1. and the non-negativity of any element $\boldify{\cal{E}}_{ij}^* \boldify{\cal{E}}_{ij}$.
\end{proof}

\begin{corollary}
    If $\cal{E}$ is a product channel.
    \begin{equation}
        u_{A \to A}(\cal{E}) = \frac{1}{2}( u(\cal{E}_{+Z}) + u(\cal{E}_{-Z}))
    \end{equation}
\end{corollary}

\begin{proof}
    If $\cal{E}=\cal{E}_A \otimes \cal{E}_B$, the $2^{\nd}$ term always contains the element $\vbra{Y_0}\boldify{\cal{E}_B}\vket{Y_i}$, which must be zero for a valid CPTP map.
\end{proof}

Additionally, it can be shown that Lemma \ref{local-unit-inequal} holds for any two orthogonal initial states on qubit $B$.

\begin{corollary}\label{wider-theorem}
    The local sub-unitarity of a bipartite channel $u_{A \to A}(\cal{E})$ can be bounded by the average unitarity of the channel with two specific initial conditions,
    \begin{equation}
        u_{A \to A}(\cal{E}) \leq \min [ \ \frac{1}{2}( u(\cal{E}_{+\bm{b}}) + u(\cal{E}_{-\bm{b}})) \ ],
    \end{equation}
    where $\cal{E}_{\pm \bm{b}}(\rho)=\tr_B[\cal{E}(\rho \otimes \frac{1}{2}(\ident_B \pm \bm{b} \cdot \bm{\sigma}) )]$.
\end{corollary}
\begin{proof}
    This follows from Lemma \ref{local-unit-inequal}, replacing $Z$ with a general Bloch vector on qubit $B$.
\end{proof}

Under the assumption that computational basis states induce fewer errors when prepared compared to the maximally mixed state, then estimating $u(\E_{+Z})$ and $u(\E_{-Z})$ with a RB protocol allows an upper bound to be placed on the local sub-unitarity $u_{A \to A}(\E)$, where $\E$ is the noisy channel associated with the target gate-set.

In such a case, the RB protocol would simply entail two experiments: firstly performing unitarity RB on qubit $A$ with a reset of qubit $B$ to $\ket{0}$, and then secondly with a reset to $\ket{1}$. If we assume the reset is performed completely incoherently, but with bipartite SPAM errors we have for the $1^{\st}$ experiment will produce a fit of the form
\begin{equation}
    \mathbb{E}_{s_A} [ m(s_{A})^2] = c_1 + c_2 u(\E_{+Z,M} \circ \E_{+Z} \circ \E_{+Z,P})^{k-1},
\end{equation}
where $\Lambda_{+Z,M}$ \& $\Lambda_{+Z,P}$ are the bipartite SPAM errors associated with the noisy reset of qubit $B$ to $\ket{0}$. Similarly the $2^{\nd}$ experiment will produce a fit of the form
\begin{equation}
    \mathbb{E}_{s_A} [ m(s_{A})^2] = c_1 + c_2 \ u(\E_{-Z,M} \circ \E_{-Z} \circ \E_{-Z,P})^{m-1},
\end{equation}
where $\E_{\pm Z,M}$ \& $\E_{\pm Z,P}$ are the bipartite SPAM errors associated with the noisy reset of qubit $B$.
Such a modification could then be used when the preparation of a maximally mixed state is significantly noisier compared to computational basis state preparation and reset which would detrimentally affect estimation of $u_{A\to A}(\E)$.  In the case when $\E_{\pm Z, M,P} \approx id$ an upper bound could be estimated as shown above.

\section{Unitarity bounds on the diamond norm}\label{append:unitarity-diamond-bound}
The unitarity provides improved bounds on the diamond norm compared to infidelity, while still being efficiently estimatable in a SPAM robust manner. From (eqn. (32), \cite{wallman2015bounding}) we have:
\begin{equation}
	\frac{K}{\sqrt{2}} \leq \frac{1}{2} || id - \mathcal{E}||_\diamond \leq \sqrt{ \frac{d^3 K^2}{4} + \frac{(d+1)^2 r(\mathcal{E})^2}{2} }
\end{equation}
where $K^2 = \frac{d^2 -1}{d^2} (u(\mathcal{E}) + \frac{2 d }{d-1} r(\mathcal{E}) - 1)$. Therefore when the channel $\mathcal{E}$ is unitary ($u(\mathcal{E})=1$) both bounds scale as $\mathcal{O}(\sqrt{r(\mathcal{E})})$. For purely a purely stochastic channel, where the unitarity is directly related to the infidelity, the bounds scale as $\mathcal{O}({r(\mathcal{E})})$, thereby tightening the bound of eqn. (1).

\clearpage
\small
\bibliography{references.bib}

\end{document}